\documentclass[11pt]{article}
\usepackage{amsmath,amssymb,amsfonts,amsthm,mathtools,bm,bbm,dsfont,enumitem} 
\mathtoolsset{showonlyrefs}
\usepackage{setspace}
\usepackage{graphicx, float, booktabs, threeparttable, multirow, makecell,comment}
\usepackage{caption}
\usepackage{subcaption}
\usepackage[colorlinks=true,linkcolor=blue,citecolor=blue,urlcolor=blue, hypertexnames=true]{hyperref}

\usepackage{setspace}
\usepackage[margin=1.5in]{geometry}

\usepackage[plain,noend,ruled]{algorithm2e}
\makeatletter
\renewcommand{\algocf@captiontext}[2]{#1\algocf@typo. \AlCapFnt{}#2} 
 
\def\@algocf@capt@plain{top}
\renewcommand{\algocf@makecaption}[2]{
  \addtolength{\hsize}{\algomargin}
  \sbox\@tempboxa{\algocf@captiontext{#1}{#2}}
  \ifdim\wd\@tempboxa >\hsize
    \hskip .5\algomargin
    \parbox[t]{\hsize}{\algocf@captiontext{#1}{#2}}
  \else
    \global\@minipagefalse
    \hbox to\hsize{\box\@tempboxa}
  \fi
  \addtolength{\hsize}{-\algomargin}
}
\makeatother

\DeclareMathOperator*{\argmin}{argmin}

\newtheorem{theorem}{Theorem}
\newtheorem{lemma}{Lemma}

\newtheorem{proposition}{Proposition}
\theoremstyle{definition}
\newtheorem{remark}{Remark}

\def\bbR{\mathbb{R}}
\def\bbE{\mathbb{E}}
\def\bbS{\mathbb{S}}
\def\BVS{BVS-SS}
\def\Ynew{\tilde{Y}}
\def\Xnew{\tilde{X}}

\def\d{\mathrm{d}}
\def\P{\mathbb{P}}
\def\E{\mathbb{E}}
\def\cB{\mathcal{B}}
\def\cN{\mathcal{N}}
\def\lmin{\underline{\lambda}}
\def\lmax{\overline{\lambda}}

\newcommand{\ind}{\mathbbm{1}}

\PassOptionsToPackage{unicode}{hyperref}
\PassOptionsToPackage{hyphens}{url}
\PassOptionsToPackage{dvipsnames,svgnames,x11names}{xcolor}

\usepackage{iftex}
\ifPDFTeX
  \usepackage[T1]{fontenc}
  \usepackage[utf8]{inputenc}
  \usepackage{textcomp} 
\else 
  \usepackage{unicode-math}
  \defaultfontfeatures{Scale=MatchLowercase}
  \defaultfontfeatures[\rmfamily]{Ligatures=TeX,Scale=1}
\fi
\usepackage{lmodern}
\ifPDFTeX\else  
\fi

\IfFileExists{upquote.sty}{\usepackage{upquote}}{}
\IfFileExists{microtype.sty}{
  \usepackage[]{microtype}
  \UseMicrotypeSet[protrusion]{basicmath} 
}{}
\makeatletter
\@ifundefined{KOMAClassName}{
  \IfFileExists{parskip.sty}{
    \usepackage{parskip}
  }{
    \setlength{\parindent}{0pt}
    \setlength{\parskip}{6pt plus 2pt minus 1pt}}
}{
  \KOMAoptions{parskip=half}}
\makeatother
\usepackage{xcolor}
\makeatletter
\ifx\paragraph\undefined\else
  \let\oldparagraph\paragraph
  \renewcommand{\paragraph}{
    \@ifstar
      \xxxParagraphStar
      \xxxParagraphNoStar
  }
  \newcommand{\xxxParagraphStar}[1]{\oldparagraph*{#1}\mbox{}}
  \newcommand{\xxxParagraphNoStar}[1]{\oldparagraph{#1}\mbox{}}
\fi
\ifx\subparagraph\undefined\else
  \let\oldsubparagraph\subparagraph
  \renewcommand{\subparagraph}{
    \@ifstar
      \xxxSubParagraphStar
      \xxxSubParagraphNoStar
  }
  \newcommand{\xxxSubParagraphStar}[1]{\oldsubparagraph*{#1}\mbox{}}
  \newcommand{\xxxSubParagraphNoStar}[1]{\oldsubparagraph{#1}\mbox{}}
\fi
\makeatother

\usepackage{longtable,booktabs,array}
\usepackage{calc} 

\usepackage{etoolbox}
\makeatletter
\patchcmd\longtable{\par}{\if@noskipsec\mbox{}\fi\par}{}{}
\makeatother

\IfFileExists{footnotehyper.sty}{\usepackage{footnotehyper}}{\usepackage{footnote}}
\makesavenoteenv{longtable}
\usepackage{graphicx}
\makeatletter
\def\maxwidth{\ifdim\Gin@nat@width>\linewidth\linewidth\else\Gin@nat@width\fi}
\def\maxheight{\ifdim\Gin@nat@height>\textheight\textheight\else\Gin@nat@height\fi}
\makeatother
\setkeys{Gin}{width=\maxwidth,height=\maxheight,keepaspectratio}
\makeatletter
\def\fps@figure{htbp}
\makeatother

\makeatletter
\@ifpackageloaded{caption}{}{\usepackage{caption}}
\AtBeginDocument{
\ifdefined\contentsname
  \renewcommand*\contentsname{Table of contents}
\else
  \newcommand\contentsname{Table of contents}
\fi
\ifdefined\listfigurename
  \renewcommand*\listfigurename{List of Figures}
\else
  \newcommand\listfigurename{List of Figures}
\fi
\ifdefined\listtablename
  \renewcommand*\listtablename{List of Tables}
\else
  \newcommand\listtablename{List of Tables}
\fi
\ifdefined\figurename
  \renewcommand*\figurename{Figure}
\else
  \newcommand\figurename{Figure}
\fi
\ifdefined\tablename
  \renewcommand*\tablename{Table}
\else
  \newcommand\tablename{Table}
\fi
}
\@ifpackageloaded{float}{}{\usepackage{float}}
\floatstyle{ruled}
\@ifundefined{c@chapter}{\newfloat{codelisting}{h}{lop}}{\newfloat{codelisting}{h}{lop}[chapter]}
\floatname{codelisting}{Listing}

\makeatother
\makeatletter
\@ifpackageloaded{caption}{}{\usepackage{caption}}
\@ifpackageloaded{subcaption}{}{\usepackage{subcaption}}
\makeatother

\ifLuaTeX
  \usepackage{selnolig}  
\fi
\usepackage[]{natbib}
\usepackage{bookmark}

\IfFileExists{xurl.sty}{\usepackage{xurl}}{} 
\hypersetup{
  pdftitle={Title},
  pdfauthor={Author 1; Author 2},
  pdfkeywords={3 to 6 keywords, that do not appear in the title},
  colorlinks=true,
  linkcolor={blue},
  filecolor={Maroon},
  citecolor={Blue},
  urlcolor={Blue},
  pdfcreator={LaTeX via pandoc}}

\begin{document}

\def\spacingset#1{\renewcommand{\baselinestretch}
{#1}\small\normalsize} \spacingset{1}

\title{\bf Bayesian Synthetic Control with a Soft Simplex Constraint} %\thanks{ Authors are listed alphabetically. Y.X. is grateful to Qi Li, Yonghong An and Angda Li for their helpful comments and suggestions. Q.Z. acknowledges support in part from the NSF grants DMS-2311307 and DMS-2245591.}}
  \author{Yihong Xu\footnote{Department of Economics, Texas A\&M University} \hspace{.2cm}   
    and 
    Quan Zhou\footnote{Department of Statistics, Texas A\&M University} \\
    }
    \date{}
  \maketitle

 \spacingset{1.2}

\begin{abstract}
The challenges posed by high-dimensional data and use of the simplex constraint are two major concerns in the empirical application of the synthetic control method (SCM) in econometric studies. To address both issues simultaneously, we propose a Bayesian SCM that integrates a soft simplex constraint within spike-and-slab variable selection. 
The hierarchical prior structure captures the extent to which the data supports the simplex constraint, allowing for more efficient and data-adaptive counterfactual estimation. 
The intractable marginal likelihood induced by the soft simplex constraint presents a major computational challenge, which we resolve by developing a novel  Metropolis-within-Gibbs  algorithm that updates the regression coefficients of two predictors simultaneously.  
Our main theoretical contribution is a high-dimensional  selection consistency result for the spike-and-slab variable selection under the simplex constraint, which significantly extends the current theory for high-dimensional Bayesian variable selection. 
Simulation studies demonstrate that our method performs well across diverse settings. To illustrate its practical values, we apply it to two empirical examples for estimating the effect of economic policies. 
\end{abstract}

\noindent
{\it Keywords:} average treatment effect; 
constrained linear regression;  
Metropolis-within-Gibbs sampler; 
spike-and-slab prior;
strong selection consistency.

\section{Introduction}\label{sec:intro}

\subsection{Background on Synthetic Control Methods} \label{sec:intro-background}

In many fields of social science, researchers need to identify the causal effect of an intervention or treatment of interest using only observational data, since conducting case-control experiments is either infeasible or improper. 
The most commonly used causal inference method in economics is the Difference-in-Differences 
approach. Its validity, however, relies heavily on the parallel trends assumption: the control unit must be sufficiently similar to the treated unit so that it can mimic the pre-intervention dynamics of the treated.  In many empirical settings, this assumption is overly restrictive, especially when the outcome of interest is measured at an aggregate level, such as a country's GDP. The synthetic control method (SCM) was developed to overcome this limitation, which constructs a counterfactual outcome of the  treated unit using a weighted average of the control units that closely tracks the pre-treatment trajectory of the treated.  For example, in the seminal work of \cite{abadie2003economic}, SCM was developed to study how terrorist conflicts affected the local economy in the Basque Country between the 1970s and 1990s.

Formally, SCM can be described as follows. 
Consider a panel with \( T = T_0 + T_1 \) total observations over time for \( N+1 \) units, where only the first unit receives the treatment after period $T_0$. 
The remaining \( N \) units form the control group, which are unaffected by the intervention. 
Let \( Z_{i, t} \) denote the outcome of interest for unit \( i \) at time \( t \), where \( i = 1, 2, ..., N+1 \) and \( t = 1, 2, ..., T \). 
The objective of SCM is to estimate the average treatment effect on the treated (ATT) by estimating the counterfactual outcome \( Z_{1, t}^{(0)} \) for the treated unit at times $t = T_0 + 1, \dots, T$, which represents what the outcome of unit 1 would have been in the absence of the treatment. 
The standard approach is to define \( Z_{1, t}^{(0)} \) as a weighted combination of the outcomes of the \( N \) untreated units: 
\begin{equation}\label{eq:SC}
     Z_{1, t}^{(0)}(w) = \sum_{i=2}^{N+1} w_{i-1} Z_{i, t}, 
\end{equation} 
where the weights $(w_i)_{i=1}^{N}$ are time-independent, non-negative and sum to one. A larger  \( w_i \) indicates that the $i$-th control unit better approximates the treated unit's pre-intervention characteristics. 
To find the optimal set of weights $(w_i)_{i=1}^{N}$ such that the synthetic control approximates the pre-treatment outcome of the treated unit as closely as possible, 
a common strategy is to minimize the mean squared error between the treated unit and the synthetic control in the pre-treatment period \(t = 1, \dots, T_0 \): 
\begin{align}
     \hat{w} &= \argmin_{w \in \Delta^{N-1}} \frac{1}{T_0} \sum_{t=1}^{T_0} \left(Z_{1, t} - \sum_{i=2}^{N+1} w_{i - 1} Z_{i, t} \right)^2,  \label{SCMQP}  \\ 
\text{ where }    \Delta^{N - 1} &= \left\{ u \in \mathbb{R}^N \colon u_i \geq 0 \text{ for each } i, \text{ and } \sum\nolimits_{i=1}^N u_i = 1   \right\}.  \label{eq:simplex} 
\end{align} 
For ease of presentation, we do not consider time-invariant characteristics in the above formulation, and readers are referred to \cite{abadie2021syntheticcontrol} for a more general description of SCM. 
The average difference between the observed outcome \( Z_{1, t} \) and the synthetic control \( Z_{1, t}^{(0)}(\hat{w}) \) serves as an estimate of the ATT: 
\begin{align}
    \widehat{\mathrm{ATT}} = \frac{1}{T_1} \sum_{t=T_0+1}^T \left( Z_{1, t} - \sum_{i=2}^{N+1} \hat{w}_{i-1} Z_{i,t} \right). 
\end{align}

Since its introduction, SCM has gained widespread adoption in empirical research, and \citet{athey2017state} described SCM as ``arguably the most important innovation in the policy evaluation literature in the last 15 years.'' However, the growing application of SCM has also revealed new challenges. The first is the \textit{high dimensionality}: when the number of the control units is comparable to, or even larger than, the number of pre-treatment periods, conventional  least-squares-based methods are no longer applicable. 
To address this issue, researchers have proposed alternatives for computing the weights in~\eqref{eq:SC},  such as penalized regression methods \citep{doudchenko2016balancing, carvalho2018arco}, factor-model-based approaches \citep{xu_2017, pang2022bayesian}, or variable selection \citep{kim2020bayesian, shi2023forward}. 
However, none of these methods utilizes the simplex constraint and thus the resulting estimator $\hat{w}$ loses its interpretation as weights.  How to implement SCM in a high-dimensional setting while preserving the simplex constraint remains largely unclear.

This naturally leads to another major issue with SCM: whether   \textit{the simplex constraint} should be imposed. Proponents argue that the simplex constraint prevents extrapolation~\citep{king2006dangers, abadie2021syntheticcontrol}. 
Critics, however, argue that the simplex constraint makes the extension of SCM to high-dimensional settings complicated and is often too restrictive in practice. A well-known example is the study of the economic impact of Hong Kong's reunification with China in 1997~\citep{hsiao2012panel}. 
Since no other Asian regions faced a similar ``handover'' situation as Hong Kong, it is natural to apply SCM to estimate the impact. However, since Hong Kong was long known for its outstanding economic performance at that time (known as the leading region in  the ``Four Asian Tigers''),  researchers found it unreasonable to use a convex combination of other regions to track the economic trajectory of Hong Kong. Allowing the sum of the coefficients to exceed one is in fact consistent with Hong Kong’s historically strong economic performance.

In this work, we develop a new Bayesian method for building the synthetic control  in high-dimensional settings.  Bayesian SCM has seen growing use in the recent  literature for its flexibility in integrating with high-dimensional techniques. 
In variable selection context, Bayesian approach enables statistical inference by model averaging \citep{kass1995bayes}, which contrasts with penalized regression methods like Lasso that focus on identifying a single best model. 
Further,  Bayesian SCM provides  posterior distributions for both ATT and the counterfactual outcome, which are more informative than point estimates. 
A potential limitation, however, is that the posterior fitting often requires the use of Markov chain Monte Carlo (MCMC) sampling, which can be very time-consuming compared to optimization-based methods.  
One recent work by~\cite{goh2022synthetic} considered Bayesian linear regression where the prior distribution of the regression coefficients is uniform over a convex hull, corresponding to the simplex constraint proposed by~\cite{abadie2003economic}. They also discussed the Bayesian formulation of the method of~\cite{li2020statistical}, which involves using a prior over the shifted conical hull.   
A similar model was studied in~\cite{martinez2022bayesian}, and a Bernstein-von Mises result for the Bayes estimator was proved. 
\cite{kim2020bayesian} proposed using either the horseshoe prior or the spike-and-slab prior, two well-known Bayesian approaches to variable selection, to construct synthetic control with selected control units in high-dimensional settings; the simplex constraint was dropped in their method. It was demonstrated in their paper that their method outperforms those frequentist SCM methods based on variable selection.
Nevertheless, to our knowledge, there is no existing Bayesian variable selection method for constructing synthetic control that also incorporates the simplex constraint.

\subsection{Overview of this Work} \label{sec:intro-overview}  
Motivated by the high-dimensional challenge and the mixed benefits and limitations of the simplex constraint, we propose a new SCM based on a  Bayesian model called \BVS{} (Bayesian Variable Selection with Soft Simplex constraint), which captures the sparsity among control units and allows the data to determine whether the simplex constraint is appropriate on a case-by-case basis. 
\BVS{} aims to leverage the strengths of both simplex-constrained SCM and unconstrained regression-based methods, which  serves as an effective solution to the debate surrounding the use of this constraint. There is no similar method in the frequentist SCM literature to our knowledge.  

In Section \ref{sec:bvs} we formally  introduce the \BVS{} model, which utilizes a standard spike-and-slab-type prior for selecting which control units should be used for constructing the synthetic control, a technique that has gained significant success and widespread popularity in the Bayesian community~\citep{george1993variable, george1997approaches, brown1998multivariate}. 
Different from existing models, \BVS{} assumes that there is an underlying sparse ``mean weight vector'' $\mu$ that satisfies the simplex constraint, but the actual weights $w_1, \dots, w_N$ involved in~\eqref{eq:SC} may deviate from $\mu$. Next, we introduce a key variance parameter $\tau$ quantifying the amount of this deviation. 
By imposing a non-informative prior on $\tau$, we let \BVS{} learn its value from the data, revealing whether the data actually agrees with the simplex assumption. 
A similar model was utilized in the theoretical analysis of~\cite{martinez2022bayesian}, but they did not consider variable selection and the implementation of the model. 

To generate samples from the posterior distribution of \BVS{}, in Section~\ref{sec:alg-short} we propose an original Metropolis-within-Gibbs sampler, where $\mu$ and error variance are updated by Gibbs schemes and $\tau$ is updated by  Metropolis-Hastings steps.  
Our Gibbs-type updating for $\mu$ is different from existing Gibbs schemes for similar problems, as those methods typically only update one coordinate at a time, which is not feasible in our context due to the simplex constraint (i.e., given $\mu_{-j}$, we have $\mu_j = 1 - \sum_{i \neq j} \mu_i$ almost surely). We  overcome  this challenge by theoretically finding the joint conditional posterior distribution of two coordinates $(\mu_i, \mu_j)$ given all the other parameters, which enables us to update two coordinates simultaneously. 
This full conditional posterior distribution is complicated due to the spike-and-slab prior for variable selection,  but  sampling from it can be performed straightforwardly and efficiently.  
Using these MCMC samples, we can compute the posterior mean estimates and credible intervals of the counterfactual and ATT   without extra computational cost. 

In Section~\ref{sec:theory}, we 
theoretically investigate the effect of the simplex constraint and establish high-dimensional consistency results for \BVS{} with $\tau$ fixed to zero (i.e., regression coefficients are forced to satisfy the simplex constraint). 
In Theorem~\ref{thm:post-cons}, we show that the posterior inclusion probability of the true model converges to $1$ in probability under mild assumptions on the true data-generating process (DGP). While similar results have been obtained for various Bayesian variable selection models~\citep{yang2016computational}, the presence of the simplex constraint introduces a substantial new challenge for the analysis: the marginal likelihood for each model in \BVS{} does not admit a closed-form expression and instead  involves an intractable integral over the simplex. 
To overcome this difficulty, we leverage techniques from constrained least-squares estimation and Gaussian concentration inequalities to derive closed-form upper and lower bounds for the marginal likelihood. The signal size required for detection has roughly the same order as that in the unconstrained case~\citep{yang2016computational}; see Assumption~\ref{asp:min-thershold}. 
To our knowledge, this is the first consistency result for  high-dimensional Bayesian variable selection under geometric constraints on the regression coefficients. 
Moreover, building on this result, we establish the order of the posterior expected predictive loss in Theorem~\ref{thm:mu-cons}, providing theoretical guarantees on the ATT estimation. 
Our results complement the existing theoretical developments for frequentist SCM~\citep{abadie2010synthetic, firpo2018synthetic, carvalho2018arco, li2020statistical, chernozhukov2021exact}.

Section~\ref{sec:sim} presents  simulation studies  which show that \BVS{} delivers robust and better performance than conventional methods, regardless of whether the simplex constraint is satisfied by the true DGP. Moreover, the posterior estimate of $\tau$ successfully captures the extent to which the true signals deviate from the simplex constraint. 
Section~\ref{sec:emp} presents two real-data applications of ATT estimation with \BVS{}, one from~\cite{carvalho2018arco} on the effectiveness of an anti-tax evasion program and the other from~\cite{shi2023forward} on the impact of an anti-corruption policy. Section~\ref{sec:conclusion} concludes the paper with a brief discussion. Details of the sampling algorithm, all proofs and additional numerical results are deferred to the supplementary materials.

\section{The BVS-SS Model}\label{sec:bvs}
\subsection{Model and Prior}\label{sec:bvs-model}
For generality and ease of notation, henceforth we use $Y \in \bbR^M$ to denote a response vector of $M$ observations and $X \in \bbR^{M \times N}$ to denote an arbitrary design matrix with $N$ explanatory variables.   
Assume that 
\begin{equation} \label{eq:bmodel}
    Y = X w + \epsilon, \quad \epsilon \sim N(0, \phi^{-1} I ), 
\end{equation}
where $w \in \bbR^N$ is the vector of unknown regression coefficients, $\epsilon$ is the vector of i.i.d. Gaussian noise with variance $\phi^{-1}$, and $I$ denotes the identity matrix. 
For the synthetic control problem described in Section~\ref{sec:intro-background}, we have $M = T_0$,  $Y_t = Z_{1, t}$ and $X_{t j} = Z_{j + 1, t}$  for $t = 1, \dots, T_0$ and $j = 1, \dots, N$.   

We propose a hierarchical prior distribution on $w$ which induces sparsity and incorporates the simplex constraint. To describe it, introduce an indicator vector $\gamma \in \{0, 1\}^N$ such that $\gamma_i = 1$ if and only if $w_i \neq 0$. 
Let $|\gamma| = \sum_{j = 1}^N \gamma_j$ denote the number of selected predictors.  
Given a vector $w$, we use $w_\gamma$ to denote the subvector of $w$ with entries indexed by $\{i \colon \gamma_i = 1\}$. 
We consider the following prior on $(w, \phi)$: 
\begin{equation}
    \label{eq:modelprior}
\begin{aligned} 
    \phi \sim \;& \mathrm{Gamma}(\kappa_1/2, \kappa_2 / 2), \\
    \gamma_i   \overset{\mathrm{i.i.d.}}{\sim} \;& \mathrm{Bernoulli}(\theta), \\  
     \tau   \sim  \;&  \mathrm{Gamma}(a_1, a_2), \\ 
    \mu_\gamma \mid \gamma \sim \;& \text{sym-Dirichlet}(\alpha), \\ 
     w_\gamma \mid \gamma, \mu_\gamma, \tau, \phi \sim \;& N( \mu_\gamma, \, (\tau / \phi) I ), 
\end{aligned}
\end{equation}
and when $\gamma_i = 0$, we let $\mu_i = w_i = 0$ with probability one. In~\eqref{eq:modelprior}, $\kappa_1, \kappa_2, a_1, a_2, \alpha > 0$ and $\theta \in (0, 1)$ are fixed hyperparameters, 
$\text{sym-Dirichlet}(\alpha)$ denotes the symmetric Dirichlet distribution on the simplex with concentration parameter $\alpha$ (the dimension of the simplex is clear from context), and the Gamma distribution is in shape-rate parameterization.  When $\alpha = 1$ and $|\gamma| = \ell$, the prior of $\mu_\gamma$ is a uniform distribution on $ \Delta^{|\gamma| - 1}$ with density $p(\mu_\gamma) =  \Gamma(\ell) $ for each $\mu_\gamma \in \Delta^{|\gamma| - 1}$, where $\Gamma(\ell) = (\ell - 1)!$ is the inverse volume of the simplex. 

In plain words, we assume that given $\gamma$ and $\phi$, $w_\gamma$ follows a normal distribution with mean vector $\mu_\gamma \in \Delta^{|\gamma| - 1}$ and prior variance depending on the parameter $\tau$. 
This conditional prior for $w$ can be viewed as an interpolation between the simplex constraint and the unconstrained setting. 
If the data suggests that $w$ is likely to satisfy the simplex constraint, then the posterior distribution of $\tau$ should concentrate around zero, while if the simplex constraint is significantly violated, $\tau$ should stay away from zero in the posterior. 
As will be demonstrated in our simulation study, $\tau$ can be reasonably estimated even with a relatively small sample size, which tells the user if the simplex constraint is appropriate for the data set being analyzed. 
This approach effectively reconciles the different views on the use of the simplex constraint, and thus we call our model \BVS{}  (Bayesian Variable Selection with Soft Simplex constraint).  
We note that our prior for $w$ is similar to that  studied in~\cite{martinez2022bayesian}; however, they did not consider variable selection, and in their implementation they assumed that $w$ satisfies the hard simplex constraint (though the soft constraint was used in their theory). 
Regarding the hyperparameters, $a_1, a_2$ should be chosen to reflect one's prior belief on whether simplex constraint is likely to hold, and a larger value of $\alpha$ encourages a more even distribution of the weights $\{ \mu_i \colon\gamma_i = 1\}$. 
In practice, one often has limited information about the true weight vector $w$, but the effect of these hyperparameters quickly becomes negligible as sample size increases. 

The other elements of our model are standard. The prior for $\phi$ is conjugate, and one can choose  small values for $\kappa_1, \kappa_2$ as a noninformative prior. 
The hyperparameter $\theta$ is the marginal prior probability of an explanatory variable being included in the regression model, and a smaller value of $\theta$ represents a heavier penalty on the model complexity.

\subsection{Posterior Distribution and ATT Estimation} \label{sec:bvs-post}
Let $X_\gamma$ denote the submatrix of $X$ with columns indexed by $\{i \colon \gamma_i = 1\}$. 
Under the prior distribution of \BVS{},  we have 
\begin{equation}\label{eq:cond-Y}
    Y \mid  \gamma, \mu_\gamma, \tau,  \phi \sim N \left( X_\gamma \mu_\gamma, \,  \phi^{-1} (\tau X_\gamma X_\gamma^\top  + I )   \right).
\end{equation}
A routine calculation using Woodbury identity yields the marginal likelihood of $(\mu, \tau,  \phi)$: 
\begin{align}
    p(y \mid \gamma, \mu_\gamma, \tau,  \phi) &\propto  \frac{\phi^{ M/2}}{
    \tau^{ |\gamma| / 2} \mathrm{det} (V_{\gamma, \tau})^{1/2}
    }  
    \exp\left\{ -\frac{\phi}{2}   (y - X_{\gamma } \mu_{\gamma })^\top \Sigma_{\gamma, \tau} (y - X_{\gamma } \mu_{\gamma })    \right \}, \label{eq:like1}  \\
    \text{ where }      V_{\gamma, \tau} &=  X_\gamma^\top X_\gamma + \tau^{-1} I, \quad 
    \Sigma_{\gamma, \tau} =   I - X_{\gamma  } V_{\gamma, \tau}^{-1} X_{\gamma }^\top. \label{eq:Sigma}
\end{align}
Further, using the conjugacy of normal-gamma prior, we find that the  conditional posterior distributions of $w_\gamma$ and $\phi$ are given by 
\begin{align}
    \phi \mid y, \gamma, \mu_\gamma, \tau  \sim \;& \mathrm{Gamma} \left( \frac{M + \kappa_1}{2}, \, \frac{\kappa_2 +   
     (y - X_{\gamma } \mu_{\gamma })^\top \Sigma_{\gamma, \tau} (y - X_{\gamma } \mu_{\gamma }) }{2} \right), \label{eq:full-cond-phi} \\
    w_\gamma \mid y, \gamma, \mu_\gamma, \tau,  \phi  \sim \;& N\left(  V_{\gamma, \tau}^{-1}( X^\top_\gamma y + \tau^{-1} \mu_\gamma ) , \, \phi^{-1} V_{\gamma, \tau}^{-1} \right).  \label{eq:full-cond-w}
\end{align}
Using the prior independence between $\mu, \tau, \phi$, we can express the joint posterior distribution by
$p(\gamma, \mu_\gamma, \tau, \phi \mid y) \propto p(y \mid \gamma, \mu_\gamma, \tau,  \phi)p(\mu_\gamma \mid \gamma) p(\gamma) p(\tau) p(\phi).$

Prediction with \BVS{} can be conducted as follows. 
Let $\Xnew \in \bbR^{\tilde{M} \times N}$ denote the design matrix of another $\tilde{M}$ observations, and suppose we are interested in estimating $\Xnew w$. 
Letting $\bbE_y$ denote  the posterior distribution given $Y = y$ and using~\eqref{eq:full-cond-w}, we obtain that  
\begin{equation}\label{eq:ATT-mean-estimate}
    \bbE_y[ \Xnew w ] =   \bbE_y\left[ \bbE_y(  \Xnew w \mid \gamma,  \mu_\gamma, \tau, \phi  )   \right ] 
    =   \bbE_y \left[   \Xnew V_{\gamma, \tau}^{-1}( X^\top_\gamma y + \tau^{-1} \mu_\gamma )  \right]. 
\end{equation}
Given $n$ MCMC samples $(\gamma^{(k)}, \mu_\gamma^{(k)}, \tau^{(k)}, \phi^{(k)})_{k = 1}^n$, we can approximate the above expectation using the sample average, which yields the posterior predictive mean of $\Xnew w$.  
The covariance matrix of the posterior predictive distribution of $\Xnew w$ can be calculated similarly. 
 
For synthetic control, the goal is to estimate ATT among the $\tilde{M}$ observations.
Denote their response under the treatment by $\Ynew^{(1)}$, which is assumed observed, and denote the counterfactual response without treatment by $\Ynew^{(0)}$. 
Let  $\delta = \Ynew^{(1)} - \Ynew^{(0)}$, and define ATT by 
\begin{equation}\label{eq:ATT}
   \mathrm{ATT} =   
   \frac{1}{\tilde{M}} \sum_{i = 1}^{\tilde{M}} \delta_i  =
   \frac{1}{\tilde{M}} \sum_{i = 1}^{\tilde{M}}  \left( \Ynew^{(1)}_i - \Ynew^{(0)}_i \right), 
\end{equation}
which is the main parameter of interest. 
By replacing $\Ynew^{(0)}$ using the synthetic control $\bbE_y[ \Xnew w ]$,  we obtain an estimator for $\bar{\delta}$. 
If we have samples $(w^{(k)})_{k=1}^n$ approximating the marginal posterior distribution of $w$, we can also obtain the posterior predictive distribution of ATT  by plugging $w^{(k)}$ into the estimator
$ \tilde{M}^{-1}\sum_{i = 1}^{\tilde{M}}  ( \Ynew^{(1)}_i - w^\top \Xnew_{(i)} )$, where $\Xnew_{(i)}$ denotes the $i$-th row of $\Xnew$ (treated as a column vector).

\section{Metropolis-within-Gibbs Sampling for BVS-SS}\label{sec:alg-short}

The standard solution to posterior computation for the spike-and-slab variable selection is to run a Metropolis--Hasting sampler targeting the posterior distribution of $\gamma$ by add/delete/swap proposals; see~\citet{george1997approaches, guan2011bayesian, chang2024dimension}, among many others. However, this approach is not applicable in our setting since the  likelihood $p(y \mid \gamma, \phi)$ in \BVS{} does not have a closed-form expression.  

Instead of directly sampling from the marginal posterior of $\gamma$, we devise a Metropolis-within-Gibbs scheme targeting the full posterior distribution $p(\gamma, \mu, \tau, \phi \mid y)$. 
The Metropolis--Hastings update for $\tau$ and the Gibbs update for $\phi$ are standard, which we detail in Appendix~\ref{sec:gibbs-var}. 
Since  $\gamma_i = \ind_{ \{\mu_i \neq 0\}}$, we update $(\gamma_i, \mu_i)$ simultaneously, which can be  implemented straightforwardly as a Gibbs step in the unconstrained spike-and-slab variable selection. 
However, this is again infeasible for \BVS{} due to the simplex constraint: given $\mu_{-j}$, we have $\mu_j = 1 - \sum_{i \neq j} \mu_i$ almost surely. Therefore, we propose to a novel Gibbs scheme where we update two coordinates of $(\gamma, \mu)$ simultaneously.

Let $\mu_{-(i, j)}$ denote  the subvector of $\mu$ with $(\mu_i, \mu_j)$ removed and define $\gamma_{-(i, j)}$ analogously.  To implement the Gibbs update of $(\gamma_i, \gamma_j, \mu_i, \mu_j)$ for $i \neq j$, we need to find the full conditional posterior $p(\gamma_i, \gamma_j, \mu_i, \mu_j \mid  y, \gamma_{-(i, j)},  \mu_{-(i, j)}, \tau, \phi)$. 
If $\sum_{k \neq i, j} \mu_k = 1$,   neither $i$ nor $j$ can be selected due to the simplex constraint, and thus $\gamma_i = \gamma_j = \mu_i = \mu_j = 0$.  
Now assume $\sum_{k \neq i, j} \mu_k < 1$, in which case we must have $(\gamma_i, \gamma_j) = (1, 0), (0, 1)$ or $(1, 1)$. 
The conditional posterior probability  $p( \gamma_i = 1, \gamma_j = 0 \mid  y, \gamma_{-(i, j)},  \mu_{-(i, j)}, \tau, \phi)$ is easy to evaluate (up to a normalizing constant) since the simplex constraint implies $\mu_i = 1 - \sum_{k \neq i, j} \mu_k$. The case $\gamma_i = 0, \gamma_j = 1$ can be analyzed similarly. 
When $\gamma_i = \gamma_j = 1$, we first find the full conditional posterior $p( \mu_i, \mu_j \mid  y, \gamma_{-(i, j)},  \mu_{-(i, j)}, \tau, \phi, \gamma_i = \gamma_j = 1)$, which is a degenerate truncated normal distribution due to the constraint $\mu_i > 0, \mu_j > 0$ and $\mu_i + \mu_j = 1 - \sum_{k \neq i, j}$. This is derived in Lemma~\ref{lm1}. Integrating out $(\mu_i, \mu_j)$, we get the conditional posterior probability $p( \gamma_i = \gamma_j = 1 \mid  y, \gamma_{-(i, j)},  \mu_{-(i, j)}, \tau, \phi)$  up to a normalizing constant.  Normalizing over the three cases yields the posterior distribution $p( \gamma_i, \gamma_j \mid  y, \gamma_{-(i, j)},  \mu_{-(i, j)}, \tau, \phi)$, which has a closed-form expression in terms of the cumulative distribution function of the normal distribution; see Lemma~\ref{lm2}. The resulting sampling algorithm is summarized in Algorithm~\ref{alg:gibbs}, where $\beta_{i,j}, \Lambda_{i,j}$ are defined in Lemma~\ref{lm1}. More details are provided in Appendix~\ref{sec:gibbs}.

\begin{algorithm}[!h]\setstretch{1.1}
\caption{Metropolis-within-Gibbs sampling for \BVS{}}\label{alg:gibbs}  
\KwIn{data set $(X, Y)$, hyperparameters $\kappa_1, \kappa_2, a_1, a_2, \alpha, \theta$, number of iterations $n$,  algorithm parameters $n_\tau \in \mathbb{N},  \eta > 0$, 
initial state $(\mu^{(0)}, \tau^{(0)}, \phi^{(0)})$.}     
\For{$t = 1, \dots, n$}{   
Set $\mu \leftarrow \mu^{(t - 1)}$ \; 
    \For{$i = 1, \dots, N - 1$}{
        \For{$j = i +1, \dots, N$}{
        Calculate $s \leftarrow  1 - \sum_{k \neq i, j} \mu_k $ \; 
        \If{$s = 0$}{  
        Set  $\mu_i \leftarrow 0$,  $\mu_j \leftarrow 0$\;
        }\Else{ 
        Draw $(\gamma_i, \gamma_j)$ with probability $p(\gamma_i, \gamma_j \mid y,  \mu_{-(i, j)}, \tau, \phi)$ by Lemma~\ref{lm2}\;
        \If{$\gamma_i = 1, \gamma_j = 0$}{
        Set $\mu_i \leftarrow s$,  $\mu_j \leftarrow 0$\; 
        }\ElseIf{$\gamma_i = 0, \gamma_j = 1$}{
        Set $\mu_i \leftarrow 0$,  $\mu_j \leftarrow s$\; 
        }\Else{
        Draw $u \sim N_{ (0, s) }\left(  \beta_{i, j},      ( \phi \Lambda_{i, j})^{-1}  \right)$  by  Lemma~\ref{lm1}\; 
        Set $\mu_i \leftarrow u$, $\mu_j \leftarrow s - u$\; 
        }
        }
        }
    }
Set $\mu^{(t)} \leftarrow \mu$\;
Draw $\phi^{(t)}$ from $p(\phi \mid y, \mu^{(k)}, \tau^{(k-1)})$ by~\eqref{eq:full-cond-phi}\;
Set $\tau \leftarrow \tau^{(t - 1)}$\; 
\For {$i = 1, \dots, n_\tau$}{
Propose $\tau^*$ by $\log \tau^*  = \log \tau + N(0, \eta)$\;  
Set $\tau \leftarrow \tau^*$ with acceptance probability $\rho(\tau, \tau^*)$ given by  equation~\eqref{eq:acc-tau}\; 
}
Set $\tau^{(t)} \leftarrow \tau$\; 
Draw $w^{(t)}$ from   $p(w \mid y, \mu^{(t)}, \tau^{(t)}, \phi^{(t)})$ by~\eqref{eq:full-cond-w}\; 
} 
\KwOut{samples $(\mu^{(t)}, \tau^{(t)}, \phi^{(t)}, w^{(t)})_{t=1}^n$}  
\end{algorithm}

\section{Theoretical Results for BVS-SS}\label{sec:theory}
In this section, we study the theoretical properties of \BVS{}. 
For ease of analysis, we fix the hyperparameter $\tau$ and consider two complementary scenarios: $\tau = 0$ and $\tau \rightarrow \infty$.  
Following the convention in the variable selection literature~\citep{yang2016computational}, we interpret a model $\gamma$ as both a binary vector and a subset of $[N]:=\{1,2,\dots,N \}$; that is, to indicate that the $i$-th covariate is selected, we can write either $\gamma_i = 1$ or $i \in \gamma$.  We use $|\cdot|$ to denote set cardinality, and thus $|\gamma|$ denotes the model size. 

\subsection{High-dimensional  Selection Consistency with \texorpdfstring{$\tau = 0$}{tau = 0}} \label{sec:HDC-S}
For the first  scenario, we fix the prior parameter $\tau = 0 $ and assume that the true DGP satisfies the simplex constraint.   
Our goal is to establish, in high-dimensional regimes, the ``strong selection consistency'' property of \BVS{} which means that the posterior probability of the true model, denoted by $\gamma^*$, converges to one in probability under the true DGP; here, ``strong'' means that this property is stronger than other consistency criteria such as pairwise selection consistency~\citep{casella2009consistency}.  
Since $\tau = 0$ implies $w = \mu$ almost surely,  we can rewrite the model as 
\begin{equation}  \label{eq:truemode-tau0}
    Y = X \mu + \epsilon, \quad \epsilon \sim N(0, \phi^{-1} I ). 
\end{equation} 
For technical convenience, we slightly modify the prior specification as follows: 
\begin{align}
 p(\gamma) &\propto \Gamma(|\gamma|)^{-1}  \left\{   \mathsf{1}^\top (X_\gamma^\top X_\gamma)^{-1} \mathsf{1} \right\}^{1/2} \det(X_\gamma^\top X_\gamma)^{1/2}   \theta^{|\gamma|} \, (1 - \theta)^{N - |\gamma|} \ind_{\bbS_L}(\gamma),  \label{eq:gamma-prior}  \\ 
 p(\phi \mid \gamma) &\propto \phi^{ (\kappa_1 + |\gamma|-1)/2 - 1} e^{- \phi \kappa_2 /2} ,  \label{eq:phi-prior} \\   
    \mu_\gamma \mid \gamma &\sim \text{Dir}(\mathsf{1}), 
\end{align} 
where $\bbS_L = \left\{ \gamma \subset [N] \colon 1 \leq |\gamma| \leq L \right\}$ denotes the sparse model space with size bounded by $L$, and Dir$(\mathsf{1})$ denotes the  Dirichlet distribution with parameter vector $\mathsf{1} = (1, \dots, 1)$ (i.e., uniform distribution over the simplex $\Delta^{|\gamma| - 1}$).  
Compared to the original prior specification given by~\eqref{eq:modelprior}, the changes are minor. First, rather than letting $\gamma_i \overset{\mathrm{i.i.d.}}{\sim} \mathrm{Bernoulli}(\theta)$, we reweight the prior probability of each $\gamma$ by  three factors, where   $\Gamma(|\gamma|)^{-1}$ is the volume of the simplex  $\Delta^{|\gamma| - 1}$, and $\mathsf{1}^\top (X_\gamma^\top X_\gamma)^{-1} \mathsf{1}$ is the variance of $\mathsf{1}^\top \hat{\mu}_\gamma$ with $\hat{\mu}_\gamma$ denoting the OLS estimator for model $\gamma$.  This adjustment serves only to simplify the presentation, and we prove in  Lemma~\ref{lm:gamma-prior} below that the effect of this reweighting is typically negligible when $\theta$ decays much faster than $\sqrt{M}$. The sparsity condition $|\gamma| \leq L$ is standard and necessary for high-dimensional consistency analysis. 
Second, we let the prior of $\phi$ depend on $\gamma$ so that for a larger model, the prior mean of $\phi$ is larger (i.e., the error variance is smaller). This modification  enables a simpler comparison of the marginal likelihoods of different models after integrating out $\phi$.  
When $\kappa_1$ is large (e.g. of order $M$), the difference between the two prior specifications for $\phi$ becomes negligible. 
Third, for the conditional prior for $\mu_\gamma$ given $\gamma$, we set $\alpha = 1$ in~\eqref{eq:modelprior} to ensure that the prior density $p(\mu_{\gamma} \mid \gamma)$ only depends on $|\gamma|$.

\begin{lemma}\label{lm:gamma-prior}
  Let $\gamma' = \gamma \cup \{ j \}$ for some $j \notin \gamma$, and $p(\gamma)$ be given by~\eqref{eq:gamma-prior}.  Then, 
  \begin{equation*}
\frac{\sqrt{M \lmin} }{L}    \leq    \frac{ p (\gamma') / p(\gamma)}{ \theta / (1 - \theta) } \leq \sqrt{ \frac{L M}{ \lmin} }. 
  \end{equation*}
\end{lemma}
\begin{proof}
  See Appendix~\ref{sec:proof-lm-gamma-prior}. 
\end{proof}

By integrating out $\mu_\gamma$, we can express the posterior density of $(\gamma, \phi)$ as 
\begin{align} 
    p(\gamma, \phi \mid y) &\propto p(\gamma) p(\phi \mid \gamma)   \bbE_{ \mu_\gamma \sim \mathrm{Dir}(\mathsf{1}) } \left[ p(y \mid \gamma, \mu_\gamma, \phi) \right] , \\
  \text{ where }  p(y \mid \gamma, \mu_\gamma, \phi)  &\propto  \phi^{M/2} \,
    \exp\left\{ -\frac{\phi}{2} \| y - X_\gamma \mu_\gamma \|_2^2  \right\}.   \label{eq:ga-post} 
\end{align}
Unlike the unconstrained Bayesian linear regression,  $p(\gamma, \phi \mid y) $ is not available in closed form due to the intractable integral involved in    $p(y \mid \gamma, \phi)$. Consequently, the marginal posterior probability of $\gamma$, $p(\gamma \mid y) \propto p(\gamma)   \int p(\phi \mid \gamma) p(y \mid \gamma, \phi) \d \phi$, has no closed-form expression either.   
To establish  the high-dimensional consistency result, we make the following assumptions on the design matrix, prior parameters and true DGP, where all parameters involved  may vary with the sample size $M$, except the universal constants denoted by $c_{\mu}, c_{\theta}, c_M$ and $c_j$ for integer $j$. In particular, the number of variables $N$, the maximum model size $L$, and the true model size $\ell^* = |\gamma^*|$ are all allowed to go to infinity with $M$.

\renewcommand{\theenumi}{(A\arabic{enumi})}   
\renewcommand{\labelenumi}{\theenumi}   

\begin{enumerate}  
    \item The design matrix $X$ satisfies that $\|X_j\|_2^2 = M$ for each $j \in [N]$, and there exists  a constant  $\lmin \in (0, 1]$ such that $ \lambda_{\mathrm{min}}(X^\top_\gamma X_\gamma) \geq  M \lmin$ for all $\gamma \in \bbS_L$,  
where $\lambda_{\mathrm{min}}$ denotes the smallest eigenvalue.   \label{asp:design} 

\item   The prior distribution on $\gamma$ is given by~\eqref{eq:gamma-prior} where $\theta$ satisfies   $\theta/(1 - \theta) = N^{- c_{\theta} \, L}$   for some universal constant $c_{\theta} > 0$.   
\label{asp:prior} 

\item  The prior distribution on $\phi$ is given by~\eqref{eq:phi-prior} where 
hyperparameters $\kappa_1, \kappa_2$ satisfy $0 < \kappa_1 \leq M$, and  $0 \leq \kappa_2 \leq \sigma^2  M / 2$.   \label{asp:hyperparameter}

\item  The true DGP is given by $Y \mid X \sim \cN  ( X_{\gamma^*} \mu^*_{\gamma^*}, \, \sigma^2 I_M  )$  where  $\sigma > 0$,  $\gamma^*$ and $\mu^*_{\gamma^*}$ satisfy the following conditions:  $\ell^* \coloneqq |\gamma^*|  \leq L \wedge \sqrt{L \log N}$,   $\mu^*_{\gamma^*} \in \Delta^{ {\ell^*} - 1}$, and   
\begin{equation}\label{eq:beta-min}
    \min_{j \in \gamma^*}   | \mu^*_j |  \geq \frac{ c_{\mu} \sigma \sqrt{L \log N} }{ \lmin \sqrt{M   }}, 
\end{equation}
for some universal constant $c_{\mu} > 0$.  
\label{asp:min-thershold} 
\item $L \geq 3$ and $M \geq c_M \ell^* L \log N$ for some universal constant $c_M > 0$. 
\label{asp:sample-size}
\end{enumerate}
 
Assumption~\ref{asp:design} requires that each column of $X$ has the same magnitude and that the minimum eigenvalue  of the Gram matrix is well controlled for any model of size at most $L$. It is a mild condition, often known as ``restricted eigenvalue,'' commonly used in the literature~\citep{shang2011consistency,narisetty2014bayesian,yang2016computational}.  
Assumption~\ref{asp:prior} provides guidance for choosing the prior: the  prior inclusion probability $\theta$ should be small enough in order to penalize large models and avoid overfitting. Similarly, Assumption~\ref{asp:hyperparameter} requires that the choices of $\kappa_1, \kappa_2$ are not too extreme. 
Assumption~\ref{asp:min-thershold} assumes that the true DGP satisfies the simplex constraint, and~\eqref{eq:beta-min} imposes a lower bound on the smallest nonzero coefficient of the true DGP. This is also a standard assumption and often known as the ``$\beta$-min condition.'' 
Intuitively, this assumption is needed since coefficients close to zero cannot be consistently identified through variable selection. Note that the order of the detection threshold required by~\eqref{eq:beta-min} is essentially the same as that in the unconstrained Bayesian variable selection~\citep{yang2016computational}. 
Finally, Assumption~\ref{asp:sample-size} specifies the sample size  needed for consistent variable selection. 
We  use $\P^*$ to denote the probability measure associated with the true DGP specified in Assumption~\ref{asp:min-thershold}.

Under the above assumptions, we  prove that $p(\gamma^* \mid y)$ converges to one in probability with respect to $\P^*$ given sufficiently large $c_{\mu}, c_{\theta}$ and $c_M$. A key step in our proof is to establish the following posterior probability ratio bound. 

\begin{theorem}\label{th:marg-posterior-ratio}
Suppose Assumptions~\ref{asp:design} to~\ref{asp:sample-size} hold  with $c_M \geq 4$ and $c_\mu \geq 6$. With probability at least $1 - c_1 N^{-c_2}$ for some universal constants $c_1, c_2 > 0$,  the following bound holds for all $\gamma \in \bbS_L$:   
\begin{equation} \label{eq:post-ratio-simplified}
    \frac{ p(\gamma \mid y )}{ p(\gamma^* \mid y )}  
    \leq 3 \left( \frac{\theta \sqrt{ 2 \pi } }{1 - \theta} \right)^{ |\gamma| - \ell^* }   
    \left( \frac{  \kappa_2 + \| y - X_{\gamma^*} \hat{\mu}_{\gamma^*} \|_2^2  + \sigma^2 \ell^* \log N }{ \kappa_2 + \| y - X_\gamma \hat{\mu}_\gamma \|_2^2  } \right)^{  (\kappa_1 + M ) / 2}, 
\end{equation}   
where $\hat{\mu}_\gamma = (X_\gamma^\top X_\gamma)^{-1} X_\gamma^\top  y$. 
\end{theorem}

\begin{proof}[Proof sketch]
The complete proof of Theorem~\ref{th:marg-posterior-ratio} is presented in Appendix~\ref{sec:proof-th-ratio}. Here we provide a sketch of the proof to highlight the key challenges and proof techniques.  
As shown in~\eqref{eq:ga-post}, the marginal likelihood given $(\gamma, \phi)$ can be expressed as $p(y \mid \gamma, \phi) =   \bbE_{ \mu_\gamma \sim \mathrm{Dir}(\mathsf{1}) } \left[ p(y \mid \gamma, \mu_\gamma, \phi) \right]$, where the expectation is an integral over the simplex.  To analyze this intractable integral,  we first show that, for any $\mu_\gamma \in \Delta^{|\gamma| - 1}$,  the likelihood $p(y \mid \gamma, \mu_\gamma, \phi)$ can be expressed as 
\begin{align}
    p(y \mid \gamma, \mu_\gamma, \phi) &\propto   \phi^{M/2} \, \exp\left\{-\frac{\phi}{2} \left(  \| y - X_\gamma \hat{\mu}_\gamma  \|_2^2  + b_\gamma  + \| X_\gamma  (\mu_\gamma  -    \check{\mu}_\gamma ) \|_2^2 \right) \right\},  \label{eq:decomp-like} \\
\text{ where }   \check{\mu}_\gamma &=  \hat{\mu}_\gamma + \frac{1- \mathsf{1}^\top\hat{\mu}_\gamma} {\mathsf{1}^\top (X_\gamma^\top X_\gamma)^{-1} \mathsf{1}} (X_\gamma^\top X_\gamma)^{-1} \mathsf{1}, \quad 
      b_\gamma  =  \frac{   (1 - \mathsf{1}^\top \hat{\mu}_\gamma )^2  }{\mathsf{1}^\top (X_\gamma^\top X_\gamma)^{-1} \mathsf{1} }.  
\end{align}
That is, we decompose the residual sums of squares $\| y - X_\gamma \mu_\gamma \|_2^2$   into three parts. The first component, $\| y - X_\gamma \hat{\mu}_\gamma \|_2^2$, is the residual sum of squares for the least-squares estimator in the unconstrained setting.  The second component, $b_\gamma \geq 0$, measures the deviation of   $\hat{\mu}_\gamma$ from the simplex constraint. When $\hat{\mu}_\gamma \in \Delta^{|\gamma| - 1}$, we have $b_\gamma = 0$. Note that the first two components do not depend on $\mu_\gamma$. 
The last component,   $\|  X_\gamma  (\mu_\gamma - \check{\mu}_\gamma) \|_2^2 $, measures the deviation of $\mu_\gamma$ from $\check{\mu}_\gamma$, where $\check{\mu}_\gamma$  is the least-squares estimator for model $\gamma$ under the linear constraint $\mathsf{1}^\top \check{\mu}_\gamma = 1$. It has the following interpretation:  $X_\gamma \check{\mu}_\gamma$ is the projection of $X_\gamma \hat{\mu}_\gamma$ onto the affine space $\{ X_\gamma u \colon \mathsf{1}^\top u = 1 \}$. See Appendix~\ref{sec:constrained-ls} for a brief review on the constrained least-squares estimation. 
  
In Proposition~\ref{prop:post-expression}, we use~\eqref{eq:decomp-like} and the modified priors given in~\eqref{eq:gamma-prior} and \eqref{eq:phi-prior} to show that 
\begin{equation}
p(  \gamma , \phi \mid y ) \propto  f(|\gamma|) \,  \phi^{(M + \kappa_1)/2}      \, e^{-\frac{\phi}{2} \left(  \| y - X_\gamma \hat{\mu}_\gamma  \|_2^2  + b_\gamma + \kappa_2 \right)  } \,   \P \left( \tilde{U}_{\gamma, \phi} \in {\tilde{\Delta}^{|\gamma| - 1}} \right), 
\end{equation}   
where $f(|\gamma|)$ is a function only depending on $|\gamma|$ and prior parameter $\theta$, $\tilde{U}_{\gamma, \phi}$ is a normal random vector whose mean  equals the subvector of $\check{\mu}_{\gamma}$ consisting of its first $|\gamma| - 1$ entries,  and $\tilde{\Delta}^\ell =  \{  \tilde{u} \in \bbR^{\ell} \colon  \tilde{u}_i \geq 0 \text{ for each $i$, and } \sum\nolimits_{i=1}^\ell \tilde{u}_i  \leq 1  \}.$ 
For any $\gamma \in \bbS_L$, we have $ \P  ( \tilde{U}_{\gamma, \phi} \in {\tilde{\Delta}^{|\gamma| - 1}}  ) \leq 1$ and $b_\gamma \geq 0$. Hence, by  integrating over $\phi$, we get an upper bound on $p(\gamma \mid y)$ that holds without any additional assumption; see Proposition~\ref{coro:upper-any-model}.  

To find a lower bound on $p(\gamma^*, \phi \mid y)$, we need to find upper bounds on $ \| y - X_{\gamma^*} \hat{\mu}_{\gamma^*}  \|_2^2$ and $b_{\gamma^*}$ and a lower bound on $\P ( \tilde{U}_{\gamma^*, \phi} \in {\tilde{\Delta}^{\ell^* - 1}}  )$. To this end, we construct some high-probability events in Lemma~\ref{lm:gaussian-conc} to ensure that the error vector $\epsilon$ is well-behaved. This argument is standard in the high-dimensional literature and ensures that $ \| y - X_{\gamma^*} \hat{\mu}_{\gamma^*}  \|_2^2 = O_p(\sigma^2 M)$.  For $b_{\gamma^*}$,  we expect that it cannot be too large under our assumptions, as it measures how well $\hat{\mu}_{\gamma^*}$ satisfies the simplex constraint. 
In Lemma~\ref{lm:true-model-bgamma},  we prove that $b_{\gamma^*} = O_p(\sigma^2 \ell^* \log N)$.  Deriving a lower bound on $\P ( \tilde{U}_{\gamma^*, \phi} \in {\tilde{\Delta}^{\ell^* - 1}}  )$   is considerably more involved. We invoke  the beta-min condition given in Assumption~\ref{asp:min-thershold}, which ensures that the true regression coefficients cannot be too small, to show that $\check{\mu}_{\gamma^*}$ lies in the simplex $\Delta^{\ell^* - 1}$. 
Then,  in Lemma~\ref{lm:normal-bound}, we apply standard Gaussian concentration inequalities to get 
\begin{equation} \label{ball-bound} 
 1 - \P \left( \tilde{U}_{\gamma^*, \phi} \in {\tilde{\Delta}^{\ell^* - 1}} \right)  
  \leq  \exp\left( - \frac{4 \phi \, \sigma^2 L \log N}{   (\ell^*)^2  } + \frac{1}{2} \right). 
\end{equation}  
We visualize the intuition behind the proof of~\eqref{ball-bound} in Figure~\ref{graphlb}.  
The form of the bound~\eqref{ball-bound} enables us to integrate over $\phi$, and by using Assumptions~\ref{asp:hyperparameter} and~\ref{asp:sample-size}, we obtain a lower bound on $p(\gamma^* \mid y)$; see Proposition~\ref{coro:lower-true}. 
Theorem~\ref{th:marg-posterior-ratio} then follows by combining Propositions~\ref{coro:upper-any-model} and~\ref{coro:lower-true}. 
\end{proof}

\begin{figure}[H]
    \centering
    \includegraphics[ width=0.5\linewidth]{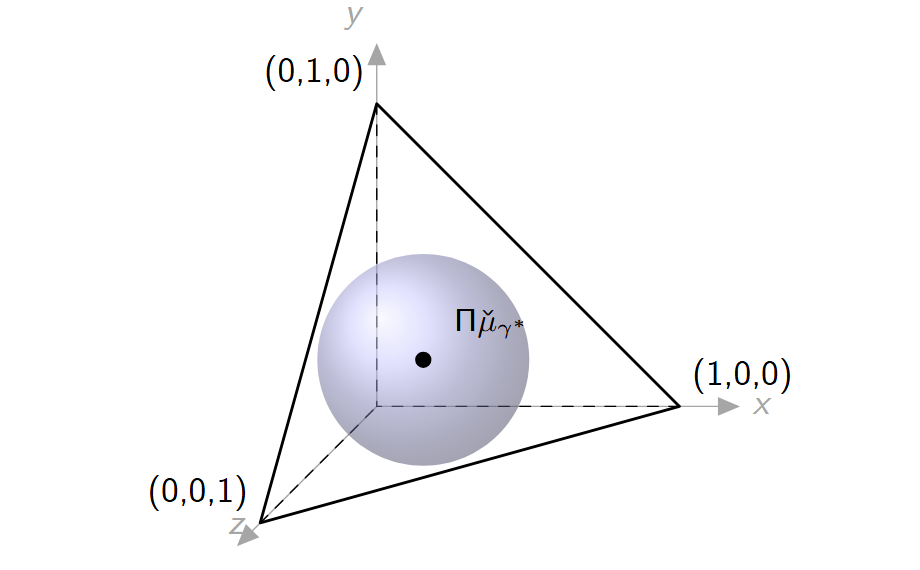}
    \caption{ Intuition for inequality~\eqref{ball-bound}. Here $\ell^* = 4$,  the tetrahedron  represents $\tilde{\Delta}^3$, and $\Pi \check{\mu}_{\gamma^*}$ denotes the mean of $\tilde{U}_{\gamma^*, \phi}$. We  show that $\Pi \check{\mu}_{\gamma^*} \in \tilde{\Delta}^{3}$ and   use concentration inequality to bound the probability mass of  a ball centered at $\Pi \check{\mu}_{\gamma^*}$ and contained in $\tilde{\Delta}^{ 3}$.  }
    \label{graphlb}
\end{figure}

By Theorem~\ref{th:marg-posterior-ratio}, we can bound $p(\gamma \mid y)/p(\gamma^* \mid y)$ by analyzing the behavior of unconstrained least-squares estimators, which is well understood in the variable selection literature. 
This analysis then yields the desired strong selection consistency.

\begin{theorem} \label{thm:post-cons}
  Suppose Assumptions~\ref{asp:design} to~\ref{asp:sample-size} hold for sufficiently large $c_\theta, c_\mu, c_M$. Then, $\P^* \{ p(\gamma^* \mid y) \geq 1 - c_3N^{- c_4 L} \}  \geq 1 - c_1 N^{-c_2 }$ for some universal constants $c_1, c_2, c_3, c_4 > 0$. 
\end{theorem}
\begin{proof}
    See Appendix~\ref{sec:post-cons-proof}. 
\end{proof}

The strong selection consistency ensures that the true model $\gamma^*$ can be identified with high probability. Next, we consider the  estimation of the regression coefficient vector $\mu$ and the task of prediction. 
The posterior expected  $L^2$ loss for estimating $\mu$   can be expressed as 
\begin{equation}\label{eq:mu-L2-loss}
  \bbE_y \left[   \| \mu - \mu^* \|_2^2   \right] 
= \sum_{\gamma \in \bbS_L  }  p(\gamma \mid y) \left( \| \mu^*_{\gamma^c} \|_2^2 +  \int    \| \mu_\gamma - \mu^*_\gamma \|_2^2  \;  p (\mu_\gamma, \phi \mid y, \gamma) \d \mu_\gamma \d \phi \right), 
\end{equation}
where we recall $\bbE_y$  denotes the expectation over the posterior distribution given $Y = y$. 
Similarly, given a new design matrix $\Xnew$, we can define the posterior expected predictive loss by $\bbE_y [   \| \Xnew \mu - \Xnew \mu^* \|_2^2   ]$.  
In Theorem~\ref{thm:mu-cons}, we characterize the rate at which the two loss functions go to zero.   

\begin{theorem}\label{thm:mu-cons}
    Under the setting of Theorem~\ref{thm:post-cons},  
    \begin{equation}
        \bbE_y \left[   \| \mu - \mu^* \|_2^2   \right] = O_p \left( \max\left\{ \frac{ \sigma^2   L \log N  }{ \lmin^2  M }, \; N^{- L / (\ell^*)^2 }  \right\} \right). 
    \end{equation}  
Let $\Xnew \in \bbR^{\tilde{M} \times N }$ be   such that $ \max_{\gamma \in \bbS_L} \lambda_{\mathrm{max}}(\Xnew^\top_\gamma \Xnew_\gamma) \leq \tilde{M} \lmax$ for some $\lmax > 0$. Then, 
    \begin{equation}
        \bbE_y \left[    \| \Xnew \mu - \Xnew \mu^* \|_2^2   \right] =    O_p \left(  \tilde{M} \lmax \,  \max\left\{ \frac{ \sigma^2   L \log N  }{ \lmin^2  M }, \; N^{- L / (\ell^*)^2 }  \right\} \right). 
    \end{equation} 
\end{theorem}
\begin{proof}
    See Appendix~\ref{sec:mu-cons-proof}. 
\end{proof}

The consistency of our ATT estimator follows from the the prediction loss bound in Theorem~\ref{thm:mu-cons}. 
We expect that this result can be extended to more general scenarios, including cases where $[X^\top \; \Xnew^\top]$ is a stationary process~\citep{li2020statistical}.

\subsection{Limiting Behavior when \texorpdfstring{$\tau \rightarrow \infty$}{tau rightarrow infty}} \label{sec:HDC-NS}

We also investigate the  performance of \BVS{} when  $\tau \to \infty$. In this case, the prior places increasing mass on values of $w$ that do not satisfy the simplex constraint.  
We compare the behavior of the posterior distribution $p(\gamma \mid y)$ with the conventional unconstrained version, which we denote by $\tilde{p}(\gamma \mid y)$ (see  Appendix~\ref{sec:const-NS-proof} for details). 
The following theorem shows that the two posterior distributions become essentially the same as $\tau \rightarrow \infty$. 

\begin{theorem} \label{thm:eqbvscss}
   Let $\tilde{p}(\cdot \mid y, \tau)$ be the posterior distribution given $\tau$ under the unconstrained  spike-and-slab prior detailed in Appendix~\ref{sec:const-NS-proof}. Then
       $\lim_{\tau \rightarrow \infty}  p(\gamma \mid y, \tau ) / \tilde{p}(\gamma \mid y, \tau)  =1.$ 
\end{theorem}
\begin{proof}
   See Appendix~\ref{sec:const-NS-proof}.
\end{proof}

In Theorem~\ref{thm:eqbvscss}, we treat the data $(X, y)$ as fixed and let $\tau \rightarrow \infty$, but we expect that our argument can be extended to a fully non-asymptotic analysis. Such an extension may be combined with existing variable selection consistency results to show that our model also achieves selection consistency  when the true DGP significantly violates the simplex constraint but $\tau$ is sufficiently large.

\section{Simulation Studies}\label{sec:sim}

\subsection{Simulation Settings}\label{sec:sim-setting}
We simulate the observed training data $X \in \bbR^{M \times N}$ and $Y \in \bbR^{M}$ by  $Y =  X w^* + \epsilon,$  where  $\epsilon_i \overset{\mathrm{i.i.d.}}{\sim} N(0, 1/ \phi^*   )$,  and $w^* \in \bbR^N$ is the true regression coefficient vector such that $w^*_j = 0$ for $j = J+1, \dots, N$ (that is, only the first $J$ predictors have nonzero effects on $Y$). 
Similarly, we simulate the data under treatment  $\Xnew \in \bbR^{\tilde{M} \times N}$ and $\Ynew^{(1)} \in \bbR^{\tilde{M}}$ by 
\begin{align} 
    \Ynew^{(1)} =\;& \Xnew w^* + \delta + \tilde{\epsilon}, \text{ where }
    \delta_i \overset{\mathrm{i.i.d.}}{\sim} N(\delta^*, 1/ \nu^* ), \; 
    \tilde{\epsilon}_i \overset{\mathrm{i.i.d.}}{\sim} N(0, 1/ \phi^*   ). 
\end{align}
The vector $\delta \in \bbR^{\tilde{M}}$  represents the treatment effects, and $\delta^* \in \bbR$ denotes the true ATT. 
In our simulation, we use $M = \tilde{M} \in \{25, 50, 100, 200\}$ and $N \in \{20, 50\}$. 
When $N = 20$, we set $J = 5$, and when $N = 50$, we set $J = 10$. 
Given $J$, we define a vector $\mu^* \in \Delta^{N - 1}$ by
\begin{equation}
    \mu^*_j = \left\{\begin{array}{cc}
     j / S_J  &   \text{ if } 1 \leq j \leq J, \\ 
       0, &  \text{ if }  J + 1 \leq j \leq N, 
    \end{array}
    \right. \quad \text{ where }  S_J = \frac{1}{2} J (J + 1). 
\end{equation}
Then, we define $w^*$ by  $w^* = \lambda \mu^*$ where $\lambda > 0$ is a scaling factor.  When $\lambda = 1$, the true DGP satisfies the simplex constraint. Assuming $\lambda > 0$, we have $\lambda = \|w^*\|_1$.  
We consider $\lambda \in \{1, 2, 3\}$ in our simulation studies. 
Given $w^*$, we set the variance parameters by $v^* = \phi^* = 4 / \| w^* \|_2^2$. 
For each choice of $(M, N, \lambda)$, we generate 100 replicates where entries of $X, \Xnew$ are also sampled independently from the standard normal distribution. See  Appendix~\ref{sec:more_sim} for simulation results where the true DGP is driven by a factor model.  We fix true ATT $\delta^* = 0.5$ throughout our simulation studies.

We set the hyperparameters of \BVS{} by $\kappa_1 = \kappa_2 = 1, a_1 = 0.01, a_2 = 0.1, \alpha = 1, \theta = 0.2.$  The choice of $(a_1, a_2)$ guarantees that $\tau$ has prior mean $0.1$ and prior variance $1$. 
Further, for numerical stability, we truncate the prior distribution of $\tau$ by assuming that $\tau \geq 10^{-6}$.  
For every simulated data set, we initialize Algorithm~\ref{alg:gibbs}  at $\tau^{(0)} = \phi^{(0)} = 1$ and $\mu^{(0)}$ sampled from its prior distribution, and we run Algorithm~\ref{alg:gibbs} for $1,000$ iterations and discard the first $500$ iterations as burn-in. 

For comparison, we have also implemented the following methods for ATT estimation. First, we use the \texttt{R} package \texttt{glmnet} to perform variable selection and ATT estimation via Lasso; the parameter of Lasso is chosen by cross-validation. 
Second, we compute the OLS estimator for $w^*$ using the entire data set $(X, Y)$ without variable selection, denoted by $\hat{w}_{\mathrm{LS}}$. Then ATT is estimated by 
\begin{equation}\label{eq:att-est}
    \widehat{\mathrm{ATT}}_{\mathrm{LS}} = \frac{1}{\tilde{M}} \sum\nolimits_{i = 1}^{\tilde{M}} \left(  \Ynew^{(1)}_i - \hat{w}_{\mathrm{LS}}^\top  \Xnew_{(i)} \right), 
\end{equation}
where $\Xnew_{(i)}$ is the $i$-th row of $\Xnew$ treated as a column vector. 
We also compute the simplex-constrained estimator $\hat{w}_{\mathrm{QP}}$ (QP stands for ``quadratic programming'') defined by  
$\hat{w}_{\mathrm{QP}} = \argmin_{w \in \Delta^{N-1}}   \| Y - X w \|_2^2,$  
which was introduced in~\cite{abadie2003economic} and can be obtained by quadratic programming~\citep{goldfarb2006dual, abadie2011synth}. 
The corresponding ATT estimator $\widehat{\mathrm{ATT}}_{\mathrm{QP}}$ is obtained similarly by replacing 
$\hat{w}_{\mathrm{LS}}$ with $\hat{w}_{\mathrm{QP}}$ in~\eqref{eq:att-est}. 
Finally, we consider two oracle estimators that are not implementable in practice. 
Denote by $\gamma^* = \gamma(w^*)$ the true indicator vector corresponding to $w^*$. 
We compute the OLS and QP estimators using $(X_{\gamma^*}, Y)$. That is, these two estimators have access to $\gamma^*$ (which is not possible in reality) and will be used for comparison in our simulation studies. 
Note that since the regression coefficient estimator of Lasso and OLS are scale-equivariant, the ATT estimators of Lasso and two OLS methods does not depend on the signal size factor $\lambda$. 

\subsection{Estimation of Average Treatment Effect}\label{sec:sim-att}
For each method considered, we compute the mean squared error of ATT estimation  over 100 replicates; denote it by $\mathrm{MSE}(\widehat{\mathrm{ATT}})$. 
To measure the efficiency of an ATT estimator, we consider the oracle OLS estimator OLS$(\gamma^*)$ as the reference estimator, and define the relative efficiency of an estimator $\widehat{\mathrm{ATT}}$ by 
\begin{equation}
    \mathrm{RE}( \widehat{\mathrm{ATT}} ) = \frac{ \mathrm{MSE}(\widehat{\mathrm{ATT}}_{\mathrm{OLS}(\gamma^*) }) }{\mathrm{MSE}(\widehat{\mathrm{ATT}}) }.
\end{equation}
If $  \mathrm{RE}( \widehat{\mathrm{ATT}} ) = 100\%$, then this estimator is as efficient as the oracle OLS estimator. 

We show the relative efficiency of the other five estimators with $\lambda = \|w^* \|_1 = 1, 2, 3$ in Figure~\ref{plot:re1}. When $\lambda = 1$,  our method outperforms Lasso substantially  when $M = 25$ or $50$, which shows that some prior knowledge about the simplex constraint can be very helpful when the sample size is small or moderate. By comparing \BVS{} with the two QP methods,  we see that performing variable selection also significantly improves ATT estimation when the sample size is small. Indeed, when $N = 50, M = 25$, QP with the entire data set is not feasible since the underlying optimization problem does not have a unique solution. We also note that as sample size increases, the performance of \BVS{} quickly approaches that of the two oracle estimators, OLS$(\gamma^*)$ and QP$(\gamma^*)$. OLS with the entire data set performs poorly unless the sample size is very large, which is also expected since it does not perform variable selection.

\begin{figure}[!h]
    \centering
    \includegraphics[width=0.9\linewidth]{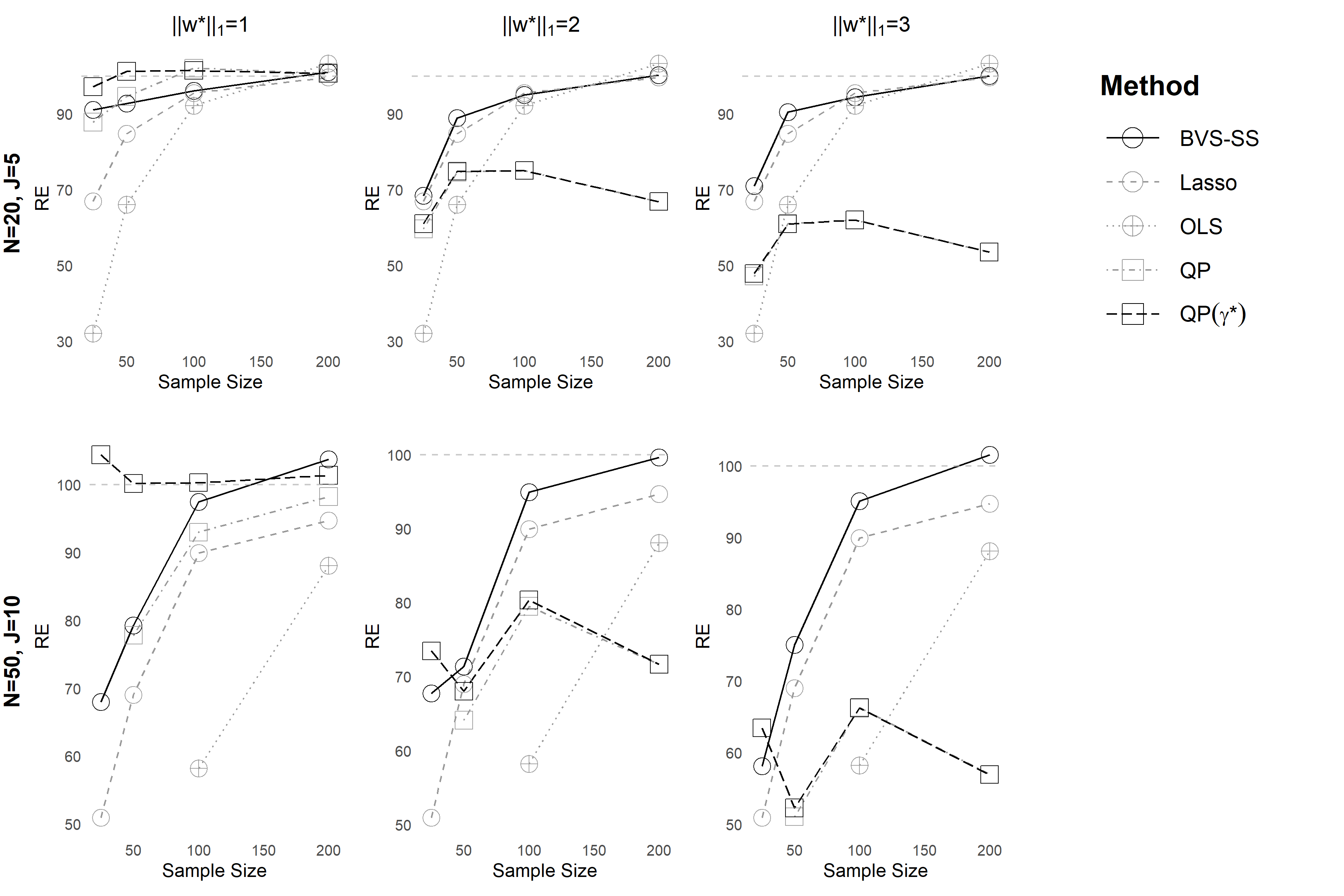}
    \caption{ Relative efficiency (\%) vs. sample size for ATT estimation. \BVS{} and  QP($\gamma^*$) are in black; Lasso, OLS and QP in gray.  OLS($\gamma^*$) always has RE=100\% (horizontal line).  }
    \label{plot:re1}
\end{figure}

When $\lambda = 2$ or $3$,   the simplex constraint is violated by the true DGP.  
Though \BVS{} utilizes the simplex constraint, Figure~\ref{plot:re1} shows that its performance remains robust. 
\BVS{} outperforms Lasso in all scenarios, and when the ratio $N/M$ is large, the advantage of \BVS{} is very significant. 
Indeed, we observe that when $N = 50, M = 25$, Lasso almost fails to recover any signal contained in $w^*$ (see also Table~\ref{table:vs}). 
When the sample size becomes sufficiently large (say, $M \geq 100$), \BVS{} again achieves close-to-optimal performance, and it outperforms QP$(\gamma^*)$ by a wide margin, which knows the true model but assumes the simplex constraint.  
The robust and superior performance of  \BVS{} across the three settings is largely due to the use of the soft simplex constraint, which enables \BVS{} to adapt to the unknown level of $\|w^*\|_1$ in the given data.

\subsection{Variable Selection and 
Estimation of Variance Parameters}\label{sec:sim-gibbs}
 
We now examine the performance of \BVS{} in terms of variable selection and the estimation of the two variance parameters $\tau, \phi$. 
First, we compare the variable selection result of \BVS{} with that of Lasso in Table~\ref{table:vs}, where ``$\ell^1$-loss'' denotes the $\ell^1$-distance between the selected model $\hat{\gamma}$ and the true model $\gamma^*$, and ``model size'' refers to the cardinality of $\hat{\gamma}$ (i.e.,  the number of nonzero entries of $\hat{w}$). For \BVS{}, the two statistics are averaged over the collected MCMC samples. 
It is evident that \BVS{} has a much higher accuracy than Lasso across all settings.
Especially when $M$ is large, \BVS{} can correctly identify most entries of $\gamma$, while the performance of Lasso remains similar to that observed with small $M$.

Next, we visualize the distribution of the posterior mean estimate of $\phi$ (error precision) across 100 replicates in Figure~\ref{fig:phi-50}. 
This statistic reflects  how well the \BVS{} model fits the data. 
We only show the result for $N = 50$, as the distribution for $N = 20$ is similar. 
Observe that regardless of $\|w^*\|_1$, the posterior mean of $\phi$ increases with $M$, which is expected since a larger sample size enables \BVS{} to detect more signals in the data, resulting in a smaller estimate of the error variance.  
Notably, when $\|w^*\|_1 = 2$ or $3$, the posterior estimation of $\phi$ appears highly accurate for $M \geq 100$, indicating that our model achieves an optimal balance between model complexity and fitting.  

\begin{table}[!h] 
\centering
\begin{tabular}{ccccccc}
  \toprule
  & &   & \multicolumn{2}{c}{\BVS{}} & \multicolumn{2}{c}{Lasso}    \\ 
$\|w^*\|_1$  & Setting & $M$ & $\ell^1$-loss & model size &  $\ell^1$-loss & model size \\ 
  \midrule
  \multirow{8}{*}{\makecell{$\|w^*\|_1$ = 1}}  &
\multirow{4}{*}{\makecell{$N = 20$ \\ $J = 5$}}   
   & 25 & 3.1 & 3.9 & 6.3 & 10.2 \\ 
   &  & 50 & 1.9 & 4.8 & 6.7 & 11.1 \\ 
   &  & 100 & 1.3 & 5.1 & 6 & 10.9 \\ 
   &  & 200 & 0.8 & 5.4 & 6.3 & 11.3 \\ 
\cmidrule{2-7}
   & \multirow{4}{*}{\makecell{$N = 50$ \\ $J = 10$}}     & 25 & 10.9 & 6.2 & 13 & 14.2 \\ 
   &  & 50 & 7.8 & 9.1 & 15.4 & 22.2 \\ 
   &  & 100 & 5.5 & 10.9 & 14.4 & 22.9 \\ 
   &  & 200 & 4.1 & 11.4 & 13.4 & 22.5 \\ 
\midrule 
  \multirow{8}{*}{\makecell{$\|w^*\|_1$ = 3}}  &
\multirow{4}{*}{\makecell{$N = 20$ \\ $J = 5$}}   
    & 25 & 3.3 & 3.3 & 6.3 & 10.2 \\ 
   &  & 50 & 1.9 & 4.8 & 6.7 & 11.1 \\ 
   &  & 100 & 1.2 & 5 & 6 & 10.9 \\ 
   &  & 200 & 0.7 & 5.2 & 6.3 & 11.3 \\ 
\cmidrule{2-7}
   & \multirow{4}{*}{\makecell{$N = 50$ \\ $J = 10$}} & 25 & 10.2 & 4.5 & 13 & 14.2 \\ 
   &  & 50 & 8.1 & 9.2 & 15.4 & 22.2 \\ 
   &  & 100 & 5.6 & 11.2 & 14.4 & 22.9 \\ 
   &  & 200 & 3.7 & 10.8 & 13.4 & 22.5 \\ 
\bottomrule
\end{tabular}
\caption{Variable selection performance of \BVS{} and Lasso averaged over 100 replicates.} \label{table:vs}
\end{table}

\begin{figure}[!h]
    \centering
    \includegraphics[width=0.75\linewidth]{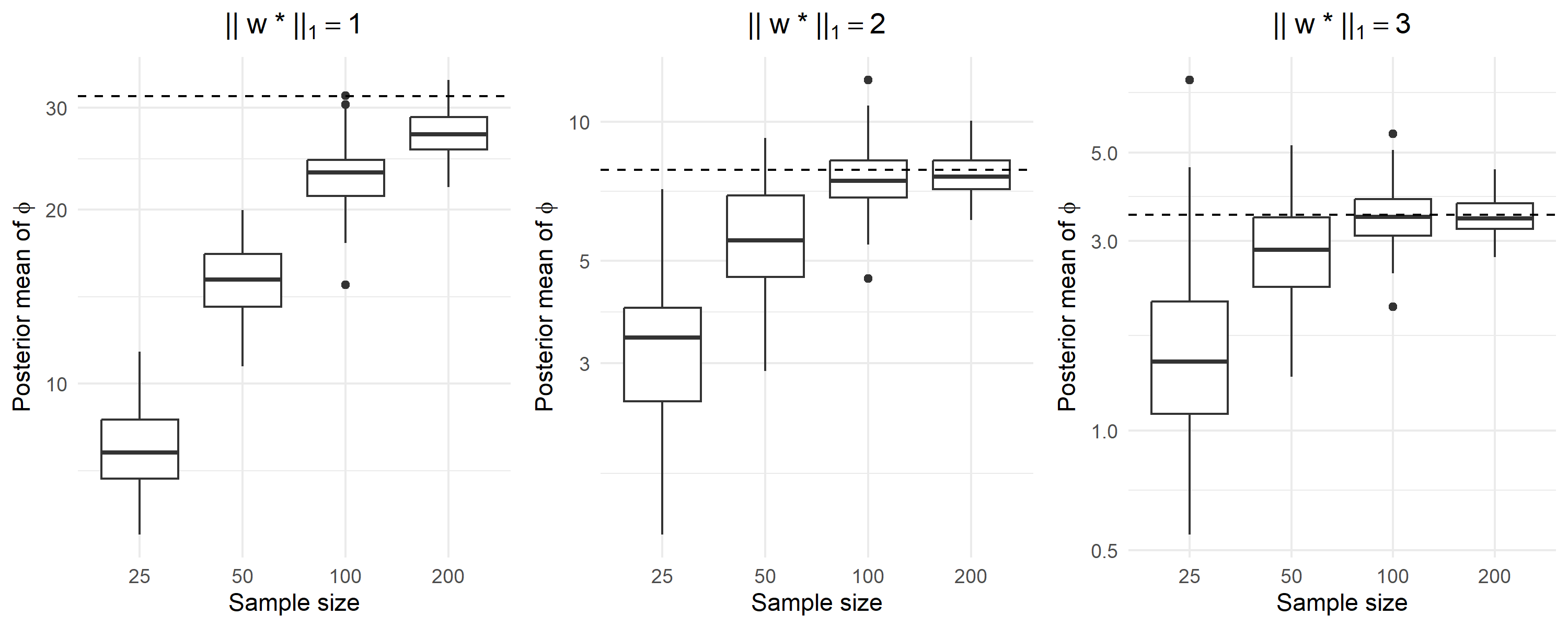}
  \caption{Distribution of the posterior mean of $\phi$ across $100$ replicates with $N = 50, J = 10$. 
  The true value $\phi^*$ is indicated by the dotted line.}  
    \label{fig:phi-50}
\end{figure}

Finally, we examine the posterior mean estimate of $\tau$ when $N = 50$, for which the result is presented in Figure~\ref{fig:tau-50}.  As discussed in Section~\ref{sec:bvs-model}, $\tau$ can be considered as an indicator of whether the simplex constraint is satisfied by the given data, and we now explain why Figure~\ref{fig:tau-50} provides compelling empirical evidence supporting this. 
When $\|w^*\|_1 = 1$, we see that the posterior mean of $\tau$ decreases as $M$ increases, since a larger sample size offers stronger evidence that the simplex constraint is satisfied, which draws $\tau$ towards zero. 
In contrast, when $\|w^*\|_1 = 2$ or $3$, the trend reverses, since a larger sample size enables \BVS{} to detect deviations from the simplex constraint. 
When $M = 200$, we find that the ratio of the posterior mean estimate of $\tau$ and that of $\phi$ equals $0.0003$ for $\|w^*\|_1 = 1$, $0.024$ for $\|w^*\|_1 = 2$, and $0.081$ for $\|w^*\|_1 = 3$. This actually aligns well with the true data-generating mechanism: since we use $J = 10$ for $N = 50$,  we may estimate the ``true'' mean squared deviation of $w^*$ from the simplex constraint by $0.1 \times \sum_{j=1}^{10} (w^*_j - 0.1)^2$, which equals $0.021$ for $\|w^*\|_1 = 2$ and $0.065$ for $\|w^*\|_1 = 3$.

\begin{figure}[!h]
    \centering
    \includegraphics[width=0.75\linewidth]{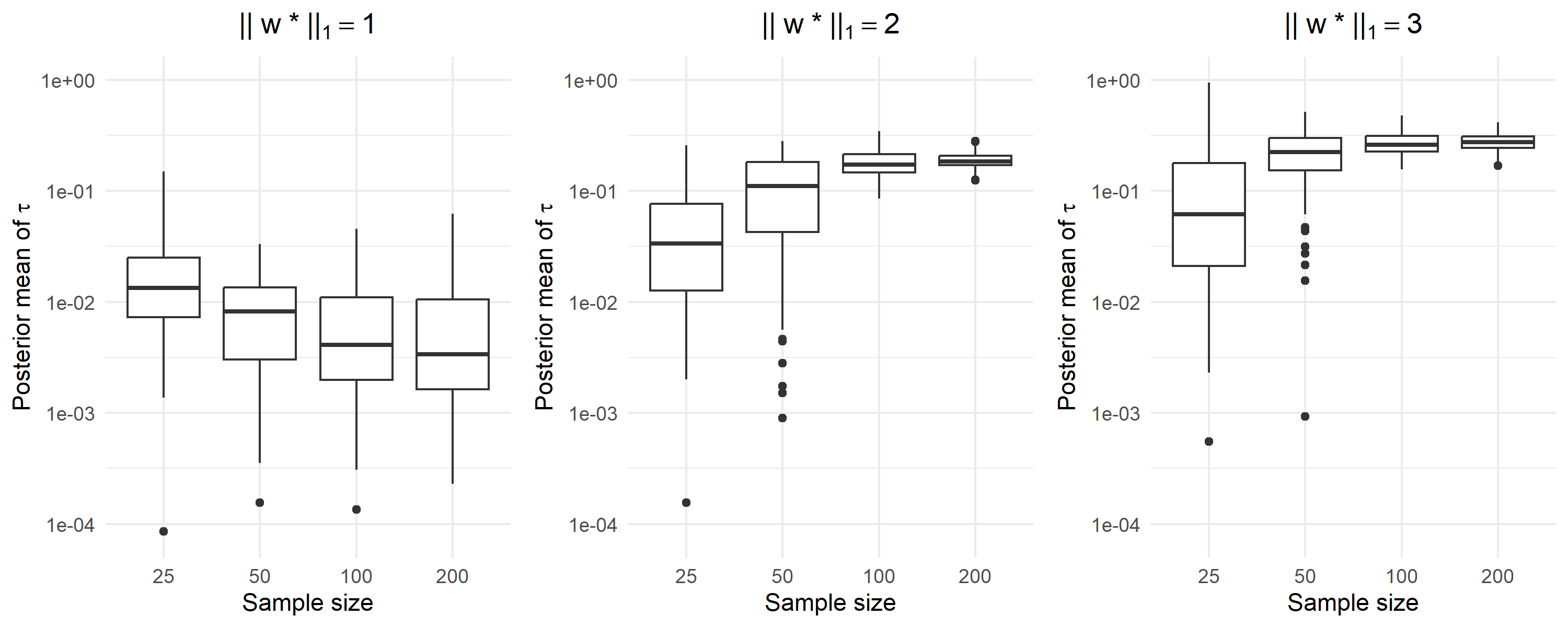}
    \caption{Distribution of the posterior mean of $\tau$ across $100$ replicates  with $N = 50, J = 10$.}
    \label{fig:tau-50}
\end{figure}

\section{Empirical Examples} \label{sec:emp}
In this section, we revisit two empirical examples from the SCM literature. 
We will use the notation introduced in Section~\ref{sec:intro-background}: $T_0$ is the number of pre-treatment time periods,  $T_1$ is the number of post-treatment time periods, the total number of units is $N + 1$ where the first one is used as the response and the other $N$ units as the explanatory variables. 
(For \BVS{}, $M = T_0$ and $\tilde{M} = T_1$.)
The first example has $N < T_0$, while the second has $N \gg T_0$, and  it will be shown that our method works well in both  settings.

\subsection{Nota Fiscal Paulista: Anti-tax Evasion}
An anti-tax evasion program, Nota Fiscal Paulista (NFP), was implemented in S\~{a}o Paulo, Brazil in October 2007 with the goal of reducing tax evasion by incentivizing consumers to request electronic receipts in exchange for potential tax rebates and participation in monthly lotteries promoted by the government. 
However, those restaurant owners and retailers may pass on part of the tax costs to consumers by raising product prices. 
Although the NFP was implemented sequentially across various sectors, including restaurants, bakeries, bars, and food service retailers, the literature has suggested that the food away from home (FAH) index could be used as a suitable indicator for price levels of these sectors. \cite{carvalho2018arco} proposed an artificial counterfactual (ArCo) estimation procedure, by connecting Lasso to SCM without simplex constraint, to explore how the NFP affects the inflation on FAH. They also included monthly GDP growth, retail sales growth and monthly credit growth. Their ArCo estimation indicated that the ATT of NFP is 0.4478\%. 

We obtain the data analyzed in~\cite{carvalho2018arco} from the $\texttt{R}$ package $\texttt{ArCo}$. 
The pre-treatment periods span $T_0 = 33$ months, and the post-treatment periods span  $T_1 = 23$ months. Slightly different from the original study, the data includes only $N = 8$ control units, which is one fewer than that mentioned in their paper, and there is no data related to credit and retail sales.  
We run our Algorithm~\ref{alg:gibbs} for 1000 iterations to estimate the impact of NFP, with the first half of the samples dropped for burn-in. We set the hyperparameters by $\kappa_1 = \kappa_2 = 1, a_1 = 0.01, a_2 = 0.1, \alpha = 1, \theta = 0.25.$ Table \ref{EP1table} summarizes the posterior mean and 95\% credible intervals of ATT, $\tau$, $\phi$ and the model size $|\gamma|$. The ATT estimation is close to that of \citet[Column 1, Table 5]{carvalho2018arco}, but our credible interval is much smaller compared to their confidence set. 
While the analysis of \citet{carvalho2018arco} selected all the predictors to construct the synthetic control, our method only chooses about two predictors on average. 
This is consistent with the findings of our simulation study presented in Section \ref{sec:sim}, which shows that Lasso has the tendency of over-estimating the model size compared to \BVS{}. 
In Appendix~\ref{sec:ep1-more}, we provide the trace plots for $\phi$ and $\log(\tau)$  and visualize the posterior distributions of ATT, $\log(\tau)$, $\phi$, model size $|\gamma|$, and the counterfactual estimate are provided. 

\begin{table}[!h]
\centering
\begin{tabular}{rllll}
  \hline
 & ATT & $\tau$ & $\phi$ & $|\gamma|$ \\ 
  \hline
Mean & 0.288 & 0.14 & 3.85 & 2.29 \\ 
  95\% credible interval & (0.141, 0.419) & (0, 1.44) & (1.94, 6.05) & (1, 4) \\ 
   \hline
\end{tabular}
\caption{NFP Impact on Food Inflation} 
\label{EP1table}
\end{table}

\subsection{China's Anti-corruption Campaign}

\cite{shi2023forward}  investigated how China's unprecedented anti-corruption campaign, launched in November 2012,  affected the importation of luxury watches. This was motivated by the fact that high-end imported watches are often used as discreet bribery gifts in place of blatant cash payments, and empirical studies have shown that the imports of luxury watches  in China move with government leadership changes~\cite{lan2018swiss}. 
\cite{shi2023forward} used data from the United Nations Comtrade Database, taking the monthly growth rate of the category ``watches with case of, or clad with, precious metal''  as the outcome and constructing the synthetic counterfactual from $N = 87$ other commodity categories.  
January 2013 is considered the time of treatment, which was the month immediately following the announcement of the Eight-Point  anti-corruption policy in China. The pre-treatment period spans from February 2010 to December 2012, yielding $T_0 = 35$ observations, and the post-treatment period covers January 2013 to December 2015, yielding $T_1 = 36$  observations. 
Their forward selection algorithm identified three categories as control units: ``knitted or crocheted fabric'', ``cork and articles of cork'', and ``salt, sulfur, earth, stone, plaster, lime, and cement''. 
Their estimated counterfactual predicted that without the anti-corruption campaign, luxury watch imports would have continued to increase by 1.7\% in January 2013, while in reality imports dropped by 42\%. 
Their treatment effect estimation indicated a reduction of 3.09\% in luxury watch imports per month over the post-treatment period.

We obtain the data of~\cite{shi2023forward} from the $\texttt{R}$ package $\texttt{fdPDA}$, which contains all the 88 commodity categories (including the response and $N=87$ control categories), and run our Algorithm~\ref{alg:gibbs} for 1000 iterations with the same  hyperparameters as in the last example. Table~\ref{EP2table} indicates that our ATT estimation (-2.1\%) is smaller than that of \cite{shi2023forward} in absolute value, but the 95\% credible interval indicates that the anti-corruption has a significant effect on the high-end watch imports. 
On average, our model size $|\gamma|$ is larger than that of~\cite{shi2023forward} (who only selected 3 control categories), but Figure~\ref{EP2_hist} shows that the posterior distribution of $|\gamma|$ still concentrates on small models with size $\leq 5$.  Additional simulation results are provided in Appendix~\ref{sec:ep2-more}.

\begin{table}[h]
\centering

\begin{tabular}{rllll}
  \hline
 & ATT & $\tau$ & $\phi$ & $|\gamma|$ \\ 
  \hline
Mean & -0.021 & 0.069 & 20.86 & 5.09 \\ 
  95\% credible interval & (-0.032, -0.008) & (0, 0.641) & (12.22, 32.76) & (1, 20) \\ 
   \hline
\end{tabular}
\caption{ Anti-corruption Campaign's Impact on Luxury Watches Importation} 
\label{EP2table}
\end{table}

It is worth noting that in our posterior samples, the categories ``cork'' and ``salt'' appear with negligible probabilities. Instead, the most frequently selected categories are those related to textiles---for example, ``knitted or crocheted fabrics'' (51.4\%), ``special woven fabrics such as lace, tapestries, trimmings, and embroidery'' (16.6\%), and ``wool, fine or coarse animal hair'' (12\%). In addition, the category ``optical, photographic, cinematographic, measuring, checking, medical or surgical instruments and apparatus; parts and accessories'' also appears with a relatively high probability (10.2\%). Compared with ``cork'' and ``salt'', these categories seem more plausible to serve as control units. The textile products share decorative features similar to those of wristwatches, while the optical and photographic instruments also rely on precision mechanical engineering.

\section{Conclusion and Future Directions} \label{sec:conclusion}
We propose a novel Bayesian synthetic control method, \BVS{}, that integrates a relaxed simplex constraint with Bayesian spike-and-slab variable selection.
Our approach introduces a hierarchical prior to determine how closely the data should adhere to the simplex constraint, a feature that is missing in existing frequentist or Bayesian synthetic control frameworks and helps mitigate the ongoing debate surrounding the use of simplex constraints. 
To efficiently compute the posterior distribution of \BVS{}, we develop a novel Metropolis-within-Gibbs sampler, which overcomes the hindrance caused by the simplex constraint by updating two regression coefficients simultaneously from the full conditional posterior. We also show that in high-dimensional settings, 
when the true DGP satisfies the simplex constraint, \BVS{} with $\tau = 0$ consistently selects the true control units and achieves consistency in terms of both  regression coefficient estimation and prediction. Simulation studies and real-data examples illustrate the advantages and usefulness of our proposed algorithm.

There are several interesting directions that future researchers may dive in. First, it is known that Gibbs sampling methods are susceptible to high collinearity among the explanatory variables, which is a common situation in SCM applications where factor models are often used to model the DGP.  
Though it is unclear if such issues affect the methodology proposed in this work (as our algorithm updates two coordinates at a time), it would be interesting to investigate the use of other MCMC techniques, such as pseudo-marginal sampling~\citep{andrieu2009pseudo}.  
Second, it will be interesting to further generalize \BVS{} to account for time-invariant characteristics, which are often used together with time-variant predictors in the SCM literature for constructing  counterfactual outcomes.   
Third, more numerical studies or model extensions can be performed to study the performance of \BVS{}  when the model is mis-specified (e.g., the data is non-stationary or contains heterogeneous errors), which would be of great interest to economists.  
We hope our method will serve as an inspiration for future research and facilitate the widespread use of Bayesian synthetic control methods.

\section*{Acknowledgements}
Y.X. is grateful to Qi Li, Yonghong An and Angda Li for their helpful comments and suggestions. This work was supported in part by NSF via grants DMS-2311307 and DMS-2245591.

\clearpage 
\newpage

\renewcommand{\thesection}{S\arabic{section}}
\renewcommand{\theequation}{S\arabic{equation}}
\renewcommand{\thetheorem}{S\arabic{theorem}}
\renewcommand{\thelemma}{S\arabic{lemma}}
\renewcommand{\theproposition}{S\arabic{proposition}}
\renewcommand{\thecorollary}{S\arabic{corollary}}
\renewcommand{\theassumption}{S\arabic{assumption}}
\renewcommand{\thefigure}{S\arabic{figure}}
\renewcommand{\thetable}{S\arabic{table}}

\setcounter{section}{0}
\setcounter{equation}{0}
\setcounter{theorem}{0}
\setcounter{lemma}{0}
\setcounter{proposition}{0}
\setcounter{corollary}{0}
\setcounter{assumption}{0}
\setcounter{figure}{0}
\setcounter{table}{0}

\phantomsection\label{APPENDIX}
\bigskip

\begin{center}

{\large\bf APPENDIX}

\end{center}

\section{Details of the Metropolis-within-Gibbs Sampler} \label{sec:gibbs}
For Bayesian variable selection involving a prior on the indicator vector $\gamma$ as we have considered, the standard approach to posterior sampling is to integrate out the regression coefficient vector and use add-delete-swap proposals to construct a Metropolis--Hastings algorithm targeting the marginal posterior distribution of $\gamma$ (or the joint distribution of $\gamma$ and some variance parameter); see, e.g.,~\citet{ george1997approaches, chipman2001practical, guan2011bayesian, zhou2022dimension}. 
However, this approach is not applicable to our model, since $\mu_\gamma$ cannot be integrated out in closed form from the joint posterior distribution of $(\gamma, \mu_\gamma, \tau, \phi)$. 
To tackle this very unique challenge, we propose a novel Metropolis-within-Gibbs sampler targeting the joint posterior distribution of $(\mu, \tau, \phi)$. 
Note that since $\gamma$ can be viewed as a function of $\mu$ with $\gamma_i = \ind_{ \{\mu_i \neq 0\}}$,  henceforth we will denote the conditional distribution in~\eqref{eq:cond-Y} by $p(y \mid \mu, \tau, \phi)$, and $\gamma$ should be understood as $\gamma = \gamma(\mu)$. 

\subsection{Updating of \texorpdfstring{$\tau$ and $\phi$}{tau and phi}}\label{sec:gibbs-var}
The updates for $\tau$ and $\phi$ are standard.
We update $\phi$ by drawing it from the full conditional distribution given in~\eqref{eq:full-cond-phi}, 
and we use a Metropolis--Hastings update for $\tau$ that is invariant with respect to $p(\tau \mid y, \mu, \phi)$. Explicitly, we propose a new value for $\tau^*$ by a Gaussian random walk on log-scale with variance $\eta > 0$ and accept $\tau^*$ with probability 
\begin{equation}\label{eq:acc-tau}
    \rho(\tau, \tau^*) = \min\left\{ 1,  \frac{ p(y \mid   \mu, \tau^*, \phi) p(\tau^*)}{p(y \mid  \mu, \tau, \phi) p(\tau) } \frac{  \tau^* }{ \tau }   \right\}, 
\end{equation}
where the likelihood is given in~\eqref{eq:like1} and $p(\tau)$ denotes the prior density.  In Algorithm~\ref{alg:gibbs}, we repeat this step $n_\tau$ times so that  $\tau$ is more likely to be updated to some value with higher conditional posterior density given $\mu$ and $\phi$. 

\subsection{Updating of \texorpdfstring{$\mu$}{mu}}\label{sec:gibbs-mu}

For $\mu$, we propose a Gibbs-type updating scheme, which is quite different from any existing Gibbs scheme  for similar problems.  
Let $\mu_{-j}$ denote the subvector of $\mu$ with $\mu_j$ removed. 
Since the prior of \BVS{} assigns probability one to $\mu \in \Delta^{N - 1}$, 
given $\mu_{-j}$, we have $\mu_j = 1 - \sum_{i \neq j} \mu_i$ almost surely, and thus  the full conditional distribution $p(\mu_j \mid y, \mu_{-j}, \tau, \phi)$ is also degenerate. 
Hence, a Gibbs scheme that updates each $\mu_j$ from its full conditional posterior will not move at all. 
This motivates us to devise a more complicated updating scheme by updating $\mu_i, \mu_j$ (with $i\neq j$) simultaneously from the full conditional posterior $p(\mu_i, \mu_j \mid y,  \mu_{-(i, j)}, \tau, \phi)$, where $\mu_{-(i, j)}$ denotes the subvector of $\mu$ with the $\mu_i, \mu_j$ removed.  

Let $\mu_{-(i, j)}, \tau, \phi$ be given and define $s = 1 - \sum_{k \neq i, j} \mu_k$. 
If $s = 0$, the simplex constraint implies that $\mu_i = \mu_j = 0$ almost surely. 
If $s \in (0, 1]$, the conditional posterior distribution of $\mu_i, \mu_j$ falls into one of the following three cases:
\begin{enumerate}[label=(\roman*)]
    \item given $\mu_{-(i, j)},  \gamma_i = 1, \gamma_j = 0$, we have $\mu_i = s$ and $\mu_j = 0$ a.s.; 
    \item given $\mu_{-(i, j)},  \gamma_i = 0, \gamma_j = 1$, we have $\mu_i = 0$ and $\mu_j = s$ a.s.;
    \item given $\mu_{-(i, j)},  \gamma_i = 1, \gamma_j = 1$, we have $\mu_i = u$ and $\mu_j = s - u$ for some $u \in (0, s)$ a.s.  
\end{enumerate} 
Therefore, to find a closed-form expression for this conditional distribution, it suffices to find 
$p(\gamma_i, \gamma_j \mid y,  \mu_{-(i, j)}, \tau, \phi)$  and the conditional distribution  $p(\mu_i \mid y,  \mu_{-(i, j)}, \tau, \phi, \gamma_i = \gamma_j = 1)$. 

Let $\gamma_{-(i, j)} = \{ k \notin \{i, j\} \colon \mu_k \neq 0\} \subset [N]$ denote the model corresponding to $\mu_{- (i, j)}$.   
We can define the four possible models under consideration by 
\begin{equation}\label{eq:def-gamma-ij} 
\gamma^0 = \gamma_{-(i, j)}, \quad 
\gamma^i = \gamma_{-(i, j)} \cup \{i\}, \quad 
\gamma^j = \gamma_{-(i, j)} \cup \{j\}, \quad 
\gamma^{ij} = \gamma_{-(i, j)} \cup \{i, j\}.  
\end{equation} 
We will also interpret $\gamma^0, \gamma^i, \gamma^j, \gamma^{ij}$  as indicator vectors in $\{0, 1\}^N$. 
We first find the full conditional distribution of $(\mu_i, \mu_j)$ given other parameters and $\gamma_i = \gamma_j = 1$  in the following lemma.  
For simplicity, we assume $\alpha = 1$ in this section, and in Remark~\ref{rmk:alpha} we explain how to extend the results to other integer values of $\alpha$.  

\begin{lemma}\label{lm1}
Assume $\alpha = 1$. 
Fix $i \neq j$ such that $s = 1 - \sum_{k \neq i, j} \mu_k \in (0, 1]$. 
The conditional distribution of  $(\mu_i, \mu_j) \mid y,  \mu_{-(i, j)}, \tau, \phi, \gamma_i = \gamma_j = 1$  is degenerate with $\mu_j = s - \mu_i$ and 
\begin{align}
    &\mu_i \mid y,  \mu_{-(i, j)}, \tau, \phi, \gamma_i = \gamma_j = 1
    \sim N_{ (0, s) }\left(  \beta_{i, j}, \;   \frac{1 }{ \phi \Lambda_{i, j}}  \right).  
\end{align}
In the above expression, $N_{(a, b)}$ denotes the univariate truncated normal distribution restricted to the interval $(a, b)$,
\begin{align}
\beta_{i, j} =\;&  \frac{1}{  \Lambda_{i, j} }  (X_i - X_j)^\top  \Sigma_{\gamma^{ij},  \tau} \left(  y -  s X_j - \sum\nolimits_{k \neq i, j}\mu_k  X_k \right),  \label{eq:def-beta-ij} \\ 
    \Lambda_{i, j} =\;&  (X_j - X_i)^\top  \Sigma_{\gamma^{ij},  \tau} (X_j - X_i),   \label{eq:def-Lambda-ij}  
\end{align}
where $\gamma^{ij}$ is given by~\eqref{eq:def-gamma-ij} and  $ \Sigma_{ \gamma, \tau }$ is given by~\eqref{eq:Sigma}.  
\end{lemma}

\begin{proof}
We can rewrite the conditional posterior of $(\mu_i, \mu_j)$ as 
\begin{align}
    & \quad p(\mu_i, \mu_j \mid y, \phi, \tau, \mu_{-(i, j)}, \gamma_i = \gamma_j = 1) \\
    &\propto p(y \mid  \phi, \tau, \mu_i, \mu_j, \mu_{-(i, j)}  ) \, p(\mu_i, \mu_j \mid \mu_{-(i, j)}, \gamma_i = \gamma_j = 1 ), 
\end{align} 
where  the marginal likelihood $p(y \mid  \phi, \tau, \mu_i, \mu_j, \mu_{-(i, j)} )$ is shown in \eqref{eq:like1}.    Since $s = 1 - \sum_{k \neq i, j} \mu_k$ and the symmetric Dirichlet distribution becomes a uniform distribution on the simplex  $\Delta^{|\gamma|-1}$ when $\alpha = 1$, we have 
\begin{align} 
    p(\mu_i = u, \mu_j = s - u \mid  \mu_{-(i, j)}, \gamma_i = \gamma_j = 1 ) =\;&  ( |\gamma^{i j} | - 1 )!, \quad \forall u \in (0, s). 
\end{align}  
Hence, omitting all constant terms that do not depend on $(\mu_i, \mu_j)$, we find that  the conditional posterior density  of $(\mu_i, \mu_j)$ is 
\begin{align} 
p(\mu_i = u, \mu_j = s - u \mid y, \phi, \tau, \mu_{-(i, j)}) \propto    \exp\left\{ -\frac{\phi}{2}  \check{y}(u)^\top  \Sigma_{\gamma^{i, j},  \tau} \check{y}(u)   \right\}. 
\end{align}
where 
\begin{equation}\label{eq:checky-def}
\check{y}(u) =  y - u X_i - (s - u) X_j -  \sum\nolimits_{k \neq i, j}\mu_k  X_k. 
\end{equation} 
Using the notation $\beta_{i, j}$ and $\Lambda_{i, j}$ defined in~\eqref{eq:def-beta-ij} and~\eqref{eq:def-Lambda-ij}, we obtain that  
\begin{align*}
    \check{y}(u)^\top  \Sigma_{\gamma^{i, j},  \tau} \check{y}(u) 
    &=  \Lambda_{i,j} u^2 - 2 \Lambda_{i, j} \beta_{i, j} u + \tilde{C} 
\end{align*}
for some constant $\tilde{C} $ that does not depend on $(\mu_i, \mu_j)$. The claim thus follows.   
\end{proof}

Next, we find the conditional posterior probabilities of $\gamma^0, \gamma^i, \gamma^j$, and $\gamma^{ij}$. 

\begin{lemma}\label{lm2} 
Fix $i \neq j$ such that $s = 1 - \sum_{k \neq i, j} \mu_k \in (0, 1]$.  
Let  $\check{y}(u)$ be given by~\eqref{eq:checky-def}. The conditional posterior probabilities of  $\gamma^0, \gamma^i, \gamma^j, \gamma^{ij}$ defined by~\eqref{eq:def-gamma-ij} are  given by 
\begin{align}
p(\gamma^0 \mid y,  \mu_{-(i, j)}, \tau, \phi)  =\;& 0,  \\
p(\gamma^i \mid y,  \mu_{-(i, j)}, \tau, \phi)  =\;&  
C \, s^{\alpha - 1}  A_{\alpha}(\gamma^i, \tau) 
 \exp\left\{  -\frac{\phi}{2}   \check{y}(s)^\top  \Sigma_{ \gamma^i, \tau}  \check{y}(s)  \right\},  \\
p(\gamma^j \mid y,  \mu_{-(i, j)}, \tau, \phi)  =\;& 
C \, s^{\alpha - 1}  A_{\alpha}(\gamma^j, \tau) 
 \exp\left\{ -\frac{\phi}{2}  \check{y}(0)^\top  \Sigma_{ \gamma^j, \tau} \check{y}(0)   \right\},  \\
 p(\gamma^{ij} \mid y,  \mu_{-(i, j)}, \tau, \phi)  =\;&   C \,  A_{\alpha}(\gamma^{ij}, \tau) \int_0^s  u^{\alpha - 1} (s - u)^{\alpha - 1}
 \exp\left\{ -\frac{\phi}{2}   \check{y}(u)^\top  \Sigma_{ \gamma^{ij}, \tau}  \check{y}(u)   \right\} \d u, 
\end{align}
where $C > 0$ is a constant   and 
\begin{equation}  
    A_{\alpha} (\gamma, \tau) =  \frac{ \Gamma( |\gamma| \alpha) }{ \Gamma(\alpha)^{|\gamma|} } \, \tau^{-|\gamma| / 2}  \mathrm{det} (V_{\gamma, \tau})^{-1/2} \, \theta^{|\gamma|} (1 - \theta)^{N - |\gamma|}. 
\end{equation} 
In particular, when $\alpha = 1$, 
\begin{align}
     p(\gamma^{ij} \mid y,  \mu_{-(i, j)}, \tau, \phi)  
    & =      \frac{ C 
    A_{\alpha}(\gamma^{ij}, \tau) \sqrt{2 \pi} }{  \sqrt{ \phi \Lambda_{i, j} } } 
   \exp\left\{ -\frac{\phi}{2} \left\{  \check{y}(0)^\top  \Sigma_{\gamma^{ij},  \tau}   \check{y}(0) -   \beta_{i, j}^2 \Lambda_{i, j} \right\} \right\} \\
   &\quad \times 
    \left\{ \Phi \left( (s -\beta_{i, j} )\sqrt{\phi \Lambda_{i, j}}   \right) -  
     \Phi \left(   -\beta_{i, j}  \sqrt{\phi \Lambda_{i, j}}   \right) \right\}, 
\end{align}
where $\beta_{i, j}, \Lambda_{i, j}$ are as given in Lemma~\ref{lm1} and $\Phi$ denotes the CDF of $N(0, 1)$. 
\end{lemma}

\begin{proof}
We can express the conditional posterior of $(\gamma_i, \gamma_j, \mu_i, \mu_j)$ as 
\begin{align}
     & \quad p(\gamma_i, \gamma_j, \mu_i, \mu_j \mid y, \phi, \tau, \mu_{-(i, j)})\\ 
    & = C'
    p(y \mid  \phi, \tau, \mu_i, \mu_j, \mu_{-(i, j)} ) \, p(\gamma_i, \gamma_j,  \mu_i, \mu_j \mid \mu_{-(i, j)}) \, p( \mu_{-(i, j)} ) \\
    & = C' p(y \mid  \phi, \tau, \mu_i, \mu_j, \mu_{-(i, j)} ) \, p(\gamma, \mu), \\
     & = C' p(y \mid  \phi, \tau, \mu_i, \mu_j, \mu_{-(i, j)} )\, p(\mu \mid \gamma) \, p(\gamma),
\end{align} 
where $C'$ denotes some unknown constant that does not depend on $(\mu_i, \mu_j)$,  and in the second step we have used that $\gamma$ is determined by $\mu$. Note that $p(\gamma, \mu)$ can only be nonzero when $\gamma \in \{ \gamma^i, \gamma^j, \gamma^{ij} \} $.  
In particular, $\gamma^0$ has zero conditional prior probability given $\mu_{-(i, j)}$ since it would violate the simplex constraint.  By~\eqref{eq:modelprior}, $p(\gamma  ) = \theta^{|\gamma|} (1 - \theta)^{N - |\gamma|}$ and 
\begin{align*}
p(\mu \mid \gamma)   =  \frac{\Gamma(|\gamma| \alpha) }{ \Gamma(\alpha)^{|\gamma|}} \prod_{k \in \gamma} \mu_k^{\alpha - 1}. 
\end{align*}
Given  $ \mu_{-(i, j)}$ and $(\gamma_i, \gamma_j) = (1, 0)$,  we have $(\mu_i, \mu_j) = (s, 0)$ with probability one.  
Hence, using the likelihood given by~\eqref{eq:like1} and the expression for $A_\alpha(\gamma, \tau)$, we get 
\begin{align*}
  & \quad p(\gamma_i = 1, \gamma_j = 0  \mid y, \phi, \tau, \mu_{-(i, j)}) \\
  & = C'        p(y \mid  \phi, \tau, \mu_i = s, \mu_j = 0, \mu_{-(i, j)} )\, p(\mu_i = s, \mu_j = 0, \mu_{-(i, j)} \mid \gamma^i) \, p(\gamma^i) \\
  & = C s^{\alpha - 1}  A_\alpha(\gamma^i, \tau)  \exp\left\{  -\frac{\phi}{2}   \check{y}(s)^\top  \Sigma_{ \gamma^i, \tau}  \check{y}(s)  \right\}, \\
&\text{ where } C  = C' \phi^{M/2}  \prod_{k \in \gamma \setminus \{i, j\}} \mu_k^{\alpha - 1}.
\end{align*} 
 This proves the claimed expression for $p(\gamma^i  \mid y, \phi, \tau, \mu_{-(i, j)})$. An analogous calculation yields the claimed expression for $p(\gamma^j  \mid y, \phi, \tau, \mu_{-(i, j)})$.
  
To find the conditional posterior probability of $\gamma^{ij}$, we use the conditional distribution of $\mu_i$ given $\gamma_i = \gamma_j = 1$ in Lemma~\ref{lm1} to get 
\begin{align}
 & \quad    p(\gamma_i = 1, \gamma_j = 1 \mid y, \phi, \tau, \mu_{-(i, j)}) \\ 
&= \int_0^s p(\gamma_i = 1, \gamma_j = 1, \mu_i = u,  \mu_j = s - u \mid y, \phi, \tau, \mu_{-(i, j)}) \d u  \\
& = C p(\gamma^{ij}) \int_0^s p(y \mid  \phi, \tau, \mu_i = s, \mu_j = s - u, \mu_{-(i, j)} )\, p(\mu_i = u, \mu_j = s - u, \mu_{-(i, j)} \mid \gamma^{ij})   \d u  \\
& = C'  A_\alpha(\gamma^{ij}, \tau)    \int_0^s \exp\left\{  -\frac{\phi}{2}   \check{y}(u)^\top  \Sigma_{ \gamma^i, \tau}  \check{y}(u)  \right\}  u^{\alpha - 1} (s - u)^{\alpha - 1} \d u,  \label{eq:gamma-ij}
\end{align}
which proves the claimed expression for $p(\gamma^{ij}  \mid y, \phi, \tau, \mu_{-(i, j)})$. 

When $\alpha = 1$, a routine calculation shows that the integral in~\eqref{eq:gamma-ij} can be computed as 
\begin{align*}
  &\quad  \int_0^s \exp\left\{  -\frac{\phi}{2}   \check{y}(u)^\top  \Sigma_{ \gamma^i, \tau}  \check{y}(u)  \right\}    \d u \\
& =   \exp\left\{ -\frac{\phi}{2} \left( \check{y}(0)^\top  \Sigma_{\gamma^{i, j},  \tau}   \check{y}(0)^\top -   \beta_{i, j}^2 \Lambda_{i, j} \right) \right\}  
  \frac{\sqrt{2 \pi}}{ \sqrt{ \phi \Lambda_{i, j} } } \int_0^s  \frac{ \sqrt{ \phi \Lambda_{i, j} } }{\sqrt{2 \pi}}  \exp\left\{ -\frac{\phi \Lambda_{i, j} }{2} 
\left( u -  \beta_{i, j} \right)^2 \right\} \d u, 
\end{align*}
where the integral on the right-hand side involves the Gaussian density with mean $\beta_{i, j}$ and variance $(\phi \Lambda_{i, j})^{-1}$. Hence,  
\begin{align}
     \int_0^s  \frac{ \sqrt{ \phi \Lambda_{i, j} } }{\sqrt{2 \pi}}  \exp\left\{ -\frac{\phi \Lambda_{i, j} }{2} 
\left( u -  \beta_{i, j} \right)^2 \right\} \d u =\; \Phi \left( (s -\beta_{i, j} )\sqrt{\phi \Lambda_{i, j}}   \right) -  
     \Phi \left(   -\beta_{i, j}  \sqrt{\phi \Lambda_{i, j}}   \right) , 
\end{align}
proving the claim. 
\end{proof}

\begin{remark}\label{rmk:alpha}
For any positive integer $\alpha$, the integral involved in $ p(\gamma^{ij} \mid y,  \mu_{-(i, j)}, \tau, \phi)$ can always be expressed by using the PDF and CDF of the standard normal distribution,
and thus the  conditional posterior probabilities of $\gamma^i, \gamma^j, \gamma^{ij}$ can be computed very efficiently. 
For non-integer-valued $\alpha$, numerical integration is needed. 
The conditional posterior distribution of $(\mu_i, \mu_j)$ given $\gamma^{ij}$ is 
$$p(\mu_i = u, \mu_j = s - u \mid y,  \mu_{-(i, j)}, \tau, \phi, \gamma_i = \gamma_j = 1) \, \propto \,
u^{\alpha - 1} (s - u)^{\alpha - 1}
 e^{ -\frac{\phi}{2}   \check{y}(u)^\top  \Sigma_{ \gamma^{ij}, \tau}  \check{y}(u)  },$$
for $u \in (0, s)$.  When $\alpha = 1$,  $u$ follows the truncated normal distribution, as shown in Lemma~\ref{lm1}, and this allows for straightforward sampling of $(\mu_i, \mu_j)$.
For other integer values of $\alpha$,  one can use rejection sampling to generate samples of $u$, where the truncated normal distribution can be used as a reference distribution. 
\end{remark}

\subsection{Algorithm} \label{sec:gibbs-alg}
The pseudocode for our proposed   Metropolis-within-Gibbs sampler targeting $p(\mu, \tau, \phi \mid y)$  is given in Algorithm~\ref{alg:gibbs} in the main text.  
The most computationally intensive part of Algorithm~\ref{alg:gibbs} is the evaluation of expressions of the form $V_{\gamma, \tau}^{-1} z$ for a given vector $z$, which may seem highly time-consuming  if $|\gamma|$ is large. 
But this can be implemented efficiently by computing and updating the Cholesky decomposition of $V_{\gamma, \tau}$, a technique commonly employed in Bayesian spike-and-slab variable selection~\citep{smith1996nonparametric, george1997approaches}. 
More specifically, at the beginning of the $t$-th iteration, we compute the Cholesky decomposition of $V_{\gamma, \tau^{(t-1)}}$, where $\gamma$ is the model selected in the last iteration. Then, each time we draw $(\mu_i, \mu_j)$ from its full conditional posterior, we need to change one or two coordinates of $\gamma$, and the resulting Cholesky decomposition of $V_{\gamma, \tau^{(t-1)}}$ can be obtained by standard updating algorithms~\citep{golub2013matrix}.   

Certain variations of Algorithm~\ref{alg:gibbs} are also straightforward to implement.
For example, one can replace the fixed-scan Gibbs updating of $\mu$ in Algorithm~\ref{alg:gibbs} by a random-scan update: in the $t$-th iteration, we uniformly sample a pair $i < j$ such that $\mu^{(t-1)}_i + \mu^{(t-1)}_j > 0$ and update $(\mu^{(t)}_i, \mu^{(t)}_j)$ from its conditional posterior distribution (an acceptance-rejection step is needed to correct for the proposal bias), and then draw $\phi^{(t)}$ from the full conditional posterior and update $\tau^{(t)}$ by one Metropolis--Hastings step. 
In this case, the iterative complex factorization algorithm proposed by~\citet{zhou2019fast}  can be used to efficiently solve the linear system $V_{\gamma, \tau}^{-1} z$.

\begin{remark}\label{rmk:z-mean} 
Consider the application of \BVS{} to SCM. Let $(\Xnew, \Ynew^{(1)})$ denote another data set with $\tilde{M}$ observations under the treatment. 
Since Algorithm~\ref{alg:gibbs} outputs samples $(\mu^{(t)}, \tau^{(t)}, \phi^{(t)}, w^{(t)})_{t=1}^n$ from the joint posterior distribution, we can simply compute an ATT estimate for each sample by 
\begin{equation*}
    \Ynew^{(0, t)} = \Xnew w^{(t)}, \quad \widehat{\mathrm{ATT}}^{(t)} = \frac{1}{\tilde{M}} \sum_{i = 1}^{\tilde{M}} (\Ynew_i^{(1)} - \Ynew_i^{(0, t)} ). 
\end{equation*}
Then, $(\widehat{\mathrm{ATT}}^{(t)})_{t=1}^n$ approximates the posterior distribution of the ATT estimate, and their average yields the final ATT estimator. 
However, if the objective is to compute the posterior mean estimate for ATT, one can further reduce the variance of the estimator by integrating out $w^{(t)}$   as in~\eqref{eq:ATT-mean-estimate} and compute the counterfactual $ \Ynew^{(0, t)} $ by 
\begin{equation*}
     \Ynew^{(0, t)} = \bbE_y[   \Xnew w \mid \gamma^{(t)}, \mu^{(t)}, \tau^{(t)}, \phi^{(t)} ]  = \Xnew V_{\gamma^{(t)}, \tau^{(t)}}^{-1} \left( X^\top_{\gamma^{(t)} } y  + ( 1 / \tau^{(t)} ) \mu^{(t)}_{\gamma^{(t)} } \right). 
\end{equation*} 
\end{remark}

\newpage

\section{Constrained Least-squares Estimators} \label{sec:constrained-ls}
In this section, we review some known results about the least-squares estimator under a linear constraint~\citep{amemiya1985advanced}. We present the results in the general form although for our analysis,  we only need to consider the constraint $\mathsf{1}^\top \mu = 1$, which is part of the simplex constraint. 

\begin{theorem}\label{lm:constrained-ls}
Let $y \in \bbR^n$, $Z \in \bbR^{n \times p}$, $Q \in \bbR^{p \times k}$ for some $k < p$, and $v \in \bbR^k$. 
Assume that $Z^\top Z$ is invertible. 
Let $\mathcal{U} = \{ \mu \in \bbR^p \colon   Q^\top \mu = v \}$. 
Define the unconstrained least-squares estimator $\hat{\mu}$ and the linearly constrained least-squares estimator $\check{\mu} \in \mathcal{U}$ by 
\begin{align*}
\hat{\mu}  = \argmin_{\mu \in \bbR^p}  \| y - Z \mu \|^2_2, \quad 
\check{\mu}  = \argmin_{\mu \in \mathcal{U} }  \| y - Z \mu \|^2_2. 
\end{align*}  
Let $R \in \bbR^{ p \times (p - k)}$ be any matrix such that $R^\top Q = 0$ and $[Q \; R]$ is invertible. Then, the following results hold. 

\renewcommand{\theenumi}{(\roman{enumi})}
\renewcommand{\labelenumi}{\theenumi}

\begin{enumerate}
    \item $\check{\mu}$ can be expressed by 
    \begin{align}
    \check{\mu} &= \hat{\mu} - (Z^\top Z)^{-1} Q \left[ Q^\top (Z^\top Z)^{-1} Q \right]^{-1} (Q^\top \hat{\mu} - v) \\ 
    & = R (R^\top Z^\top Z R)^{-1} R^\top Z^\top y + \left[ I - R(R^\top Z^\top Z R)^{-1} R^\top Z^\top Z  \right] Q (Q^\top Q)^{-1} v. 
    \end{align} 
    \item For any $\mu \in \mathcal{U}$, we have $\| Z (\hat{\mu} - \mu) \|^2 = \| Z (\hat{\mu} - \check{\mu}) \|^2 + \| Z (\check{\mu} - \mu) \|^2.$ Hence, 
    $$\| y - Z \mu \|^2_2 = \| y - Z \hat{\mu} \|^2_2 + \| Z (\hat{\mu} - \check{\mu}) \|^2_2 + \| Z (\check{\mu} - \mu) \|^2_2.$$
    \item If $y = Z \mu^* + \epsilon$ for some $\mu^* \in \mathcal{U}$, then 
    \begin{equation}
        \check{\mu} = \mu^* + R (R^\top Z^\top Z R)^{-1} R^\top Z^\top \epsilon. 
    \end{equation}
\end{enumerate}
\end{theorem}

\begin{proof}
The claimed expressions for $\check{\mu}$ given in part (i) are proved in Chapters 1.4.1 and 1.4.2 of~\citet{amemiya1985advanced}. Part (iii) is proved in Chapter 1.4.3 of~\citet{amemiya1985advanced}.  
To prove part (ii), we use part (i) to get that,  for any $w$ such that $Q^\top w = 0$, 
\begin{align*}
    w^\top Z^\top Z (\check{\mu} - \hat{\mu}) = - w^\top Q \left[ Q^\top (Z^\top Z)^{-1} Q \right]^{-1} (Q^\top \hat{\mu} - v)  = 0. 
\end{align*}
Since $Q^\top ( \check{\mu} - \mu ) = v - v = 0$ for any $\mu \in \mathcal{U}$, this implies that $\| Z (\hat{\mu} - \mu) \|^2 = \| Z (\hat{\mu} - \check{\mu}) \|^2 + \| Z (\check{\mu} - \mu) \|^2$. The claimed decomposition for $\| y - Z \mu\|^2_2$ then follows from that $Z \hat{\mu}$ is the orthogonal projection of $y$ onto the column space of $Z$. 
\end{proof}

Next, we prove some linear algebra results that will be useful for analyzing the constrained least-squares estimator. 

\begin{lemma} \label{lemma716} 
 Let $Q \in \bbR^{p \times k}$ and $R \in \bbR^{ p \times (p - k)}$ be such that $R^\top Q = 0$ and $[Q \; R]$ is invertible. For any positive definite matrix $A \in \bbR^{p \times p}$, 
    \begin{align} 
        Q(Q^\top A^{-1} Q)^{-1}Q^\top &= A - A R(R^\top A R)^{-1} R^\top A, \label{lemma716eq} \\
        \lambda_{\max}(  R (R^\top A R)^{-1} R^\top ) &\leq \lambda_{\max} (A^{-1}).  \label{lemma716-det}
    \end{align}
\end{lemma}

\begin{proof}
The first identity is well known in the statistical literature; see, e.g.,~\cite{smyth1996conditional}.  For completeness, we provide a proof here.  Since $A$ is positive definite, it has a real-valued square root $A^{1/2}$. 
Consider the block matrix $B = [ A^{-1/2} Q \quad A^{1/2} R ] \in \bbR^{p \times p}$, where the two component matrices are orthogonal due to the assumption $R^\top Q = 0$. Hence,  the orthogonal projection on the column space of $B$ can be written as $ P_1 + P_2$ where 
\begin{align*}
    P_1 &= A^{-1/2} Q \left(Q^\top A^{-1} Q\right)^{-1} Q A^{-1/2}, \\
    P_2 &= A^{1/2} R \left(R^\top A  R\right)^{-1} R A^{ 1/2}. 
\end{align*}
Since $\mathrm{rank}(B) = \mathrm{rank}(A^{-1/2} Q ) + \mathrm{rank}(A^{1/2} R )$ and $[Q \; R]$ is invertible, the matrix $B$ also has full rank. It follows that $P_1 + P_2 = I$, which yields~\eqref{lemma716eq}.   

To prove~\eqref{lemma716-det}, it suffices to note that for any vector $z \in \bbR^p$, 
\begin{align*}
    z^\top R (R^\top A R)^{-1} R^\top z = z^\top A^{-1/2} P_2 A^{-1/2} z \leq \|  A^{-1/2} z\|^2 \leq  \lambda_{\max}(A^{-1}) \| z \|^2, 
\end{align*}
where the first equality follows from the fact that $P_2$ is an orthogonal projection matrix. 
\end{proof}

\newpage 

\section{Proofs for Section~\ref{sec:theory}} \label{sec:proofs}
As in Section~\ref{sec:gibbs}, we use $|\gamma|$ and $\ell$ interchangeably. All corresponding indices and symbols (e.g., star, tilde, etc.) are defined equivalently. 
\subsection{Proof of Lemma~\ref{lm:gamma-prior}} \label{sec:proof-lm-gamma-prior}
\begin{proof}
It suffices to prove that under Assumption~\ref{asp:design}, for  any $ \gamma , \gamma' \in \bbS_L$ such that  $\emptyset \neq \gamma \subseteq \gamma'$,  we have 
\begin{align}
  &  \frac{|\gamma'|}{ M \lmin}  \geq \mathsf{1}^\top (X_{\gamma'}^\top X_{\gamma'})^{-1} \mathsf{1} \geq   \mathsf{1}^\top (X_{\gamma}^\top X_{\gamma})^{-1} \mathsf{1} \geq \frac{1}{M},  \label{eq:quad-bound} \\
  &  ( M\lmin)^{|\gamma'|-|\gamma|} \leq \frac{ \det(X_{\gamma'}^\top X_{\gamma'} ) }{\det(X_{\gamma }^\top X_{\gamma } ) } \leq M^{|\gamma'|-|\gamma|}.  \label{eq:detbound-final} 
\end{align}

Consider~\eqref{eq:quad-bound} first. 
Let $S = X_{\gamma' \setminus \gamma}^\top (I - H_\gamma ) X_{\gamma' \setminus \gamma}$. 
Since $X_{\gamma'}= [ X_{\gamma}  \; X_{\gamma' \setminus \gamma} ]$, applying the block matrix inversion formula shows that for any $z^\top = [z_1^\top \; z_2^\top]$, 
\begin{align*}
   z^\top (X_{\gamma'}^\top X_{\gamma'})^{-1} z  
   &=  z_1^\top (X_{\gamma}^\top X_{\gamma})^{-1} z_1 + \tilde{z}^\top S^{-1} \tilde{z} 
  \geq z_1^\top (X_{\gamma}^\top X_{\gamma})^{-1} z_1, 
\end{align*}
where $\tilde{z} = z_2 -  X_{\gamma' \setminus \gamma}^\top X_{\gamma} (X_{\gamma}^\top X_{\gamma})^{-1} z_1$. In particular, $\mathsf{1}^\top (X_{\gamma'}^\top X_{\gamma'})^{-1} \mathsf{1}   \geq  \mathsf{1}^\top (X_{\gamma}^\top X_{\gamma})^{-1} \mathsf{1}.$ 
When $\gamma$ only contains one covariate, Assumption~\ref{asp:design} implies that $X_\gamma^\top X_\gamma = M$, which yields the asserted lower bound. 
For the asserted upper bound,   by the min-max theorem for  eigenvalues, 
\begin{align}
   \mathsf{1}^\top (X_{\gamma'}^\top X_{\gamma'})^{-1} \mathsf{1}  
   \leq \|\mathsf{1}\|_2^2  \; \lambda_{\max}\left( (X_{\gamma'}^\top X_{\gamma'})^{-1} \right) 
   =  \frac{ |\gamma'| }{\lambda_{\min} (  X_{\gamma'}^\top X_{\gamma'}  ) } 
   \leq \frac{|\gamma'|}{ M \lmin}, 
\end{align}
where the last step follows from   Assumption~\ref{asp:design}.  

To prove~\eqref{eq:detbound-final}, we first assume $|\gamma'|-|\gamma|=1$ and let $ \gamma' \setminus \gamma = \{j\}$. Applying the block matrix inversion formula, we obtain that 
\begin{align*}
 \det( X_{\gamma'}^\top X_{\gamma'} ) = \det (X_{\gamma}^\top X_{\gamma}) \det (S). 
\end{align*}
Since $S = X_{j}^\top (I - H_\gamma ) X_{j}$ and $I-H_\gamma$ is a projection matrix, we have $S \leq M$ under Assumption~\ref{asp:design}. 
Moreover, by Lemma~\ref{lemma:yang2016lemma5} below,  we have $S \geq M \lmin$. Hence, 
\begin{align} \label{eq:det2}
   M\lmin \leq  \frac{\det( X_{\gamma'}^\top X_{\gamma'} ) }{ \det( X_{\gamma}^\top X_{\gamma} ) } \leq M.
\end{align}

To conclude the proof, observe that for any $\gamma \subseteq \gamma'$,   $ \det( X_{\gamma'}^\top X_{\gamma'} ) / \det (X_{\gamma}^\top X_{\gamma}) $ can be written as the product of $(|\gamma'| - |\gamma|)$ ratios, where the $k$-th ratio is $\det( X_{\gamma_{k+1}}^\top X_{\gamma_{k+1}} ) / \det (X_{\gamma_k}^\top X_{\gamma_k}) $  for some $\gamma_k, \gamma_{k+1}$ such that $\gamma_{k+1} \setminus \gamma_k = \{j\}$ for some $j$.  
\end{proof}

\begin{lemma} \label{lemma:yang2016lemma5}
Under Assumption~\ref{asp:design}, for  any $ \gamma , \gamma' \in \bbS_L$ such that  $\gamma \subseteq \gamma'$,
    \begin{align}
    \lambda_{\min} (X_{\gamma' \setminus \gamma}^{\top} (I - H_{\gamma} )X_{\gamma' \setminus \gamma}) \geq M \lmin.
\end{align}
\end{lemma}
\begin{proof}
See \citet[pp. 2522]{yang2016computational}.  
\end{proof}

\subsection{Proof of Theorem~\ref{th:marg-posterior-ratio} } \label{sec:proof-th-ratio} 

We divide the proof into multiple parts. First, in Section~\ref{sec:lm:gaussian-conc},  we construct the events that happen with probability at least $1 - c_1 N^{-c_2}$ under the true data-generating probability measure  $\P^*$, where $c_1, c_2 > 0$ are some universal constants.  In Section~\ref{sec:prop:post-expression}, we derive an explicit expression for $p(\gamma, \phi \mid y)$ by using the property of the constrained least-squares estimators. 
In Section~\ref{sec:true-model}, we find a lower bound on $p(\gamma^*, \phi \mid y)$ by analyzing both the unconstrained and constrained least-squares estimators under $\gamma^*$, 
In Section~\ref{subsec:proof-th-ratio}, we prove Theorem~\ref{th:marg-posterior-ratio} by integrating over $\phi$. 
Throughout our proof, we use  
\begin{equation}
H_{\gamma} \coloneqq  X_{\gamma} (X_{\gamma}^\top X_{\gamma})^{-1} X_{\gamma}^\top 
\end{equation} 
to denote the projection matrix  onto the column space of $X_\gamma$.

\subsubsection{High-probability Events} \label{sec:lm:gaussian-conc}

We begin by constructing the events that are essential to our non-asymptotic analysis of the posterior distribution. They enable us to control the behavior of the error vector $\epsilon$. The proof is based on standard concentration inequalities for normal distributions. 

\begin{lemma}\label{lm:gaussian-conc} 
Suppose Assumption \ref{asp:design} holds, $M \geq \log N$, $N \geq 2$ and $L \geq 3$.   
Then, 
$$\P^*( E_1 \cap E_2  \cap E_3 \cap E_4) \geq 1 - c_1 N^{-c_2},$$ 
for some universal constants $c_1, c_2 > 0$, where 
\begin{align} 
    E_1 &= \left\{ \max_{\substack{(\gamma_1, \gamma_2) \in \bbS_L \times \bbS_L\\ \gamma_2 \subset \gamma_1}}
    \frac{ \epsilon^\top ( H_{\gamma_1} - H_{\gamma_2}) \epsilon }{ \ell_1 - \ell_2 } \leq  3 \sigma^2 L  \log N \right\}, \\ 
       E_2 &= \left\{ 
    \frac{1}{M}\max_{ 1 \leq j \leq N}  | \epsilon^\top X_j | \leq  \frac{ 3 \sigma \sqrt{\log N} }{\sqrt{M}}\right\}, \\
    E_3 & = \left\{ \epsilon^\top H_{\gamma^*} \epsilon \leq   \sigma^2 \ell^* \log N \right\}, \\
    E_4 & = \left\{ \frac{7}{8} \leq  \frac{\epsilon^\top \epsilon}{\sigma^2 M} \leq \frac{5}{4} \right\}. 
\end{align}  
\end{lemma}

\begin{proof} 
Given any $\gamma_2 \subset \gamma_1$, we can write $H_{\gamma_1} - H_{\gamma_2} = \sum_{k}  (H_{\gamma_k} - H_{\gamma_{k+1}} )$ where $\gamma_k = \gamma_{k+1} \cup \{j\}$ for some $j \notin \gamma_{k+1}$. Hence,  $E_1 = \tilde{E}_1$, where 
    \begin{align} 
        \tilde{E}_1 &= \left\{ \max_{\substack{\gamma \in \bbS_{L-1}  \\ j \in [N], j \notin  \gamma}}
    \epsilon^\top ( H_{\gamma \cup \{j\} } - H_{\gamma}) \epsilon  \leq 3 \sigma^2 L  \log N \right\}. 
\end{align}   
The rank of $H_{\gamma \cup \{j\} } - H_{\gamma}$ is 1, and thus $\sigma^{-2} \epsilon^\top ( H_{\gamma \cup \{j\}} - H_{\gamma}) \epsilon \sim   \chi^2_{1}$.  
The concentration inequality for normal distribution yields
\begin{align}
   \P^*\left\{ \epsilon^\top ( H_{\gamma \cup \{j\} } - H_{\gamma}) \epsilon \geq 2 \sigma^2  t \right\} \leq 2 e^{-t}, \quad \forall \, t \geq 0. 
\end{align}
Choose $t =  (3/2)L \log N$. Applying the union bound and using $L \geq 3$, we get 
\begin{align}
   \P^*(E_1)  = \P^*(\tilde{E}_1) &\geq 1 - \sum_{\substack{\gamma \in  \bbS_{L-1} \\ j \in [N], j \notin  \gamma}} \P^*\left\{ \epsilon^\top ( H_{\gamma \cup \{j\} } - H_{\gamma}) \epsilon \geq 3 \sigma^2 L \log N \right\}  \\
   & \geq 1 - 2 N^{L+1} \cdot N^{-\frac{3 L}{2}} \\
   & \geq 1 -2 N^{-\frac{1}{2}  }.  
\end{align} 

Consider $E_2$. Define 
\begin{equation}
    d(\epsilon)= \max_{ 1 \leq j \leq N}  | \epsilon^\top X_j |. 
\end{equation}
A routine argument using~\ref{asp:design} shows that $d(\epsilon)$ is Lipschitz continuous in $\epsilon$ with Lipschitz constant $\sqrt{M}$. 
Since $\epsilon \sim N(0, \sigma^2 I_M)$ under $\P^*$, by the concentration inequality for Lipschitz functions of Gaussian vectors, we have 
\begin{equation}\label{eq:conc-lipschitz}
    \P^* \left(  | d(\epsilon) - \E^*[d(\epsilon)] | \geq t \sigma \sqrt{M} \right) \leq 2 e^{-  t^2/2}. 
\end{equation}
Since each $X_j^\top \epsilon \sim N(0, \sigma^2 M)$, a standard result for the maximum of (possibly dependent) Gaussian random variables yields that  
\begin{equation}\label{eq:expect-max-gauss}
    \E^*\left[ d(\epsilon) \right]     \leq  \sigma \sqrt{2  M  \log (2N) } \leq 2 \sigma \sqrt{2  M  \log N }, 
\end{equation}  
whenever $N > 1$. 
Letting $t = \sqrt{   \log N}$ in~\eqref{eq:conc-lipschitz} and using~\eqref{eq:expect-max-gauss}, we obtain that 
\begin{align*}
   \P^*(E_2) &= 1 - \P^* \left(  d(\epsilon) > 3 \sigma \sqrt{ M \log N } \right)\\ 
   &\geq 1 - \P^* \left(  | d(\epsilon) - \E^*[d(\epsilon)] |  >   \sigma \sqrt{ M \log N } \right) \\
   & \geq 1 - 2 N^{- 1/2} 
\end{align*}
An application of the union bound concludes the proof. 

For $E_3$,   the standard concentration inequality for chi-squared distributions yields that $\P^*(E_3) \geq 1 - c_1 N^{-c_2 \ell^*}$ for some universal constants $c_1, c_2 > 0$~\citep{laurent2000adaptive}.  
Similarly, one can also verify that $\P^*(E_4) \geq 1 - c_1 e^{ -c_2 M }$ for some $c_1, c_2 > 0$. 
\end{proof}

\subsubsection{Posterior Density of \texorpdfstring{$ (\gamma, \phi)$}{(gamma, phi)}} \label{sec:prop:post-expression}
The presence of the simplex constraint makes the technical analysis substantially different from that for unconstrained variable selection models. 
The main challenge is to bound the marginal likelihood 
\begin{equation}
  p( y \mid \gamma ,\phi) = \bbE_{ \mu_\gamma \sim \mathrm{Dir}(\mathsf{1}) } \left[ p(y \mid \gamma, \mu_\gamma, \phi) \right].   
\end{equation} 
To address this, we begin by deriving an equivalent expression for $p(y \mid \gamma, \mu_\gamma, \phi)$ using Theorem~\ref{lm:constrained-ls}.  For each $\gamma \in \bbS_L$, define the OLS estimator for the model $\gamma$ by 
\begin{equation}\label{eq:def-ls}
    \hat{\mu}_\gamma = (X_\gamma^\top X_\gamma)^{-1} X_\gamma^\top  y. 
\end{equation}   
Let $\check{\mu}_\gamma$ denote the least-squares estimator under the  constraint $\mathsf{1}^\top \mu_\gamma = 1$. 
By Theorem~\ref{lm:constrained-ls},   
\begin{equation}\label{eq:def-check-mu}
    \check{\mu}_\gamma =  \hat{\mu}_\gamma + \frac{1- \mathsf{1}^\top\hat{\mu}_\gamma} {\mathsf{1}^\top (X_\gamma^\top X_\gamma)^{-1} \mathsf{1}} (X_\gamma^\top X_\gamma)^{-1} \mathsf{1}.
\end{equation}
Further, for any $u \in \Delta^{|\gamma| - 1}$, we have 
\begin{equation}\label{eq:rss-decomp}
    \big\| y - X_\gamma u \big\|^2_2 = \big\| y - X_\gamma \hat{\mu}_\gamma \big\|_2^2
    +  \big\|  X_\gamma  (u - \check{\mu}_\gamma) \big\|_2^2 + b_\gamma, 
\end{equation}
where 
\begin{equation}\label{eq:def-b-gamma}
b_\gamma = \big\|  X_\gamma  (\hat{\mu}_\gamma - \check{\mu}_\gamma) \big\|_2^2 =  \frac{   (1 - \mathsf{1}^\top \hat{\mu}_\gamma )^2  }{\mathsf{1}^\top (X_\gamma^\top X_\gamma)^{-1} \mathsf{1} }. 
\end{equation}   
Equation~\eqref{eq:rss-decomp} decomposes the residual sum of squares for any  $u$ satisfying the simplex constraint into three parts. The first term, $\| y - X_\gamma \hat{\mu}_\gamma \|_2^2$, is the residual sum of squares for the least-squares estimator in the unconstrained setting. 
The second term, $\|  X_\gamma  (u - \check{\mu}_\gamma) \|_2^2 $, measures the deviation of $u$ from the constrained least-squares estimator $\check{\mu}_\gamma$. Note that $X_\gamma \check{\mu}_\gamma$ is the projection of $X_\gamma \hat{\mu}_\gamma$ onto the affine space $\{ X_\gamma u \colon \mathsf{1}^\top u = 1 \}$; in particular, we always have $\mathsf{1}^\top \check{\mu}_\gamma  = 1$.  
The last quantity,   $b_\gamma \geq 0$, measures the deviation of $\check{\mu}_\gamma$ from $\hat{\mu}_\gamma$, which captures how well the model satisfies the simplex constraint. When $\hat{\mu}_\gamma \in \Delta^{|\gamma| - 1}$, we have $b_\gamma = 0$.

For each $\gamma$, we also define two matrices $D_\gamma \in \bbR^{|\gamma| \times (|\gamma| - 1) }$ and $\Pi_\gamma \in \bbR^{ (|\gamma| - 1) \times |\gamma| }$ by 
\begin{equation}\label{eq:def-D}
        D_\gamma = \begin{bmatrix}
I_{ |\gamma| - 1}  \quad 
(-\mathsf{1}) 
\end{bmatrix}^\top, \quad 
\Pi_\gamma = \begin{bmatrix}
    I_{ |\gamma| - 1} &  0 
\end{bmatrix}. 
\end{equation}
Clearly $D_\gamma$ is a full-rank matrix whose column space is orthogonal to the one vector $\mathsf{1}$, and $\Pi_\gamma u$ yields the subvector of $u$ containing the first $|\gamma| - 1$ entries.   
Since the dimension of $D_\gamma$ or $\Pi_\gamma$  is always clear from context, we will omit the subscript  and simply write $D$ and $\Pi$.

In Proposition~\ref{prop:post-expression}, we show that the decomposition given in~\eqref{eq:rss-decomp} enables us to compare the marginal likelihood $p(y \mid \gamma, \phi)$ with that in the unconstrained setting. The expectation with respect to $u\sim \text{Dir}(\mathsf{1})$ is converted to the probability mass (up to a normalizing constant) of a $(|\gamma| - 1)$-dimensional Gaussian distribution centered at $\Pi \check{\mu}_\gamma$.   
  
\begin{proposition}\label{prop:post-expression}
Let $D, \Pi$ beas defined in~\eqref{eq:def-D}. Define 
\begin{equation}\label{eq:effective-penalty}
    f(\ell) = \theta^\ell (1 - \theta)^{ N - \ell} (2 \pi)^{ (\ell - 1)/2}. 
\end{equation}
For any $\gamma \in \bbS_L$ and $\phi > 0$,  the posterior density $p(  \gamma , \phi \mid y )$ given in~\eqref{eq:ga-post} can be expressed as  
\begin{equation}\label{eq:post-simplified}
p(  \gamma , \phi \mid y ) = C  f(|\gamma|) \,  \phi^{(M + \kappa_1)/2}      \, e^{-\frac{\phi}{2} \left(  \| y - X_\gamma \hat{\mu}_\gamma  \|_2^2  + b_\gamma + \kappa_2 \right)  } \,   \P \left( \tilde{U}_{\gamma, \phi} \in {\tilde{\Delta}^{|\gamma| - 1}} \right), 
\end{equation}  
where $C > 0$ is the unknown normalizing constant,   $ \tilde{U}_{\gamma, \phi} \sim \cN ( \Pi \check{\mu}_\gamma, \, (\phi  D^\top X_\gamma^\top X_\gamma D)^{-1})$, and 
\begin{equation}\label{eq:def-tilde-Delta}
    \tilde{\Delta}^\ell = \left\{  \tilde{u} \in \bbR^{\ell} \colon  \tilde{u}_i \geq 0 \text{ for each $i$, and } \sum\nolimits_{i=1}^\ell \tilde{u}_i  \leq 1   \right\}. 
\end{equation}  
\end{proposition}

\begin{proof}  
Observe that for any vector $w$ with $\mathsf{1}^\top w = 0$, we have $w = D \Pi w$. Hence, for any $u \in \Delta^{|\gamma| - 1}$,  
\begin{equation}\label{eq:reparam}
    \| X_\gamma (u  -  \check{\mu}_\gamma )  \|_2^2 =  \| X_\gamma D \, \Pi  (u  -  \check{\mu}_\gamma ) \|_2^2 = \| X_\gamma D  (\tilde{u}  - \Pi \, \check{\mu}_\gamma ) \|_2^2. 
\end{equation}
where $\tilde{u} = \Pi u \in \tilde{\Delta}^{|\gamma| - 1}$. 
By equations~\eqref{eq:rss-decomp} and~\eqref{eq:reparam} and letting $u = \mu_\gamma$, we find that  
\begin{align}
& \quad  p(y \mid \gamma, u, \phi)  \\ 
 & \propto  \phi^{M/2} \,
    \exp\left\{ -\frac{\phi}{2} \| y - X_\gamma u \|_2^2  \right\} \\ 
& =  \phi^{M/2} \, \exp\left\{-\frac{\phi}{2} \left(  \| y - X_\gamma \hat{\mu}_\gamma  \|_2^2  + b_\gamma  + \| X_\gamma D  (\tilde{ u}  -  \Pi \check{\mu}_\gamma ) \|_2^2 \right) \right\} \\
& =  \frac{ \phi^{ (M - |\gamma| + 1)/2}   (2 \pi)^{(|\gamma| - 1)/2}  }{ \sqrt{     \det (   D^\top X_\gamma^\top X_\gamma D ) } }  e^{-\frac{\phi}{2} \left(  \| y - X_\gamma \hat{\mu}_\gamma  \|_2^2  + b_\gamma \right)  }   
 \,  \cN( \tilde{u}; \,  \Pi \check{\mu}_\gamma, \, (\phi  D^\top X_\gamma^\top X_\gamma D)^{-1} ),  \label{eq:full-like}
\end{align}
where $\tilde{u} = \Pi u$ and $\cN( \cdot;  m, V)$ denotes the density function of the normal distribution with mean $v$ and covariance matrix $V$. 
We claim that 
\begin{align}\label{eq:det-likelihood}
\det(D^\top X_\gamma^\top X_\gamma D )=  \det(X_{\gamma}^\top X_{\gamma}) \mathsf{1}^\top (X_\gamma^\top X_\gamma)^{-1} \mathsf{1}. 
\end{align}
To prove this, define $\tilde{D} = [D \; \mathsf{1}]$. A routine calculation yields $\det(  \tilde{D} \tilde{D}^\top  ) = |\gamma|^2$. Hence,   
\begin{align}
        \det(X_{\gamma}^\top X_{\gamma}) & = \frac{1}{|\gamma|^2} \det(\tilde{D}^\top X_{\gamma}^\top X_{\gamma} \tilde{D}) \\
        =  \frac{1}{|\gamma|^2} \det(D^\top X_\gamma^\top X_\gamma D ) & \det\left( \mathsf{1}^\top X_\gamma^\top X_\gamma  \mathsf{1} -  \mathsf{1}^\top X_\gamma^\top X_\gamma D \left[D^\top X_\gamma^\top X_\gamma D  \right]^{-1} D^\top X_\gamma^\top X_\gamma \mathsf{1}  \right), 
\end{align}
where the second step follows from the standard block matrix determinant formula. By the identity~\eqref{lemma716eq} in Lemma~\ref{lemma716}, the matrix involved in the second determinant term can be simplified to 
\begin{align}
   \mathsf{1}^\top X_\gamma^\top X_\gamma  \mathsf{1} -  \mathsf{1}^\top X_\gamma^\top X_\gamma D \left[D^\top X_\gamma^\top X_\gamma D  \right]^{-1} D^\top X_\gamma^\top X_\gamma \mathsf{1}   
& = \frac{   |\gamma |^2 }{  \mathsf{1}^\top (X_\gamma^\top X_\gamma)^{-1} \mathsf{1}  }, 
\end{align}
which   yields~\eqref{eq:det-likelihood}. 

Using~\eqref{eq:full-like},~\eqref{eq:det-likelihood} and the prior specification given in~\eqref{eq:gamma-prior} and~\eqref{eq:phi-prior}, we get 
\begin{align}
        p(\gamma, u, \phi \mid y) &= C  p(\gamma) p(\phi \mid \gamma)  p(u \mid \gamma)  p(y \mid \gamma, u, \phi)   \\
        &= C  f(|\gamma|) \,  \phi^{(M + \kappa_1)/2}      \, e^{-\frac{\phi}{2} \left(  \| y - X_\gamma \hat{\mu}_\gamma  \|_2^2  + b_\gamma + \kappa_2 \right)  } \,  \\
        & \quad \times \cN( \tilde{u}; \,  \Pi \check{\mu}_\gamma, \, (\phi  D^\top X_\gamma^\top X_\gamma D)^{-1} ).  \label{eq:full-cond}
\end{align}
Since integrating in $u$ over $\Delta^{|\gamma| - 1}$ is the same as integrating in $\tilde{u} = \Pi u$ over $\tilde{\Delta}^{|\gamma| - 1}$, we get 
\begin{align*} 
    p(\gamma, \phi \mid y)  = \int_{ \tilde{\Delta}^{|\gamma| - 1} }  p(\gamma, u, \phi \mid y) \d \tilde{u}, 
\end{align*} 
which yields the claimed identity. 
\end{proof}

\subsubsection{Lower Bound on the Posterior Density of \texorpdfstring{$(\gamma^*, \phi)$}{(gamma*, phi)} } \label{sec:true-model} 

Next, we consider the true model $\gamma^*$. By Proposition~\ref{prop:post-expression},  to lower bound $p(\gamma^*, \phi \mid y)$, we need to find an upper bound on $b_\gamma^*$ and a lower bound on 
$\P ( \tilde{U}_{\gamma^*, \phi} \in {\tilde{\Delta}^{\ell^* - 1}} )$ where  
$ \tilde{U}_{\gamma^*, \phi} \sim \mathcal{N}( \Pi \, \check{\mu}_{\gamma^*}, (\phi  D^\top X_{\gamma^*}^\top X_{\gamma^*} D)^{-1})$. Consider $b_\gamma^*$ first. 

\begin{lemma}\label{lm:true-model-bgamma}
Suppose $\mathsf{1}^\top \mu^*_{\gamma^*} = 1$. 
On the event $E_3$ defined in Lemma~\ref{lm:gaussian-conc},   
$b_{\gamma^*} \leq    \sigma^2  \ell^*  \log N$. 
\end{lemma}
\begin{proof}
For the true model $\gamma^*$, $\hat{\mu}_{\gamma^*}$ can be written as
\begin{equation} \label{eq:muhatexp}
    \hat{\mu}_{\gamma^*} = (X_{\gamma^*}^\top X_{\gamma^*})^{-1} X_{\gamma^*}^\top  (X_{\gamma^*} \mu^*_{\gamma^*} + \epsilon ) = \mu^*_{\gamma^*} + (X_{\gamma^*}^\top X_{\gamma^*})^{-1} X_{\gamma^*}^\top  \epsilon. 
\end{equation}
Using 
$\mathsf{1}^\top \mu^*_{\gamma^*} = 1$ we find that 
\begin{equation} \label{eq:1-muhat}
      1 -  \mathsf{1}^\top  \hat{\mu}_{\gamma^*}  = - \mathsf{1}^\top  (X_{\gamma^*}^\top X_{\gamma^*})^{-1} X_{\gamma^*}^\top  \epsilon.   
\end{equation}
By Theorem~\ref{lm:constrained-ls} and the definition of $b_\gamma$ given in~\eqref{eq:def-b-gamma},  
\begin{align}
    b_{\gamma^*} &=  \big\|  X_{\gamma^*}  (\hat{\mu}_{\gamma^*} - \check{\mu}_{\gamma^*} ) \big\|_2^2  
      \leq  \big\|  X_{\gamma^*}  (\hat{\mu}_{\gamma^*} - \mu^*_{\gamma^*}  ) \big\|_2^2 
     = \epsilon^\top H_{\gamma^*} \epsilon \leq \sigma^2 \ell^* \log N, 
\end{align}  
on the event $E_3$. 
\end{proof}

Bounding $\P ( \tilde{U}_{\gamma^*, \phi} \in {\tilde{\Delta}^{\ell^* - 1}} )$ is more difficult, and we divide the argument into two lemmas. 

\begin{lemma}\label{lm:true-model}
Suppose Assumptions~\ref{asp:design} and~\ref{asp:min-thershold} hold. 
On the event $E_2$ defined in Lemma~\ref{lm:gaussian-conc},   
$\check{\mu}_{\gamma^*} \in \Delta^{\ell^* - 1}$, and 
the smallest element of $\check{\mu}_{\gamma^*}$, denoted by $\check{\mu}^*_{\min}$, satisfies  
\begin{equation}\label{eq:bound-check-mu-min}
        \check{\mu}^*_{\min} \geq  
        (c_{\mu} - 3)  \frac{ \sigma \sqrt{L \log N} }{\lmin \sqrt{M }}. 
    \end{equation}
\end{lemma}

\begin{proof} 
Since $\mathsf{1}^\top \check{\mu}_\gamma  = 1$, we only need to prove the lower bound~\eqref{eq:bound-check-mu-min}. By part (iii) of Theorem~\ref{lm:constrained-ls}, 
\begin{equation}
    \check{\mu}_{\gamma^*}  = \mu^*_{\gamma^*} + D ( D^\top X_{\gamma^*}^\top X_{\gamma^*} D )^{-1} D^\top  X_{\gamma^*}^\top  \epsilon. 
\end{equation}
Let $e_j$ denote the vector whose $j$-th entry is one and other entries are all zero. By the inequality~\eqref{lemma716-det} given in Lemma~\ref{lemma716}, on the event $E_2$, 
\begin{align}
     \| \check{\mu}_{\gamma^*} - \mu^*_{\gamma^*} \|_\infty & =  \left\|  D ( D^\top X_{\gamma^*}^\top X_{\gamma^*} D )^{-1} D^\top  X_{\gamma^*}^\top  \epsilon \right\|_\infty  \\
     &= \max_j \left| e_j^\top D ( D^\top X_{\gamma^*}^\top X_{\gamma^*} D )^{-1} D^\top  X_{\gamma^*}^\top  \epsilon  \right| \\ 
        & \leq   \lambda_{\max}\left( D ( D^\top X_{\gamma^*}^\top X_{\gamma^*} D )^{-1} D^\top \right)   \left( \max_j \| e_j \|_2 \right)  \| X_{\gamma^*}^\top  \epsilon  \|_2  \\
        & \leq    \lambda_{\max}\left(   (  X_{\gamma^*}^\top X_{\gamma^*}  )^{-1}  \right)   \, \sqrt{\ell^*}  \| X_{\gamma^*}^\top  \epsilon  \|_\infty    \\ 
        & \leq   \frac{3 \sigma  \sqrt{ \ell^* \log N} }{ \lmin \sqrt{M} },  \label{eq:eq_of_mucheck}
\end{align}
where the last step follows from Assumption~\ref{asp:design} and the definition of $E_2$.  By Assumption~\ref{asp:min-thershold},
\begin{equation}
\| \check{\mu}_{\gamma^*} - \mu^*_{\gamma^*} \|_\infty 
 \leq  \frac{3}{c_{\mu}} \min_{j \in \gamma^*}   | \mu^*_j |,  
\end{equation}
Hence, 
\begin{equation}
       \check{\mu}^*_{\min} \geq \min_{j \in \gamma^*}   | \mu^*_j | - \| \check{\mu}_{\gamma^*} - \mu^*_{\gamma^*} \|_\infty    \geq  \left( 1 - \frac{ 3 }{ c_{\mu}} \right)  \min_{j \in \gamma^*}   | \mu^*_j |, 
\end{equation}
from which the claimed bound follows. 
\end{proof}

\begin{lemma}\label{lm:normal-bound} 
Suppose Assumptions~\ref{asp:design} and~\ref{asp:min-thershold} hold with $c_{\mu} > 6$. Let 
\begin{equation}
   r =    \frac{ c_{\mu} - 3 }{\sqrt{\ell^* - 1} } \frac{ \sigma \sqrt{L \log N} }{\lmin \sqrt{M }}. 
\end{equation}
On the event $E_2$ defined in Lemma~\ref{lm:gaussian-conc}, we have 
 $$\cB_r(  \Pi \, \check{\mu}_{\gamma^*} ) \coloneqq  \{ \tilde{u} \in \bbR^{\ell^* - 1} \colon \| \tilde{u} - \Pi \, \check{\mu}_{\gamma^*} \|_2 \leq r \}  \subset \tilde{\Delta}^{\ell^* - 1}.$$ 
Further,   for $ \tilde{U}_{\gamma^*, \phi} \sim \mathcal{N}(\Pi \, \check{\mu}_{\gamma^*}, (\phi  D^\top X_{\gamma^*}^\top X_{\gamma^*} D)^{-1})$, we have 
\begin{equation} \label{eq:tildeU_prob}
 \P \left( \tilde{U}_{\gamma^*, \phi} \in {\tilde{\Delta}^{\ell^* - 1}} \right) \geq  \P ( \tilde{U}_{\gamma^*, \phi} \in  \cB_r(  \Pi \, \check{\mu}_{\gamma^*} ) )
  \geq 1 - e^{1/2} \exp\left( - \frac{4 \phi \, \sigma^2 L \log N}{   (\ell^*)^2  }  \right). 
\end{equation} 
\end{lemma}

\begin{proof} 
Since $ \Pi \, \check{\mu}_{\gamma^*}$ is the subvector of $ \check{\mu}_{\gamma^*}$, 
Lemma~\ref{lm:true-model} implies that $ \Pi \, \check{\mu}_{\gamma^*} \in \tilde{\Delta}^{\ell^* - 1}$ when $c_{\mu} > 3$. 
We first find the minimum distance from $ \Pi \, \check{\mu}_{\gamma^*}$ to 
the hyperplanes of  $\tilde{\Delta}^{\ell^* - 1}$, which  are determined by $\{ \tilde{u} \colon \tilde{u}_i=0\}$ for $i= 1, \dots, \ell^* - 1$ and $\{\tilde{u} \colon \sum_{i=1}^{\ell^* - 1} \tilde{u}_i  = 1\}$. 
The distance from $ \Pi \, \check{\mu}_{\gamma^*}$ to $\{\tilde{u}: \tilde{u}_i=0\}$ equals the $i$-th entry of $ \Pi \, \check{\mu}_{\gamma^*}$,  and the distance to $\{\tilde{u} \colon \sum_{i=1}^{\ell^* - 1} \tilde{u}_i  = 1\}$ is $(1-\mathsf{1}^\top  \Pi \, \check{\mu}_{\gamma^*} )/\sqrt{\ell^*-1}$. 
Since $\mathsf{1}^\top \check{\mu}_{\gamma^*} = 1$, 
$1-\mathsf{1}^\top  \Pi \, \check{\mu}_{\gamma^*}$ equals the last entry of $\check{\mu}_{\gamma^*}$, which proves that  the ball $\cB_r(\Pi \, \check{\mu}_{\gamma^*})$ is contained in the simplex $ \tilde{\Delta}^{\ell^* - 1}$ whenever $ r \leq \check{\mu}^*_{\mathrm{min}}  / \sqrt{\ell^* - 1}$.  
By Lemma~\ref{lm:true-model}, this condition is satisfied for the given choice of $r$.  Clearly, this implies that 
$\P ( \tilde{U} \in {\tilde{\Delta}^{\ell^* - 1}}  ) \geq  \P ( \tilde{U}_{\gamma^*, \phi} \in  \cB_r(  \Pi \, \check{\mu}_{\gamma^*} ) ).$

By a well-known result for Gaussian concentration~\citep[Theorem 5.1.4]{vershynin2018high},  for the random vector $\tilde{U} = \tilde{U}_{\gamma^*, \phi}$, we have 
\begin{align}  \label{ineq:concentration}
  \P \left( \| \tilde{U} - \Pi \, \check{\mu}_{\gamma^*} \|_2 - \bbE \| \tilde{U} - \Pi \, \check{\mu}_{\gamma^*} \|_2 >  t \right) & \leq \exp \left\{ - \frac{\phi t^2}{2 } \lambda_{\min}  (   D^\top X_{\gamma^*}^\top X_{\gamma^*} D ) \right\}, 
\end{align}
for every $t \geq 0$. 
A direct calculation using $D^\top D = I_{\ell^*}  + \mathsf{1}\mathsf{1}^\top$ shows that  $\lambda_{\max}(D^\top D) = \ell^*$ and $\lambda_{\min}(D^\top D) = 1$.
Hence,   under Assumption~\ref{asp:design}, 
\begin{equation}\label{eq:eigenvalue-comparison-DX}
        \lambda_{\min}( D^\top X_{\gamma^*}^\top X_{\gamma^*} D ) \geq  \lambda_{\min}(   X_{\gamma^*}^\top X_{\gamma^*}  ) \lambda_{\min}( D^\top D ) \geq M \lmin. 
\end{equation} 
The expectation term on the left-hand side of~\eqref{ineq:concentration} then can be bounded by 
\begin{align} \label{ineq:lemma3normub}
  \bbE \| \tilde{U} - \Pi \, \check{\mu}_{\gamma^*}  \|_2 \leq   \sqrt{ \lambda_{\max}( (\phi  D^\top X_{\gamma^*}^\top X_{\gamma^*} D)^{-1} )} \, \bbE  \chi_{\ell^* - 1}  
  \leq \frac{ \sqrt{\ell^*   } }{ \sqrt{\phi M \lmin } }, 
\end{align} 
where $\chi_{\ell^* - 1}  $ denotes a random variable following the $\chi$-distribution with $(\ell^* - 1)$ degrees of freedom. 
Applying~\eqref{ineq:concentration} with~\eqref{eq:eigenvalue-comparison-DX} and~\eqref{ineq:lemma3normub} yields
\begin{align}
    \P \left( \tilde{U} \in \cB_r(  \Pi \, \check{\mu}_{\gamma^*} ) \right) &= 1 -  \P \left(  \| \tilde{U} - \Pi \, \check{\mu}_{\gamma^*} \|_2 - \bbE \| \tilde{U} - \Pi \, \check{\mu}_{\gamma^*}  \|_2 > r - \bbE \| \tilde{U} - \Pi \, \check{\mu}_{\gamma^*}  \|_2 \right)  \\
    &\geq 1 -  \P \left(  \| \tilde{U} - \Pi \, \check{\mu}_{\gamma^*} \|_2 -  \bbE \| \tilde{U} - \Pi \, \check{\mu}_{\gamma^*}  \|_2 > r - \frac{ \sqrt{\ell^*   } }{ \sqrt{\phi M \lmin } } \right)  \\
    & \geq 1 -  \exp \left\{ - \frac{\phi M \lmin   }{2} \left( r -\frac{ \sqrt{\ell^*   } }{ \sqrt{\phi M \lmin } }   \right)^2 \right\}.  \label{eq:ball-prob-bound}
\end{align}
Using $2 ab \leq c a^2 + c^{-1} b^2 $ with $a = r$, $b =  \sqrt{\ell^*   } / \sqrt{\phi M \lmin}$ and $c = \ell^* / (\ell^* + 1)$, we get 
\begin{align*}
 \left( r -\frac{ \sqrt{\ell^*   } }{ \sqrt{\phi M \lmin } }   \right)^2 & \geq r^2 + \frac{ \ell^* }{\phi M \lmin} - \frac{\ell^*}{\ell^* + 1} r^2 - \frac{\ell^* + 1}{ \ell^*} \frac{ \ell^* }{\phi M \lmin} \\
 &= \frac{ r^2 }{\ell^* + 1} - \frac{1}{\phi M \lmin} \\
 &=  \frac{ (\check{\mu}^*_{\mathrm{min}} )^2  }{(\ell^* - 1)(\ell^* + 1)} - \frac{1}{\phi M \lmin}  \\
 & \geq  (c_{\mu} - 3)^2  \frac{ \sigma^2 L \log N}{ (\ell^*)^2 M \lmin^2   } - \frac{1}{\phi M \lmin} \\
 & \geq \frac{ 9 \sigma^2 L \log N}{ (\ell^*)^2   M  \lmin  } - \frac{1}{\phi M \lmin}, 
\end{align*} 
where the last step follows from $c_{\mu} \geq 6$ and $\lmin \leq 1$. Therefore, 
\begin{align*}
     \P \left( \tilde{U} \in \cB_r(  \Pi \, \check{\mu}_{\gamma^*} ) \right)  \geq 1 - \exp\left\{ - \frac{ 9 \phi \sigma^2 L \log N}{2 (\ell^*)^2     } + \frac{1}{2}  \right\}. 
\end{align*}
The stated bound thus  follows. 
\end{proof}

\subsubsection{Proof of Theorem~\ref{th:marg-posterior-ratio}} \label{subsec:proof-th-ratio}

\begin{proof}[Proof of Theorem~\ref{th:marg-posterior-ratio}]
This follows from Propositions~\ref{coro:upper-any-model} and~\ref{coro:lower-true} that we prove below, where we further bound $b_{\gamma^*}$ using  Lemma~\ref{lm:true-model-bgamma}.  
\end{proof} 

\begin{proposition}\label{coro:upper-any-model}
 For any $\gamma \in \bbS_L$ and $\phi > 0$,  
 \begin{align}
      p(\gamma, \phi \mid y) \leq  C f(|\gamma|) \, \Gamma\left( \frac{  M  + \kappa_1 }{2} \right) \left\{ \frac{  \| y - X_\gamma \hat{\mu}_\gamma \|_2^2  + \kappa_2 }{2}\right\}^{  - (\kappa_1 + M  ) / 2 },  
 \end{align}
 where the constant $C$ and function $f$ are as given in Proposition~\ref{prop:post-expression}.  
\end{proposition}

\begin{proof}
Since $\P ( \tilde{U}_{\gamma, \phi} \in {\tilde{\Delta}^{|\gamma| - 1}} ) \in [0, 1]$ and $b_\gamma \geq 0$,  by Proposition~\ref{prop:post-expression}, 
\begin{align*}
    p(\gamma, \phi \mid y) \leq C f(|\gamma|) \phi^{ (M + \kappa_1) / 2} e^{-\frac{\phi}{2} \left(  \| y - X_\gamma \hat{\mu}_\gamma  \|_2^2    + \kappa_2 \right)  }. 
\end{align*}
Integrating over $\phi$ yields the claimed bound. 
\end{proof}

\begin{proposition}\label{coro:lower-true}
Suppose Assumptions~\ref{asp:design} and~\ref{asp:sample-size} hold with $c_{\mu} > 6$ and $c_M \geq 4$.  On the event $E_2 \cap E_3 \cap E_4$ defined in Lemma~\ref{lm:gaussian-conc},  we have 
\begin{align}
p(\gamma^* \mid y ) \geq  \frac{C}{3}
   f(\ell^*) \Gamma\left( \frac{M + \kappa_1}{2} \right) \left\{  \frac{   \| y - X_{\gamma^*} \hat{\mu}_{\gamma^*} \|_2^2 + b_{\gamma^*} + \kappa_2 
 }{2} \right\}^{-(M+\kappa_1)/2},
\end{align}  
where the constant $C$ and function $f$ are as given in Proposition~\ref{prop:post-expression}.  
\end{proposition}

\begin{proof}
Define 
\begin{equation}\label{eq:def-K-tK}
    K = \frac{1}{2} \left( \kappa_2 +  \| y - X_{\gamma^*} \hat{\mu}_{\gamma^*} \|_2^2 + b_{\gamma^*} \right), 
    \quad \tilde{K} = \frac{ 4  \sigma^2 L \log N}{  (\ell^*)^2   }. 
\end{equation} 
By Proposition~\ref{prop:post-expression}  and Lemma~\ref{lm:normal-bound}, we have 
\begin{align*}
    p(\gamma^*, \phi \mid y) \geq C f(\ell^*) \phi^{ (M+\kappa_1) / 2}  e^{- K \phi}     \left( 1 - e^{1/2} e^{- \tilde{K} \phi}  \right).  
\end{align*}
Integrating over $\phi$ yields 
\begin{align}
    p(\gamma^* \mid y) &\geq C f(\ell^*) \int  \phi^{ (M+\kappa_1) / 2}  e^{- K \phi}     \left( 1 - e^{1/2} e^{- \tilde{K} \phi}  \right) \d \phi \\
    &=  C f(\ell^*)  \left[  \int  \phi^{ (M+\kappa_1) / 2}  e^{- K \phi}    \d \phi -  e^{ 1/2}\int  \phi^{ (M+\kappa_1) / 2}  e^{- (K + \tilde{K}) \phi} \d \phi  \right] \\
    &=  C f(\ell^*) \Gamma\left( \frac{M + \kappa_1}{2} \right) \left[   K^{-(M + \kappa_1)/2}    -  e^{ 1/2}  (K + \tilde{K})^{-(M + \kappa_1)/2}  \right]. \label{eq:phi-int-prob}  
\end{align}
It remains to  show that the term involving $K + \tilde{K}$ in~\eqref{eq:phi-int-prob} is negligible. First, we bound   $K$ and $\tilde{K}$. On the event $E_4$, 
\begin{equation}
    \| y - X_{\gamma^*} \hat{\mu}_{\gamma^*} \|_2^2   = \epsilon^\top (I - H_{\gamma^*}) \epsilon \leq \epsilon^\top \epsilon \leq \frac{5}{4} \sigma^2 M. 
\end{equation}
Using Assumption~\ref{asp:hyperparameter}, Lemma~\ref{lm:true-model-bgamma} and $M \geq 4 L \log N$,  we get 
\begin{equation}\label{eq:K-bound}
   2 K = \kappa_2  + \| y - X_{\gamma^*} \hat{\mu}_{\gamma^*} \|_2^2 +    b_{\gamma^*} \leq \frac{ 1}{2} \sigma^2 M + \frac{5}{4} \sigma^2 M + \sigma^2 \ell^* \log N  
    \leq   2 \sigma^2 M. 
\end{equation}
Meanwhile, the conditions $M \geq 4 L \log N$ and $\ell^* \leq \sqrt{L \log N}$ imply that 
\begin{equation}\label{eq:tK-bound}
 4 \sigma^2 \leq    \tilde{K} \leq   \sigma^2 M. 
\end{equation}
Hence, we can use the inequality $\log(1 + x) \geq x/(1 + x)$ for $x > -1$ to get 
\begin{align}
\left( \frac{K + \tilde{K}}{K} \right)^{ - (M+\kappa_1)/2 } & \overset{\text{by~\eqref{eq:K-bound}} }{\leq} \left( 1 + \frac{ \tilde{K}}{ \sigma^2 M } \right)^{ -(M+\kappa_1)/2 }  
    \leq  \exp\left\{ -  \frac{ M \tilde{K}}{  2( \sigma^2 M  + \tilde{K} ) } \right\}  \\
&\overset{\text{by~\eqref{eq:tK-bound}} }{\leq} \exp\left( -  \frac{     \tilde{K}}{  4 \sigma^2 } \right)  
\overset{\text{by~\eqref{eq:tK-bound}} }{\leq}  e^{-1}.   \label{eq:KK-ratio}
\end{align}   
Plugging   this bound into~\eqref{eq:phi-int-prob} yields
\begin{align}
 p(\gamma^* \mid y) &\geq   C f(\ell^*) \Gamma\left( \frac{M + \kappa_1}{2} \right) K^{-(M+\kappa_1)/2} \left\{ 1  - e^{1/2} \left( \frac{ K+ \tilde{K}}{K}\right)^{  - (M + \kappa_1)/2  } \right\}   \\  
 &\geq (1 - e^{-1/2}) 
 C f(\ell^*) \Gamma\left( \frac{M + \kappa_1}{2} \right) K^{-(M+\kappa_1)/2}.   \label{eq:marg-den-bound}
\end{align}
The claimed bound then follows by using $1 - e^{-1/2} \geq 1/3$. 
\end{proof}

\subsection{Proof of Theorem~\ref{thm:post-cons}} \label{sec:proof-post-cons}

To prove Theorem~\ref{thm:post-cons}, we first use Theorem~\ref{th:marg-posterior-ratio} to identify the rate at which $p(\gamma \mid y)/p(\gamma^* \mid y)$ goes to zero for any single $\gamma \neq \gamma^*$.
We consider two different cases separately: $ \gamma^*  \subset \gamma$ (overfitted)  and $\gamma^* \not\subset \gamma$ (underfitted).  

\subsubsection{Bounds for Overfitted Models}\label{sec:overfit}

\begin{proposition}\label{prop:overfit} 
Suppose Assumptions~\ref{asp:design} to~\ref{asp:sample-size} hold and $c_\theta$ is sufficiently large.  On the event $E_1 \cap E_2 \cap E_3 \cap E_4$ defined in Lemma~\ref{lm:gaussian-conc},  
\begin{equation}
   \sup_{ \gamma \in \bbS_L \colon \gamma^*\subsetneq\gamma  }  \frac{ p(\gamma \mid y )}{ p(\gamma^* \mid y )}   \leq 3 N^{- |\gamma \setminus \gamma^*| L }. 
\end{equation} 
\end{proposition}

\begin{proof}  
Fix an overfitted model $\gamma \supsetneq  \gamma^*$ and let $\ell = |\gamma| > \ell^*$.   
Applying the inequality $(1 + x)^n \leq e^{n x}$ for $n > 0$ and using
\begin{equation}\label{eq:L2norm-to-proj}
   \| y - X_{\gamma} \hat{\mu}_{\gamma} \|_2^2  - \| y - X_{\gamma^*} \hat{\mu}_{\gamma^*} \|_2^2   
  =   y^\top (H_{\gamma^*} - H_{\gamma} )  y,   
\end{equation}  
we obtain from Theorem~\ref{th:marg-posterior-ratio} that 
\begin{equation}
    \frac{ p(\gamma \mid y )}{ p(\gamma^* \mid y )}  
    \leq 3 \left( \frac{\theta \sqrt{ 2 \pi } }{1 - \theta} \right)^{ \ell - \ell^* }    
    \exp \left\{   \frac{\kappa_1 + M  }{2}  \frac{ y^\top (H_{\gamma} - H_{\gamma^*}) y   + \sigma^2 \ell^* \log N
    }{ \kappa_2 + \| y - X_\gamma \hat{\mu}_\gamma \|_2^2 }   \right\}.
\end{equation} 

Since $\gamma^* \subset \gamma$, we have $H_{\gamma^*} X_{\gamma^*} = H_{\gamma} X_{\gamma^*} = X_{\gamma^*}$, which implies that, on the event $E_1$,  
\begin{align}
    y^\top (H_{\gamma} - H_{\gamma^*})  y  =  \epsilon^\top (H_{\gamma} - H_{\gamma^*})  \epsilon \leq 3 (\ell - \ell^*) \sigma^2 L \log N. 
\end{align}  
This further yields that $ \sigma^2 \ell^* \log N  + y^\top (H_{\gamma} - H_{\gamma^*})  y \leq 4 (\ell - \ell^*) \sigma^2 L \log N$, as $\ell - \ell^* \geq 1$ and $L \geq \ell$. 
Similarly, on the event $E_1 \cap E_4$, we can use $M \geq 4 L \log N$ to get 
\begin{equation}
    \| y - X_\gamma \hat{\mu}_\gamma \|_2^2 = \epsilon^\top (I - H_\gamma) \epsilon \geq \frac{7}{8} \sigma^2 M - 3 \sigma^2 L \log N \geq \frac{1}{8} \sigma^2 M. 
\end{equation}  
It thus follows from Theorem~\ref{th:marg-posterior-ratio} and Assumptions~\ref{asp:prior} and~\ref{asp:hyperparameter}  that 
\begin{align*}
    \frac{ p(\gamma \mid y )}{ p(\gamma^* \mid y )}  
    &\leq 3 \left( \frac{\theta \sqrt{ 2 \pi } }{ 1 - \theta  } \right)^{ \ell - \ell^* }   
    \exp \left\{   \frac{\kappa_1 + M }{2}  \frac{ y^\top (H_{\gamma} - H_{\gamma^*}) y  + \sigma^2 \ell^* \log N
    }{ \kappa_2 + \| y - X_\gamma \hat{\mu}_\gamma \|_2^2 }   \right\}   \\
    &\leq 3 \left( \frac{\theta \sqrt{ 2 \pi } }{ 1 - \theta  } \right)^{ \ell - \ell^* }   
    \exp \left\{      \frac{ 4 (\ell - \ell^*) \sigma^2 L \log N 
    }{ ( \kappa_2 + \| y - X_\gamma \hat{\mu}_\gamma \|_2^2) / M }    \right\} \\
    & \leq 3 \left( \frac{\theta \sqrt{ 2 \pi } }{ 1 - \theta  } \right)^{ \ell - \ell^* }   
    \exp \left\{     32 (\ell - \ell^*)   L \log N \right\} \\
    & = 3 (2 \pi)^{(\ell - \ell^*) / 2}    
    \exp \left\{     (32 - c_\theta) (\ell - \ell^*)   L \log N   \right\}. 
\end{align*}  
Letting $c_{\theta}$  be sufficiently large, we obtain the claimed bound.
\end{proof}

\subsubsection{Bounds for Underfitted Models}\label{sec:underfit}

\begin{proposition}
\label{prop:underfit} 
Suppose Assumptions~\ref{asp:design} to~\ref{asp:sample-size} hold for sufficiently large $c_\theta, c_\mu, c_M$.  
On the event $E_1 \cap E_2 \cap E_3 \cap E_4$ defined in Lemma~\ref{lm:gaussian-conc},   
\begin{equation}
  \sup_{\gamma \in \bbS_L \colon \gamma^* \not\subset \gamma }   \frac{ p(\gamma \mid y )}{ p(\gamma^* \mid y )}   \leq 3 N^{- ( |\gamma \setminus \gamma^*| + |\gamma^* \setminus \gamma| )  L }. 
\end{equation}   
\end{proposition}

\begin{proof}   
The overall proof strategy parallels the argument for Proposition~\ref{prop:overfit}. 
Let $\ell = |\gamma|$, $\tilde{\gamma} = \gamma \cup \gamma^*$  and $\tilde{\ell} = |\tilde{\gamma}|$.  
We rewrite the inequality~\eqref{eq:post-ratio-simplified} given in Theorem~\ref{th:marg-posterior-ratio} as 
\begin{equation}
   \frac{ p(\gamma \mid y)}{ p(\gamma^* \mid y)}   \leq R_{\rm{over}} R_{\rm{under}}, 
\end{equation}
where 
\begin{align*}
    R_{\rm{over}} &= 3 \left( \frac{\theta \sqrt{2\pi}}{1 - \theta} \right)^{ \tilde{\ell} - \ell^* }  
    \left( \frac{  \kappa_2 + \| y - X_{\gamma^*} \hat{\mu}_{\gamma^*} \|_2^2  + \sigma^2 \ell^* \log N }{ \kappa_2 + \| y - X_{\tilde{\gamma}} \hat{\mu}_{\tilde{\gamma}} \|_2^2  } \right)^{  (\kappa_1 + M ) / 2},  \\
    R_{\rm{under}} &=   \left( \frac{\theta  \sqrt{2\pi}}{1 - \theta}  \right)^{ \ell - \tilde{\ell}   }  
    \left( \frac{  \kappa_2 + \| y - X_{\tilde{\gamma}} \hat{\mu}_{\tilde{\gamma}} \|_2^2   }{ \kappa_2 + \| y - X_\gamma \hat{\mu}_\gamma \|_2^2  } \right)^{  (\kappa_1 + M ) / 2}.  
\end{align*}
Since $\gamma^* \subset \tilde{\gamma}$, Proposition~\ref{prop:overfit} can be applied to bound $R_{\rm{over}}$ by  $R_{\rm{over}} \leq 2 N^{ -L ( \tilde{\ell} - \ell^*) }.$
 
It only remains to bound $R_{\rm{under}}$. As in the overfitted case, applying the inequality $(1 + x)^n \leq e^{n x}$ and using $\kappa_1 \leq M$ in Assumption~\ref{asp:hyperparameter}, we get
\begin{equation}\label{eq:under1}
    R_{\rm{under}} \leq  \left( \frac{\theta  \sqrt{2\pi}}{1 - \theta}  \right)^{ \ell - \tilde{\ell}   } \exp\left\{ -  \frac{  y^\top (H_{\tilde{\gamma}} - H_\gamma) y }{ ( \kappa_2 + \| y - X_\gamma \hat{\mu}_\gamma \|_2^2  ) / M } \right\}. 
\end{equation}  
The difference in residual sum of squares can be bounded as 
\begin{align*}
y^\top (  H_{\tilde{\gamma}} - H_{\gamma} ) y &=  \|(H_{\tilde{\gamma}} - H_{\gamma})(X_{\gamma^*} \mu^*_{\gamma^*} + \epsilon )\|_2^2 \\ 
   &= \|(I - H_{\gamma})X_{\gamma^*} \mu^*_{\gamma^*} + (H_{\tilde{\gamma}} - H_{\gamma})\epsilon \|_2^2    \\ 
   &\geq  \left(\|(I - H_{\gamma})X_{\gamma^*} \mu^*_{\gamma^*}\|_2 - \|(H_{\tilde{\gamma}} - H_{\gamma})\epsilon \|_2\right)^2,  
\end{align*}
where on the second line we have used $H_{\tilde{\gamma}} X_{\gamma^*} = X_{\gamma^*}$ since $\tilde{\gamma}$ is overfitted, and the last step follows from the reverse triangle inequality.  
On the event $E_1$, $ \|(H_{\tilde{\gamma}} - H_{\gamma})\epsilon \|_2^2 \leq 3 (\tilde{\ell} - \ell) \sigma^2 L \log N$. Under Assumptions~\ref{asp:design} and~\ref{asp:min-thershold}, we apply Lemma \ref{lemma:yang2016lemma5} to get 
\begin{align}
   \|(I - H_{\gamma})X_{\gamma^*} \mu^*_{\gamma^*}\|_2^2 
   & = \|(I - H_{\gamma})X_{\gamma^* \setminus \gamma} \mu^*_{\gamma^* \setminus \gamma}\|_2^2  \\ 
   & \geq M \lmin \|\mu^*_{\gamma^* \setminus \gamma}\|_2^2 \\
   & \geq  (\tilde{\ell}-\ell)    c_{\mu}^2 \sigma^2 L \log N.  \label{eq:under-rss-bound}
\end{align}  
For sufficiently large $c_\mu$, we thus find that 
\begin{equation}\label{eq:under2}
 y^\top (  H_{\tilde{\gamma}} - H_{\gamma} ) y    
 \geq \frac{1}{2}\|(I - H_{\gamma})X_{\gamma^*} \mu^*_{\gamma^*}\|_2^2.  
\end{equation} 
To bound the denominator in the exponent, we apply the triangle inequality  to get 
\begin{equation}
  \kappa_2  + \| y - X_\gamma \hat{\mu}_\gamma \|_2^2 \leq  \frac{1}{2} \sigma^2 M +      2 \epsilon^\top \epsilon + 2 \| (I - H_\gamma) X_{\gamma^*} \mu_{\gamma^*}\|_2^2  \leq 3 \sigma^2 M + 2 \| (I - H_\gamma) X_{\gamma^*} \mu_{\gamma^*}\|_2^2, 
\end{equation}
on the event $E_4$. 
Define $S = \| (I - H_\gamma) X_{\gamma^*} \mu_{\gamma^*}\|_2^2$. 
Using the elementary inequality $(x_1 + x_2)^{-1} \geq (2x_1)^{-1} \wedge (2x_2)^{-1}$, we can bound the exponent in~\eqref{eq:under1}  as 
\begin{align*}
 \frac{  y^\top (H_{\tilde{\gamma}} - H_\gamma) y }{ ( \kappa_2 + \| y - X_\gamma \hat{\mu}_\gamma \|_2^2  ) / M } 
 &\geq \frac{ S/2 }{ 3 \sigma^2   +  2S/M } \\ 
 &\geq \frac{S}{12 \sigma^2} \wedge \frac{M}{8} \\
 &\geq  \left( \frac{c_\mu^2}{12} \wedge \frac{c_M}{8} \right) (\tilde{\ell}-\ell)  L \log N, 
\end{align*}
where in the last step we have used Assumptions~\ref{asp:min-thershold} and~\ref{asp:sample-size} and the fact $\tilde{\ell}-\ell \leq \ell^*$. 
It is now clear from~\eqref{eq:under1} that as long as $c_\mu$ and $c_M$ are sufficiently large relative to $c_\theta$, 
$R_{\rm{under}} \leq  N^{-L( \tilde{\ell} - \ell) }.$  
Hence, we conclude that 
\begin{align*}
       \frac{ p(\gamma \mid y)}{ p(\gamma^* \mid y)}   \leq R_{\rm{over}} R_{\rm{under}} \leq 3 N^{-L (2 \tilde{\ell} - \ell^* - \ell)}. 
\end{align*} 
This yields the claimed bound upon since $\tilde{\ell} - \ell^* = | \gamma \setminus \gamma^*|$ and $\tilde{\ell} - \ell = |\gamma^* \setminus \gamma|$. 
\end{proof}

\subsection{Proof of Theorem~\ref{thm:post-cons}}\label{sec:post-cons-proof}

\begin{proof}
In Appendix~\ref{sec:proof-post-cons}, we prove that under our assumptions, with high probability, 
\begin{equation}
  \sup_{\gamma \in \bbS_L \setminus \{ \gamma^* \}}  \frac{ p(\gamma \mid y)}{p(\gamma^* \mid y)} \leq  3 N^{- ( |\gamma \setminus \gamma^*| + |\gamma^* \setminus \gamma| )  L }.
\end{equation} 
See Proposition~\ref{prop:overfit} for the overfitted case (i.e., $ \gamma^*  \subset \gamma$) and Proposition~\ref{prop:underfit} for the underfitted case (i.e., $\gamma^* \not\subset \gamma$).   
A routine calculation using $| \{ \gamma \in \bbS_L \colon |\gamma \setminus \gamma^*|+|\gamma^* \setminus \gamma|  = k\} | \leq N^k$ yields the result. 
\end{proof}

\subsection{Proof of Theorem~\ref{thm:mu-cons}} \label{sec:mu-cons-proof}

\begin{proof} 
Consider $ \bbE_y \left[   \| \mu - \mu^* \|_2^2  \right] $ first. 
For any $\mu$ satisfying the simplex constraint, it is easy to show that $\| \mu - \mu^* \|_2^2 \leq 2$. Hence, for any $\gamma \in \bbS_L$,  $\bbE_y \left[   \| \mu_\gamma - \mu^*_\gamma \|_2^2   \right] \leq 2$ since  $p(\mu_\gamma \mid y, \gamma)$ has support $\Delta^{|\gamma| - 1}$. Applying Theorem~\ref{thm:post-cons}, we get 
\begin{align}
      \bbE_y \left[   \| \mu - \mu^* \|_2^2  \right] 
&= \sum_{\gamma \in \bbS_L  }  p(\gamma \mid y)  \bbE_y \left[   \| \mu_\gamma - \mu^*_\gamma \|_2^2  \right] \\    
& \leq p(\gamma^* \mid y)  \bbE_y \left[   \| \mu_{\gamma^*} - \mu^*_{\gamma^*} \|_2^2  \right]  +  2 \sum_{\gamma \in \bbS_L \setminus \{\gamma^*\} } p(\gamma \mid y) \\
& \leq  \bbE_y \left[   \| \mu_{\gamma^*} - \mu^*_{\gamma^*} \|_2^2   \right]  + 2 c_3 N^{-c_4 L},   \label{eq:mu-loss-1}
\end{align}
where $c_3, c_4 > 0$  are universal constants introduced in Theorem~\ref{thm:post-cons}. 

It remains to bound 
\begin{align}
      \bbE_y \left[   \| \mu_{\gamma^*} - \mu^*_{\gamma^*} \|_2^2  \right]  
   &=  \int    \| \mu_{\gamma^*} - \mu^*_{\gamma^*} \|_2^2  \;  p (\mu_{\gamma^*}, \phi  \mid y, \gamma^*) \d \mu_{\gamma^*} \d \phi   \\[0.5em] 
   &=  \frac{1}{p(\gamma^* \mid y )}  \int    \| \mu_{\gamma^*} - \mu^*_{\gamma^*} \|_2^2  \; p (\gamma^*, \mu_{\gamma^*}, \phi  \mid y)  \d \mu_{\gamma^*} \d \phi.      \label{eq:mu-loss-ratio}
\end{align}
Since when $\ell^* = 1$, $ \bbE_y \left[   \| \mu_{\gamma^*} - \mu^*_{\gamma^*} \|_2^2  \right] = 0$ due to the simplex constraint, we assume $\ell^* \geq 2$ henceforth. 
We   use the notation introduced in Proposition~\ref{prop:post-expression} and write $u = \mu_\gamma^*$, $\tilde{u} = \Pi u$. The integral in~\eqref{eq:mu-loss-ratio} is understood as an integral over $(\tilde{u}, \phi) \in \tilde{\Delta}^{\ell^* - 1} \times [0, \infty)$.  Let $\cB_r( \Pi \, \check{\mu}_{\gamma^*} )$ be as defined in Lemma~\ref{lm:normal-bound}, and define 
$$B_1 = \{ (\tilde{u}, \phi) \colon  \tilde{u} \in \cB_r( \Pi \, \check{\mu}_{\gamma^*} ), \, \phi > 0 \}, \quad 
B_2 =  \{ (\tilde{u}, \phi) \colon  \tilde{u} \in \tilde{\Delta}^{\ell^* - 1} \setminus  \cB_r( \Pi \, \check{\mu}_{\gamma^*} ),\, \phi > 0 \}.$$ 
Since   $\cB_r( \Pi \, \check{\mu}_{\gamma^*} ) \subset \tilde{\Delta}^{\ell^* - 1}$ by Lemma~\ref{lm:normal-bound}, we can split the integral in~\eqref{eq:mu-loss-ratio} into integrals over $B_1$ and $B_2$. We claim that for some univeral constant $c_5 > 0$, 
\begin{equation}\label{eq:loss-B1-B2}
\| u -  \mu^*_{\gamma^*} \| \leq \frac{c_5 \sigma  \sqrt{ L \log N} }{ \lmin \sqrt{M}} \text{ on } B_1, \text{ and }  \| u -  \mu^*_{\gamma^*} \| \leq 2 \text{ on } B_2. 
\end{equation}
To prove the claim for $B_1$,  we note that  by the definition of $\cB_r( \Pi \, \check{\mu}_{\gamma^*} )$,  if $(\tilde{u}, \phi) \in B_1$,  
\begin{align}
    \| u - \check{\mu}_{\gamma^*} \|_2^2 &=  \| \tilde{\mu} - \Pi \, \check{\mu}_{\gamma^*}  \|_2^2 +  \| 1 - \mathsf{1}^\top \tilde{u} -(1 - \mathsf{1}^\top  \Pi \, \check{\mu}_{\gamma^*} ) \|^2 \\
    &\leq  r^2 + \| \mathsf{1}^\top (\tilde{u} - \Pi \, \check{\mu}_{\gamma^*} ) \|^2  \leq \ell^* r^2, 
\end{align}
in which the  last step follows from Cauchy-Schwarz inequality. Using the definition of $r$ given in Lemma~\ref{lm:true-model} and the assumption $\ell^* \geq 2$, we find that 
\begin{align}\label{eq:mu-B1-1}
   \| u - \check{\mu}_{\gamma^*} \|_2   \leq    \sqrt{2} (c_\mu - 3)  \frac{  \sigma  \sqrt{ L \log N} }{ \lmin \sqrt{M}},  
\end{align}

Meanwhile, repeating the argument for~\eqref{eq:eq_of_mucheck}, we get 
\begin{align}\label{eq:mu-B1-2}
      \| \mu^*_{\gamma^*} - \check{\mu}_{\gamma^*} \|_2  \leq \frac{3 \sigma  \sqrt{ \ell^* \log N} }{ \lmin \sqrt{M}},  
\end{align} 
on the event $E_2$. Combining~\eqref{eq:mu-B1-1} and~\eqref{eq:mu-B1-2} proves~\eqref{eq:loss-B1-B2}.  Plugging the bounds in~\eqref{eq:loss-B1-B2}  back into~\eqref{eq:mu-loss-ratio}, we get 
\begin{align}
    \bbE_y \left[   \| \mu_{\gamma^*} - \mu^*_{\gamma^*} \|_2^2   \right]  
    \leq  \frac{c_5^2 \sigma^2    L \log N  }{ \lmin^2  M } + \frac{2}{p(\gamma^* \mid y )}  \int_{B_2}   p (\gamma^*, u, \phi  \mid y)  \d \tilde{u} \, \d \phi.      \label{eq:mu-loss-2}
\end{align}
 
By~\eqref{eq:full-cond} and using the notation introduced in 
equation~\eqref{eq:def-K-tK},  we have  
\begin{align}
        p(\gamma^*, u, \phi \mid y) = C  f( \ell^* ) \,  \phi^{(M + \kappa_1)/2}      \, e^{-K   \phi  } \,  \cN( \tilde{u}; \,  \Pi \check{\mu}_{\gamma^*}, \, (\phi  D^\top X_{\gamma^*}^\top X_{\gamma^*} D)^{-1} ). 
\end{align}
Applying Proposition~\ref{coro:lower-true}, we get 
\begin{align}
 \frac{\int_{B_2}   p (\gamma^*, u, \phi  \mid y)  \d \tilde{u} \, \d \phi}{p(\gamma^* \mid y )}       
&\leq  \frac{3 \int_{B_2}    \phi^{(M + \kappa_1)/2}      \, e^{- K \phi } \,  \cN( \tilde{u}; \,  \Pi \check{\mu}_{\gamma^*}, \, (\phi  D^\top X_{\gamma^*}^\top X_{\gamma^*} D)^{-1} )   \d u  \, \d \phi    }{    \Gamma\left( \frac{M + \kappa_1}{2} \right) K^{-(M+\kappa_1)/2}  } \\
&= \frac{3 \int    \phi^{(M + \kappa_1)/2}      \, e^{- K \phi } [ 1 -  \P  ( \tilde{U}_{\gamma^*, \phi} \in {\tilde{\Delta}^{|\gamma| - 1}}  )  ] \d \phi    }{    \Gamma\left( \frac{M + \kappa_1}{2} \right) K^{-(M+\kappa_1)/2}  },  \label{eq:mu-ratio-B2}
\end{align}
where $\tilde{U}_{\gamma^*, \phi} \sim \cN ( \Pi \check{\mu}_\gamma, \, (\phi  D^\top X_\gamma^\top X_\gamma D)^{-1})$ is as defined in Proposition~\ref{prop:post-expression}.  Using Lemma~\ref{lm:normal-bound} and the notation $\tilde{K}$ introduced in 
equation~\eqref{eq:def-K-tK}, 
\begin{align*}
    \int    \phi^{(M + \kappa_1)/2}      \, e^{- K \phi } [ 1 -  & \P  ( \tilde{U}_{\gamma^*, \phi} \in {\tilde{\Delta}^{|\gamma| - 1}}  )  ] \d \phi  
     \leq  e^{1/2} \int    \phi^{(M + \kappa_1)/2}      \, e^{- (K + \tilde{K}) \phi } \d \phi \\
    & = e^{1/2}  \Gamma\left( \frac{M + \kappa_1}{2} \right) (K + \tilde{K})^{-(M + \kappa_1)/2}. 
\end{align*}
By~\eqref{eq:KK-ratio}, 
\begin{align*}
    \left( \frac{K + \tilde{K}}{K} \right)^{ - (M+\kappa_1)/2 }  \leq  \exp\left( -  \frac{     \tilde{K}}{  4 \sigma^2 } \right)  = \exp\left( -  \frac{L \log N}{  (\ell^*)^2   } \right).
\end{align*}
Since $e^{1/2} \leq 2$, it  follows from~\eqref{eq:mu-ratio-B2} that 
\begin{align}
 \frac{\int_{B_2}   p (\gamma^*, u, \phi  \mid y)  \d \tilde{u} \, \d \phi}{p(\gamma^* \mid y )}  
 \leq 6 \exp\left( -  \frac{L \log N}{  (\ell^*)^2   } \right).  \label{eq:mu-loss-3}
\end{align}
Combining~\eqref{eq:mu-loss-1},~\eqref{eq:mu-loss-2} and~\eqref{eq:mu-loss-3} yields the claimed bound for $ \bbE_y \left[   \| \mu - \mu^* \|_2^2  \right] $. 

For the posterior expected prediction loss,  observe that given any model $\gamma$ and $\mu_\gamma \in \Delta^{|\gamma| - 1}$, we can write $\Xnew \mu = \Xnew_{\gamma^*} \mu_{\gamma^*} + \Xnew_{\gamma \setminus \gamma^*} \mu_{\gamma \setminus \gamma^* }$. Hence,  
\begin{align*}
   \| \Xnew \mu - \Xnew \mu^* \|^2_2  &= 
   \|  \Xnew_{\gamma^*} (\mu_{\gamma^*} - \mu^*_{\gamma^*} ) + \Xnew_{\gamma \setminus \gamma^*} \mu_{\gamma \setminus \gamma^* } \|_2^2   
    \leq 2 \|  \Xnew_{\gamma^*} (\mu_{\gamma^*} - \mu^*_{\gamma^*} ) \|^2_2  +2 \|  \Xnew_{\gamma \setminus \gamma^*} \mu_{\gamma \setminus \gamma^* } \|_2^2 \\
   & \leq 2 \tilde{M} \lmax  \| \mu_{\gamma^*} - \mu^*_{\gamma^*}  \|^2_2 + 2 \tilde{M} \lmax \| \mu_{\gamma \setminus \gamma^* } \|^2_2  
    = 2 \tilde{M} \lmax \| \mu - \mu^* \|_2^2. 
\end{align*}
Therefore, the claimed bound on $ \bbE_y \left[   \| \Xnew \mu - \Xnew \mu^* \|_2^2 \right] $ immediately follows from the bound on $ \bbE_y \left[   \| \mu - \mu^* \|_2^2 \right] $. 
\end{proof}

\subsection{Proof of Theorem~\ref{thm:eqbvscss}}\label{sec:const-NS-proof} 

Consider the model $ Y = X w + \epsilon$ with $\epsilon \sim N(0, \phi^{-1} I )$. 
We use $\tilde{p}( \cdot \mid y)$ to denote the posterior distribution under the following prior:  
\begin{align}  
    \phi \sim \;& \mathrm{Gamma}(\kappa_1/2, \kappa_2 / 2), \\
    \gamma_i   \overset{\mathrm{i.i.d.}}{\sim} \;& \mathrm{Bernoulli}(\theta), \\  
     \tau   \sim  \;&  \mathrm{Gamma}(a_1,a_2), \\ 
     w_\gamma \mid \gamma, \tau, \phi \sim \;& N(0, \, (\tau / \phi) I ), \\ 
     w_i \mid \gamma_i = 0 \sim \;& \delta_0. 
\end{align}
This is the classical setup for the spike-and-slab Bayesian variable selection, except that we place an additional prior distribution on $\tau$ so that it aligns with our \BVS{} model. 

\begin{proof}[\textbf{Proof of Theorem \ref{thm:eqbvscss}}]
The posterior probability of $\gamma$  given $\tau$ under the two models can be expressed as 
\begin{equation}\label{eq:def-tilde-C}
    p(\gamma \mid y, \tau) = C_\tau \, p_0(\gamma \mid y, \tau), \quad 
    \tilde{p}(\gamma \mid y, \tau) = \tilde{C}_\tau \, \tilde{p}_0 (\gamma \mid y, \tau), 
\end{equation}
where $C_\tau, \tilde{C}_\tau$ denote the normalizing constants, and $p_0, \tilde{p}_0$ denote the un-normalized posterior distributions given by 
\begin{align} 
    p_0(\gamma \mid y, \tau) &=  p(\gamma) \int 
    p(y \mid \gamma,  \mu_\gamma, \phi, \tau) p(\mu_\gamma \mid \gamma )   p(\phi) \d \mu_\gamma \, \d \phi, \\ 
    \tilde{p}_0(\gamma \mid y, \tau) &=   \tilde{p}(\gamma)  \int 
    \tilde{p}(y \mid \gamma,  \phi, \tau) \tilde{p}(\phi) \d \phi,  \label{eq:marg-tilde-p}
\end{align}
Note that the two models share the same prior distribution on $(\gamma, \phi)$: $p(\gamma) = \tilde{p}(\gamma)$ and $p(\phi) = \tilde{p}(\phi)$. 
For the unconstrained spike-and-slab model, we have 
\begin{align} 
     \tilde{p}(y \mid \gamma, \phi, \tau) \;& = \left(\frac{\phi}{2\pi}\right)^{M/2} \tau^{-\ell/2} \mathrm{det} (V_{\gamma, \tau})^{-1/2} \exp\left\{ -\frac{\phi}{2}   y^\top \Sigma_{\gamma, \tau} y    \right \}, 
\end{align}
where $\ell = |\gamma|$, $V_{\gamma, \tau}$ and $\Sigma_{\gamma, \tau}$ are as defined in \eqref{eq:Sigma}. Hence, 
\begin{align}
   & \tilde{p}_0(\gamma \mid y, \tau)
   =  G(\gamma, \tau) \left(  y^\top \Sigma_{\gamma, \tau} y   + \kappa_2   \right)^{ -(M + \kappa_1) / 2},   \label{eq:like-tilde-p} \\ 
 \text{ where }  & G(\gamma, \tau) =  p(\gamma) \frac{ \Gamma( (M + \kappa_1) / 2 )\,  (\kappa_2 / 2)^{\kappa_1 / 2} }{ \Gamma( \kappa_1 / 2 ) \, (2 \pi)^{M/2} } \, \tau^{-\ell/2} \mathrm{det} (V_{\gamma, \tau})^{-1/2} \, 2^{(M + \kappa_1)/2}. 
\end{align}

Recall our likelihood given by~\eqref{eq:like1}: 
\begin{align}
    p(y \mid \mu_\gamma, \gamma,  \phi, \tau) \;& = \left(\frac{\phi}{2\pi}\right)^{M/2}  \tau^{-\ell / 2} \mathrm{det} (V_{\gamma, \tau})^{-1/2} 
    \exp\left\{ -\frac{\phi}{2}   (y - X_{\gamma } \mu_{\gamma })^\top \Sigma_{\gamma, \tau} (y - X_{\gamma } \mu_{\gamma })    \right \}. 
\end{align}
By integrating out $\mu_\gamma$ and $\phi$, we get  
\begin{align}
     p_0(\gamma \mid y, \tau)   =   G(\gamma, \tau)
    \bbE_{ \mu_\gamma \sim \mathrm{Dir}(\alpha) } \left[  
    \left\{ (X_{\gamma } \mu_{\gamma }-y)^{\top} \Sigma_{\gamma, \tau} (X_{\gamma } \mu_{\gamma }-y)    + \kappa_2  \right\}^{ -(M + \kappa_1) / 2}  \right],  \label{eq:like-p}
\end{align}
where $\mathrm{Dir}(\alpha)$ denotes the Dirichlet distribution with parameter vector $\alpha \mathsf{1}$ (i.e., symmetric Dirichlet distribution with concentration parameter $\alpha$).  

Combining~\eqref{eq:def-tilde-C},~\eqref{eq:like-tilde-p} and~\eqref{eq:like-p},  we get 
\begin{align}
    \frac{  p(\gamma \mid y, \tau) }{  \tilde{p}(\gamma \mid y, \tau) }=      \frac{C_\tau  }{ \tilde{C}_\tau }   \frac{ p_0(y \mid \gamma, \tau) }{ \tilde{p}_0(y \mid \gamma, \tau) } 
    = \frac{C_\tau  }{ \tilde{C}_\tau }  
    \bbE_{ \mu_\gamma \sim \mathrm{Dir}(\alpha) } \left[  
    F(\mu_\gamma; \gamma, \tau)  \right], \label{eq:ratio-two-models-1} 
\end{align}
where 
\begin{align}
    F(\mu_\gamma; \gamma, \tau) = \left( \frac{ (X_{\gamma} \mu_{\gamma}-y)^{\top} \Sigma_{\gamma, \tau} (X_{\gamma } \mu_{\gamma }-y)    + \kappa_2 }{ y^\top \Sigma_{\gamma, \tau} y   + \kappa_2  }\right)^{ -\frac{M + \kappa_1}{2} } \leq \left( \frac{ \kappa_2 }{ y^\top \Sigma_{\gamma, \tau} y   + \kappa_2  }\right)^{ -\frac{M + \kappa_1}{2} }.   
\end{align}
The inequality follows from that   $\Sigma_{\gamma, \tau} = (I + \tau X_\gamma X_\gamma^\top )^{-1}$ is always non-negative definite.  
Hence, we can apply dominated convergence theorem to  get 
\begin{align}
\lim_{\tau \rightarrow \infty } \bbE_{ \mu_\gamma \sim \mathrm{Dir}(\alpha) } \left[  F(\mu_\gamma; \gamma, \tau)  \right] = 
\bbE_{ \mu_\gamma \sim \mathrm{Dir}(\alpha) } \left[ \lim_{\tau \rightarrow \infty }  F(\mu_\gamma; \gamma, \tau)  \right] = 1.  \label{eq:conv-F}
\end{align}
To see that $F(\mu_\gamma; \gamma, \tau)$ converges to $1$ as $\tau \rightarrow \infty$, we note that
\begin{align}
   & \quad \lim_{\tau \rightarrow \infty} 
   \lambda_{\max} \left( \Sigma_{\gamma, \tau}  - (I  - X_\gamma (X_\gamma^\top X_\gamma )^{-1} X_\gamma^\top )  \right) \\
   &=    \lim_{\tau \rightarrow \infty}   \lambda_{\max} \left(    X_\gamma (X_\gamma^\top X_\gamma )^{-1} X_\gamma^\top - X_\gamma (X_\gamma^\top X_\gamma + \tau^{-1} I)^{-1} X_\gamma^\top  \right) = 0,  \label{eq:eigenequi}
\end{align}
which can be proved by using the singular value decomposition (SVD) of $X_\gamma$.  
Assume the rank of $X_\gamma$ is $\ell < M$. Let the SVD of $X_\gamma$ be $X_\gamma = U [S \; 0]^\top V^\top$, where $U \in \bbR^{M \times M}$, $V \in \bbR^{\ell \times \ell}$ are orthogonal, and $S \in \bbR^{\ell \times \ell}$ is diagonal. Denote the diagonal elements of $S$ (i.e, singular values) as $s_i$ for $i \in [\ell]$. 
A routine calculation shows that the eigenvalues of $ X_\gamma (X_\gamma^\top X_\gamma )^{-1} X_\gamma^\top - X_\gamma (X_\gamma^\top X_\gamma + \tau^{-1} I)^{-1} X_\gamma^\top $ are $1/(\tau s_i^2 + 1)$ for $i \in [\ell]$, which converge to 0 as $\tau \rightarrow \infty$. 
When $X_\gamma$ does not have full rank, we only need to change the limit to $ I  - X_\gamma (X_\gamma^\top X_\gamma )^{+} X_\gamma^\top$, where $A^+$ denotes the Moore-Penrose pseudoinverse, and then apply the same  argument.  
Hence, $\Sigma_{\gamma, \tau}$ converges to the projection matrix $I  - X_\gamma (X_\gamma^\top X_\gamma )^{-1} X_\gamma^\top $ in operator norm, which implies that $X_\gamma^\top \Sigma_{\gamma, \tau} \rightarrow 0$ in $L^2$-norm and thus $F(\mu_\gamma; \gamma, \tau) \rightarrow 1$. 

By~\eqref{eq:conv-F} and~\eqref{eq:ratio-two-models-1}, we can prove the claim provided that $\lim_{\tau \rightarrow \infty} C_\tau / \tilde{C}_\tau = 1$. To this end, recall that $C_\tau, \tilde{C}_\tau$ are normalizing constants and thus can be expressed as 
\begin{align}
   C_\tau^{-1} = \sum_{\gamma \in \bbS_L}  p_0(\gamma \mid y, \tau), \quad 
     \tilde{C}_\tau^{-1} = \sum_{\gamma \in \bbS_L} \tilde{p}_0 (\gamma \mid y, \tau). 
\end{align}
Since the set $\bbS_L$ is finite and we have shown that $p_0(y \mid \gamma, \tau) / \tilde{p}_0(y \mid \gamma, \tau) \rightarrow 1$ for each $\gamma$, we have $C_\tau / \tilde{C}_\tau \rightarrow 1$ as well. 
\end{proof}

\newpage

\section{Additional Simulation Studies}\label{sec:more_sim}
We present in the main text the simulation results for a sparse regression model, where each entry of $X$ is generated independently from a standard normal distribution. In this section, we provide additional simulations to evaluate the performance of \BVS{} in estimating the ATT, selecting relevant variables, and estimating variance components under two alternative data-generating processes.

First, in Section~\ref{sec:factor-setting}, we consider a non-sparse setting where the data $(Y, X)$ are generated from a factor model. Second, in Section~\ref{sec:factor-sparse-setting}, we examine a sparse model in which the covariates in $X$ are correlated. These simulation designs follow the factor-model setups commonly used in the literature \citep[e.g.,][]{hsiao2012panel, shi2023forward}.

\subsection{Non-sparse Factor Models}\label{sec:factor-setting}

We first assume that the training data $X \in \bbR^{M \times N}$ and $Y \in \bbR^{M}$ are generated from a non-sparse factor model,  
\begin{align}
    (Y, \, X) = F^\top \Lambda^\top + e, 
\end{align}
where $(Y, X) \in \bbR^{ M \times (N + 1) }$ denotes the block matrix obtained by concatenating $Y$ and $X$, $e$ is an $M\times (N+1)$ error matrix with each entry $e_{ij} \stackrel{\text{i.i.d.}}{\sim} N(0,0.5^2)$, $F \in \bbR^{4 \times M}$ is a factor matrix with 4 factors, and $\Lambda \in \bbR^{(N+1) \times 4}$ is the corresponding loading matrix. 
For the matrix $F$, we let the four common factors evolve according to
\begin{align}
    f_{1,i} &\sim \mathcal{N}(0,1), \\
    f_{2,i} &= 0.9 f_{2,i-1} + u_{1,i}, \\
    f_{3,i} &= 0.5 f_{3,i-1} + u_{2,i} + 0.5 u_{2,i-1}, \\
    f_{4,i} &= u_{3,i} + 0.8 u_{3,i-1} + 0.4 u_{3,i-2}, \qquad i = 1, \dots, M,   \label{eq:factor structure}
\end{align}
in which $u_{1,i}, u_{2,i}, u_{3,i}$ are i.i.d. standard normal. For $\Lambda$, we specify each loading $\lambda_{j,l}$   as
\begin{equation}
\lambda_{j,l} \stackrel{\text{i.i.d.}}{\sim}
\begin{cases}
\mathrm{Unif}[1,2], & j = 1, \dots, J + 1, \\[0.15cm]
\mathrm{Unif}\!\left[-\frac{2}{M+\tilde{M}}, -\frac{2}{M+\tilde{M}}\right], & j > J + 1,
\end{cases}
\end{equation}
where $J$ is a parameter whose value will be specified in the simulation. 

Then we simulate the data under treatment $\Xnew \in \bbR^{\tilde{M} \times N}$ and $\Ynew^{(1)} \in \bbR^{\tilde{M}}$ by
\begin{align}
    (\Ynew^{(0)}, \Xnew) &= \tilde{F}^\top \Lambda^\top + \tilde{e}, \\
    \Ynew^{(1)} &= \Ynew^{(0)} + \delta ,
\end{align}
where $\delta_i$ is defined as in Section~\ref{sec:sim-setting}, and  $\tilde{e} \in \bbR^{\tilde{M}\times (N+1)}$ are defined similarly. $\tilde{F} \in \bbR^{4 \times \tilde{M}}$ follows the same factor structure as $F$; that is, we generate $(F, \tilde{F})$ by~\eqref{eq:factor structure} for $i = 1, \dots, M+\tilde{M}$.
All control units are correlated with the treated unit through the common factors. However, the dependence persists primarily for the first $J$ units, while the correlation for the remaining units diminish as the sample size grows.  

As in the main text, we use  OLS$(\gamma^*)$ as the reference, and we report the relative efficiencies of the other five estimators in Figure~\ref{plot:re_factor}. Under a factor structure, the simplex constraint is typically violated, and thus the performances of two QP estimators deteriorate due to the violation of the simplex constraint. \BVS{} outperforms Lasso especially when the ratio $N/M$ is large and quickly approaches the performance of OLS$(\gamma^*)$. Note that OLS also approaches the benchmark OLS($\gamma^*$) as the sample size increases when $N=20$, consistent with the finding that OLS is a reasonable tool for the first stage estimation in SCM \citep[e.g., ][]{hsiao2012panel}. However, OLS fails to provide a unique estimation when the sample size is small (e.g., $M=25, N=50$). These results reinforce the key message in Section~\ref{sec:sim} in the main text: the robust and superior performance of \BVS{} is largely driven by the use of the soft simplex constraint, which provides substantial gains even when the true DGP follows a non-sparse and unknown factor structure.

\begin{figure}
    \centering
    \includegraphics[width=1\linewidth]{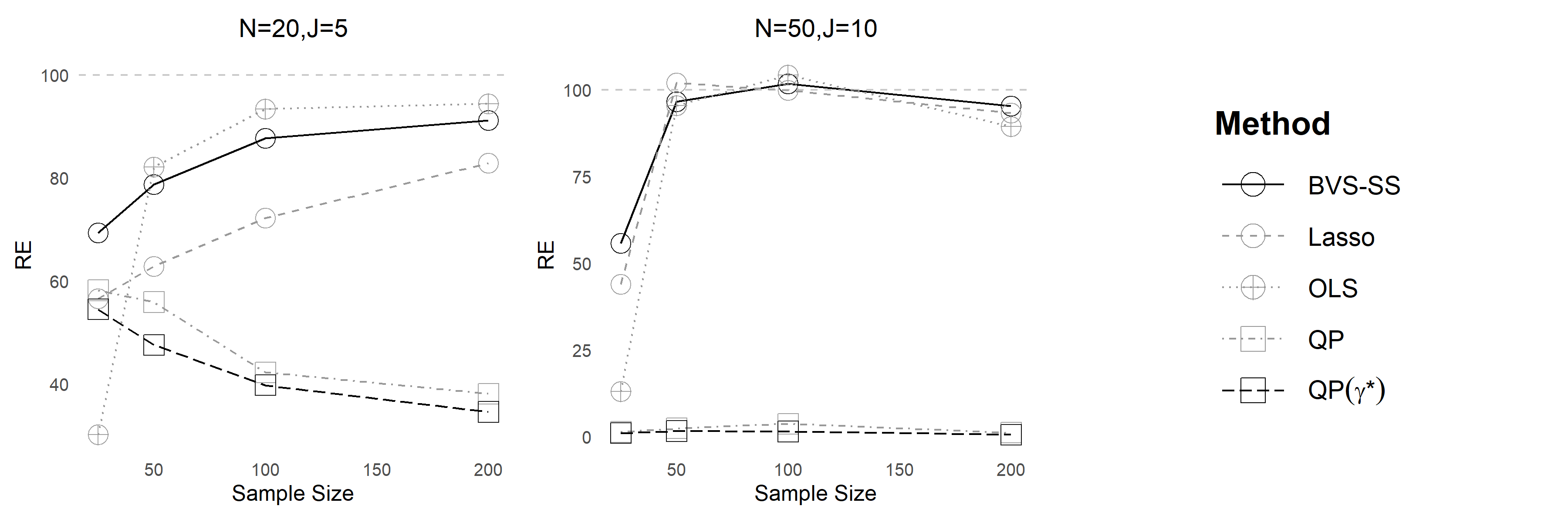}
    \caption{ Relative efficiency (\%) against sample size for ATT estimation for factor models. Left: $N=20, J=5$;  right: $N=50, J=10$. }
    \label{plot:re_factor}
\end{figure}

Meanwhile, the posterior mean of $\phi$ does not exhibit a clear increasing trend. This is due to our simulation setting, in which the true DGP is a factor model while our Bayesian estimation is based on a sparse regression model. The trend of $\tau$ remains relatively stable, exhibiting only a very slight decrease as $M$ increases. Moreover, the magnitude of the posterior mean of $\tau$ is very close to that in the sparse regression model with $||w^*||_1=2$ presented in Section~\ref{sec:sim}. This further supports the notion that the factor model is likely to violate the simplex constraint.

\begin{figure}[!h]
    \centering
    \includegraphics[width=0.7\linewidth]{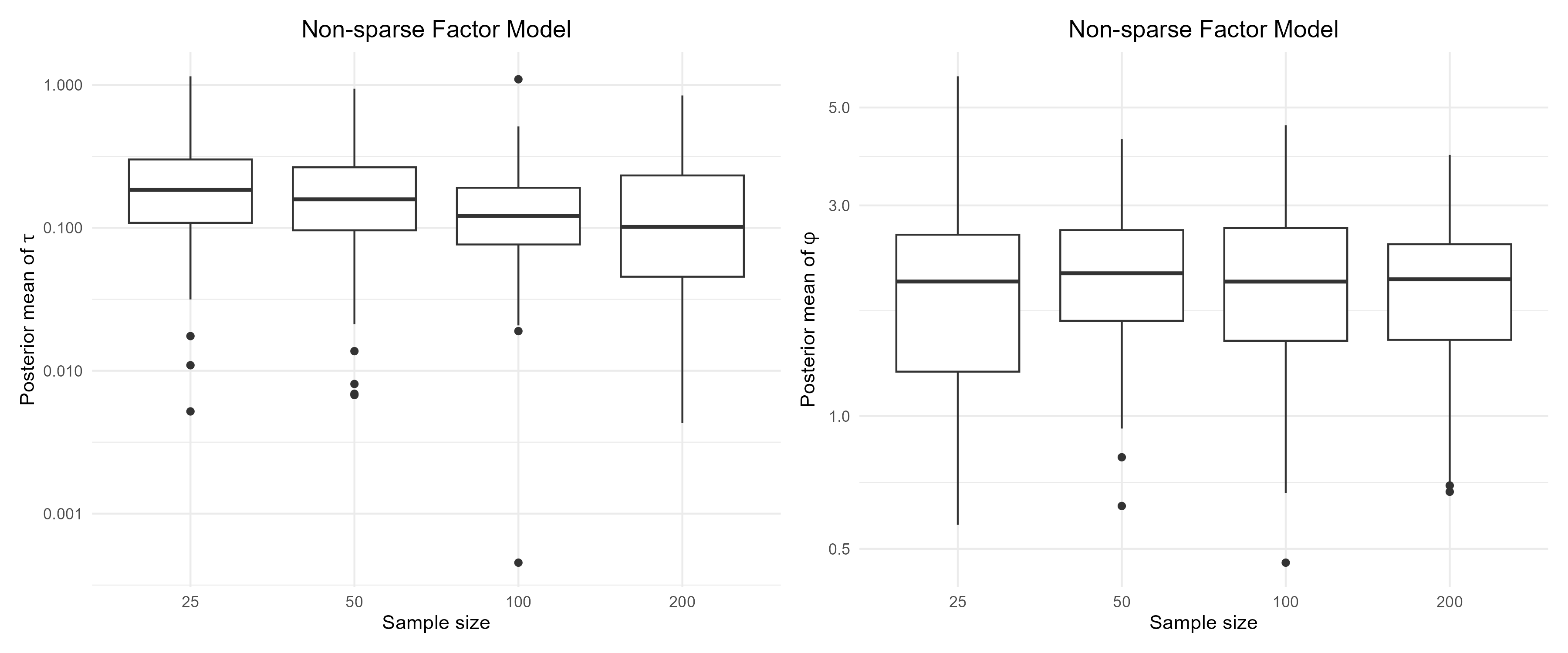}
  \caption{Distribution of the posterior mean of $\tau$ and $\phi$ across $100$ replicates with $N = 20, J = 5$.} 
    \label{fig:total-factor-20}
\end{figure}

\begin{figure}[!h]
    \centering
    \includegraphics[width=0.7\linewidth]{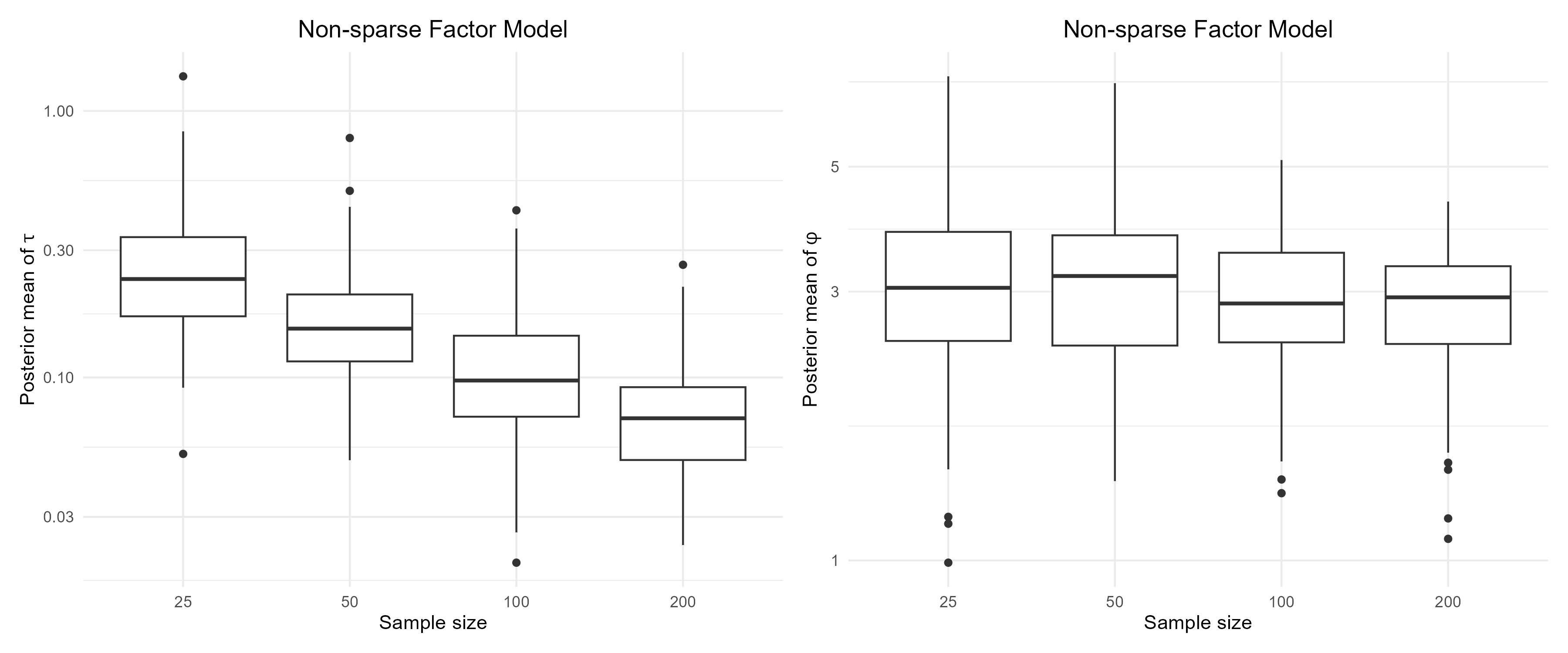}
    \caption{Distribution of the posterior mean of $\tau$ and $\phi$ across $100$ replicates  with $N = 50, J = 10$.}
    \label{fig:total-factor-50}
\end{figure}

\subsection{Sparse Factor Models}\label{sec:factor-sparse-setting}
Next, we consider a sparse model in which $X\in \bbR^{M \times N}$ is generated from a factor model with factor matrix $F \in \bbR^{4 \times M}$ and loading matrix $\Lambda \in \bbR^{N \times 4}$, and the untreated potential outcome for $Y \in \bbR^{M \times 1}$ is a linear combination of $X$, 
\begin{align}
    X &= F^\top \Lambda ^\top + e, \\
    Y &= Xw^* + \epsilon.
\end{align}
The treatment and sparse coefficients  $w^*$ are specified in the same way as in Section \ref{sec:sim-setting}. The sparsity still remains through $w^*$, but all units in the control group are linked by some factors that do not directly affect the treated unit. 
Analogously, we generate the following post-treatment data,
\begin{align}
    \Xnew &= \tilde{F}^\top \Lambda ^\top + \tilde{e}, \\
    \Ynew^{(1)} &= \Xnew w^* + \delta + \tilde{\epsilon},
\end{align}
where $\Ynew^{(1)} \in \bbR^{\tilde{M} \times 1}$, $\Xnew \in \bbR^{\tilde{M}  \times N}$, $\tilde{F} \in \bbR^{4 \times \tilde{M}}$. $\delta \in \bbR^{\tilde{M}}$ and $\tilde{\epsilon} \in \bbR^{\tilde{M}}$ are defined as in Section~\ref{sec:sim-setting}. The specifications of $(F,\tilde{F})$ and $\Lambda$ are the same as in Section~\ref{sec:factor-setting}, except that $\Lambda$ contains $N$ rows rather than $N+1$.

We report the relative efficiencies of the five estimators in Figure~\ref{plot:re_factor_sparse} with OLS$(\gamma^*)$ fixed at 100\%. The  QP estimators behave similarly to that in the main text, outperforming \BVS{} when the simplex constraint is satisfied ($\| w^* \|_1=1$), and performing poorly when $\| w^* \|_1=2$ or $3$. \BVS{} still outperforms Lasso when the ratio $N/M$ is large in all scenarios, though Lasso and OLS estimators approach the orcale OLS($\gamma^*$) slightly better when $M=200$. 
These results show that our method can still provide robust and superior performance for ATT estimation when the control units are correlated.

\begin{figure}[!h]
    \centering
    \includegraphics[width=1\linewidth]{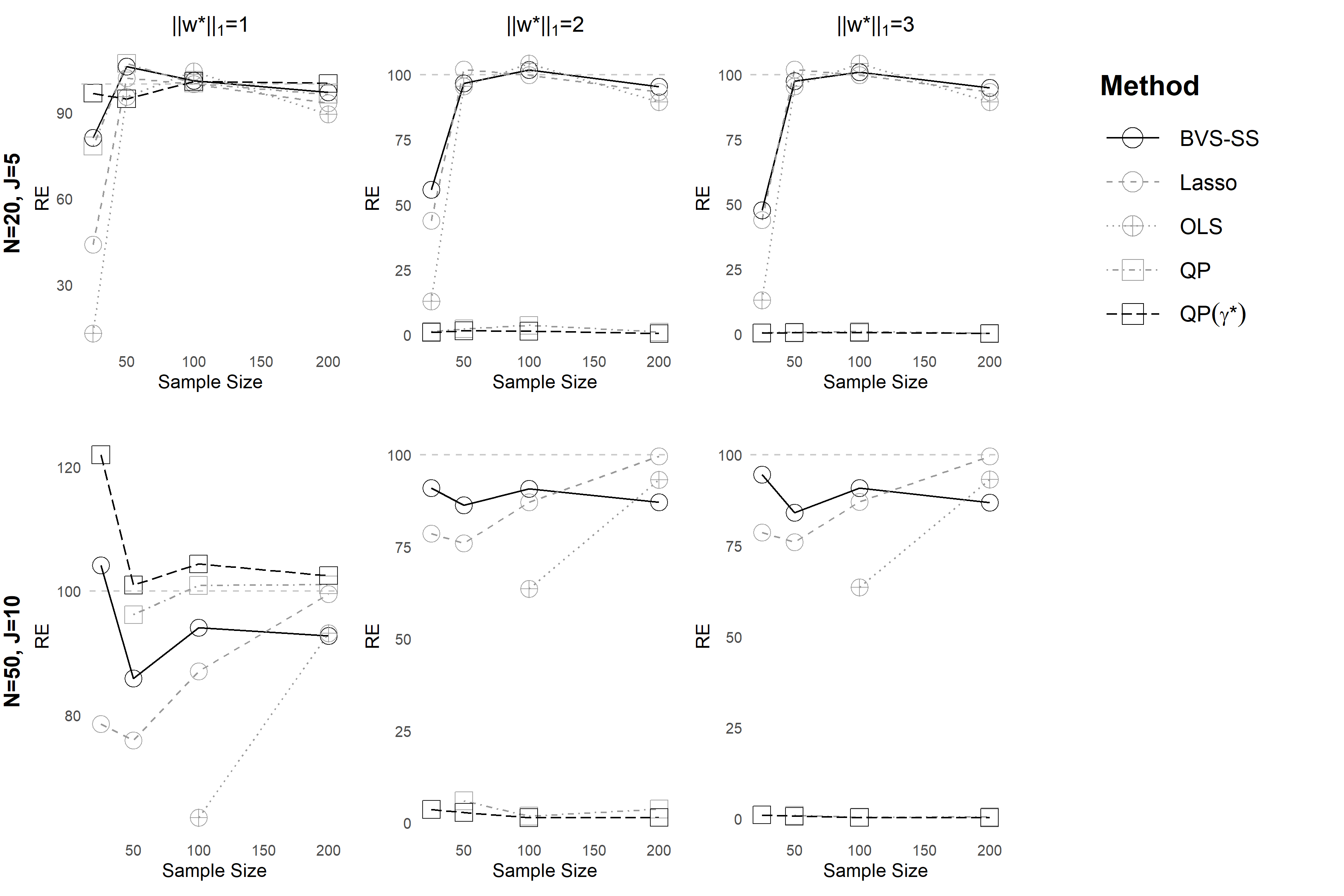}
    \caption{ Relative efficiency (\%) against sample size for ATT estimation for $X$ with correlations. First row: $N=20, J=5$;  second row: $N=50, J=10$; Top row: $\| w^* \|_1 = 1$;  middle row: $\| w^* \|_1 = 2$; bottom row: $\| w^* \|_1 = 3$. }
    \label{plot:re_factor_sparse}
\end{figure}

Regarding variable selection, \BVS{} generally attains higher accuracy than Lasso, particularly when the sample size $M$ is large. In smaller samples, however, its model selection performance deteriorates. For example, when $N=50, J=10, M\leq 50$, \BVS{} fails to select any true control unit on average, although Lasso performs nearly as poorly under these small-sample scenarios.  
This pattern is understandable: because our sampler is a Gibbs sampler that updates coefficients pairwise, strong correlations among covariates can lead to highly correlated variables being selected in place of the true signals. Importantly, despite the low accuracy of variable selection, \BVS{} delivers substantially better ATT estimation than Lasso when $M \leq 50$, as shown in Figure~\ref{plot:re_factor_sparse}, highlighting its robustness even when exact model recovery is difficult. This observation also underscores the need for more advanced MCMC algorithms which might improve selection accuracy in such correlated designs.

\begin{table}[!h] 
\centering
\begin{tabular}{ccccccc}
  \toprule
  & &   & \multicolumn{2}{c}{\BVS{}} & \multicolumn{2}{c}{Lasso}    \\ 
$\|w^*\|_1$  & Setting & $M$ & $\ell^1$-loss & model size &  $\ell^1$-loss & model size \\ 
  \midrule
  \multirow{8}{*}{\makecell{$\|w^*\|_1$ = 1}}  &
\multirow{4}{*}{\makecell{$N = 20$ \\ $J = 5$}}   
   & 25 & 5.6 & 7.2 & 5.7 & 8.9  \\ 
   &  & 50 &3.8 & 6.4 & 5.2 & 8.9   \\ 
   &  & 100 &2.2 & 5.8 & 5.5 & 9.9  \\ 
   &  & 200 & 1.2 & 5.4 & 5.1 & 9.8\\ 
\cmidrule{2-7}
   & \multirow{4}{*}{\makecell{$N = 50$ \\ $J = 10$}}        & 25 &19.5 & 19.1 & 12.7 & 10.9   \\ 
   &  & 50 &15.5 & 16.7 & 12.5 & 13.3   \\ 
   &  & 100 &10 & 13.3 & 13.9 & 20.3  \\ 
   &  & 200 & 6.1 & 11.2 & 14.1 & 21.2  \\ 
\midrule 
  \multirow{8}{*}{\makecell{$\|w^*\|_1$ = 3}}  &
\multirow{4}{*}{\makecell{$N = 20$ \\ $J = 5$}}   
   & 25 & 3.8 & 3.2 & 5.7 & 8.9  \\ 
   &  & 50 &2.4 & 3.8 & 5.2 & 8.9  \\ 
   &  & 100 &1.6 & 4.1 & 5.5 & 9.9 \\ 
   &  & 200 &1 & 4.4 & 5.1 & 9.8 \\ 
\cmidrule{2-7}
   & \multirow{4}{*}{\makecell{$N = 50$ \\ $J = 10$}}    & 25 &10.7 & 4.6 & 12.7 & 10.9   \\ 
   &  & 50 &8.4 & 5.1 & 12.5 & 13.3  \\ 
   &  & 100 & 6 & 6.1 & 13.9 & 20.3 \\ 
   &  & 200 &3.9 & 7 & 14.1 & 21.2  \\ 
\bottomrule
\end{tabular}
\caption{Variable selection performance of \BVS{} and Lasso averaged over 100 replicates.}  \label{}
\end{table}

Finally, the posterior means of $\tau$ and $\phi$ exhibit patterns very similar to those reported in the main text, and we visualize their distributions in Figures~\ref{fig:phi-50-factor-sparse} and~\ref{fig:tau-50-factor-sparse}.

\begin{figure}[!h]
    \centering
    \includegraphics[width=0.9\linewidth]{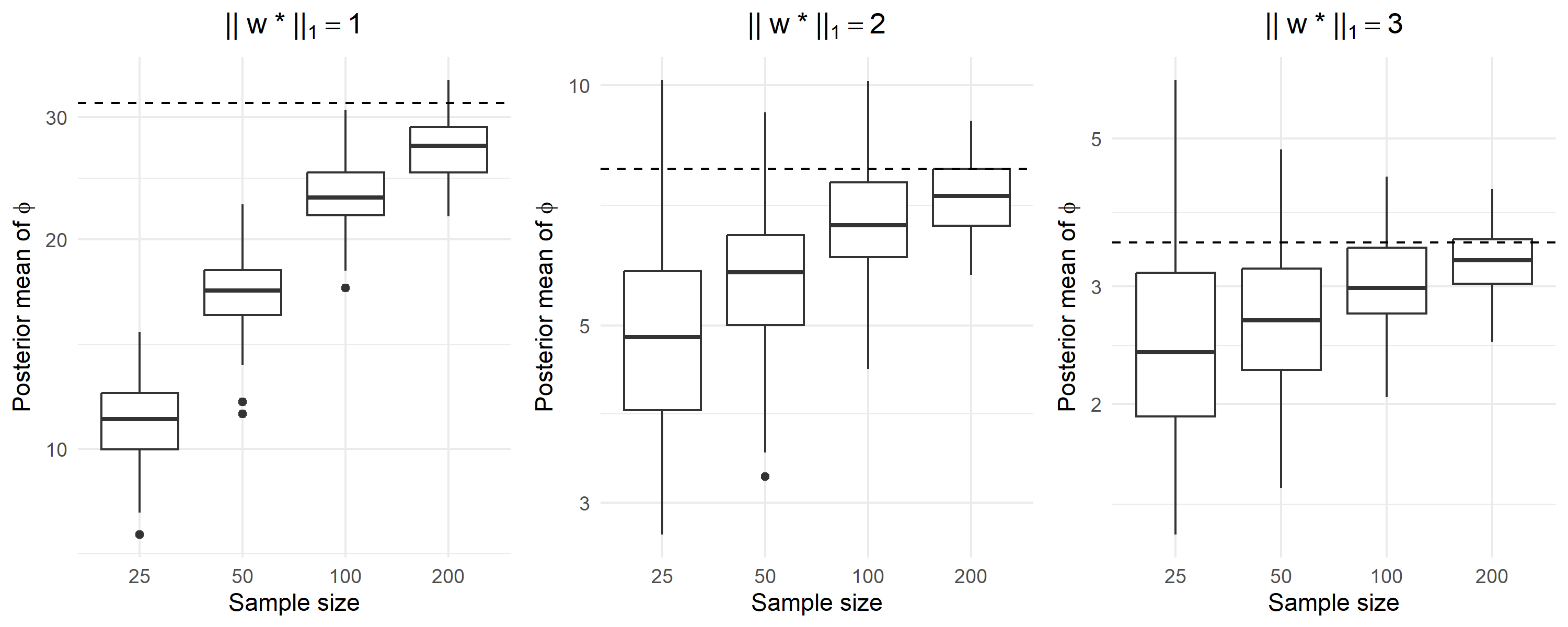}
  \caption{Distribution of the posterior mean of $\phi$ across $100$ replicates with $N = 50, J = 10$. 
  The true value $\phi^*$ is indicated by the dotted line.} 
    \label{fig:phi-50-factor-sparse}
\end{figure}

\begin{figure}[!h]
    \centering
    \includegraphics[width=0.9\linewidth]{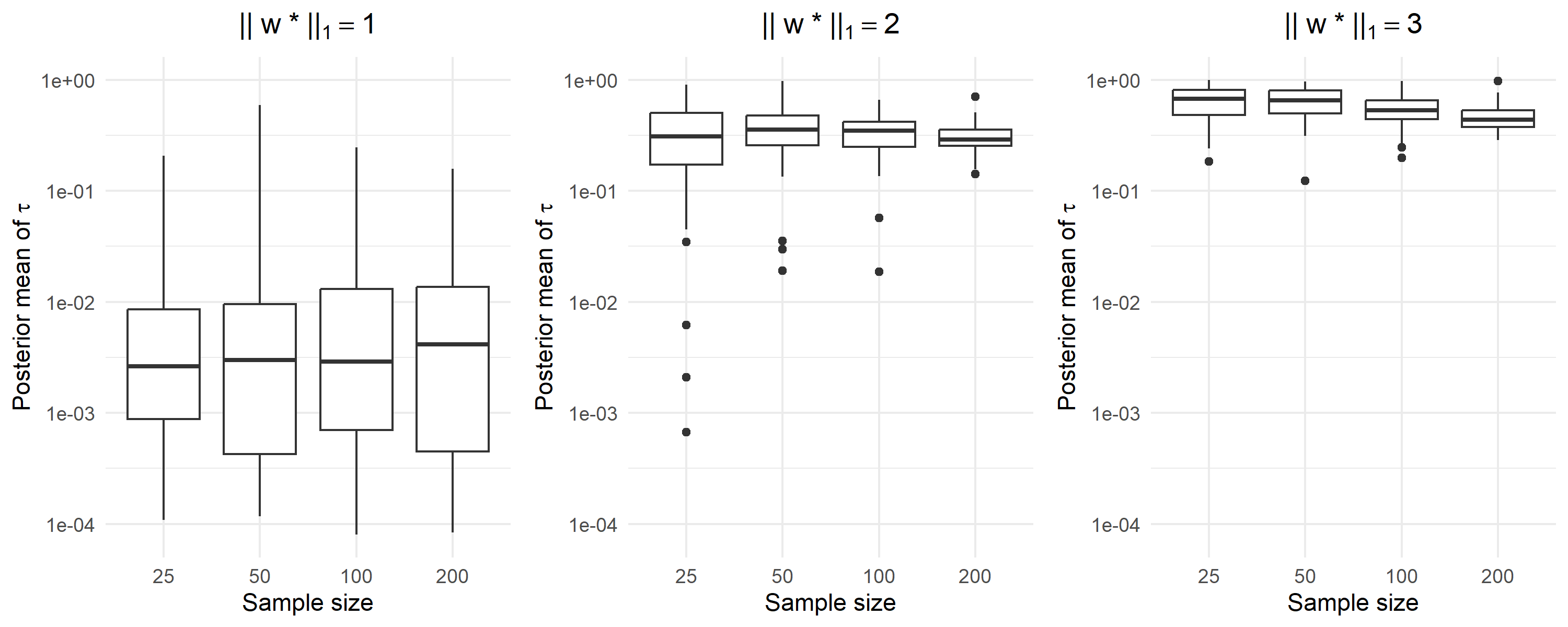}
    \caption{Distribution of the posterior mean of $\tau$ across $100$ replicates  with $N = 50, J = 10$.}
    \label{fig:tau-50-factor-sparse}
\end{figure}

\clearpage
\newpage

\section{Additional Results for Empirical Examples}\label{sec:ep-more}

This section presents supplementary details for the empirical examples analyzed in the main text.

\subsection{Additional Results for NFP Anti-tax Evasion}\label{sec:ep1-more}

We present more results for the NFP anti-tax evasion study in this section.
Figure \ref{EP1_trace} shows the trajectories of the posterior samples of $\phi$ and $\log(\tau)$ (with burn-in included). Probably due to the limited number of observations ($T_0=33$), the MCMC samples of $\tau$ exhibit large variability, indicating that the simplex assumption is hard to confirm or reject, though the ratio of the posterior means of $\tau$ and $\phi$ given in Table~\ref{EP1table} is relatively small. 
The trace plot suggests that our sampler has converged fast.

\begin{figure}[H]
    \centering
    \includegraphics[width=0.6\linewidth]{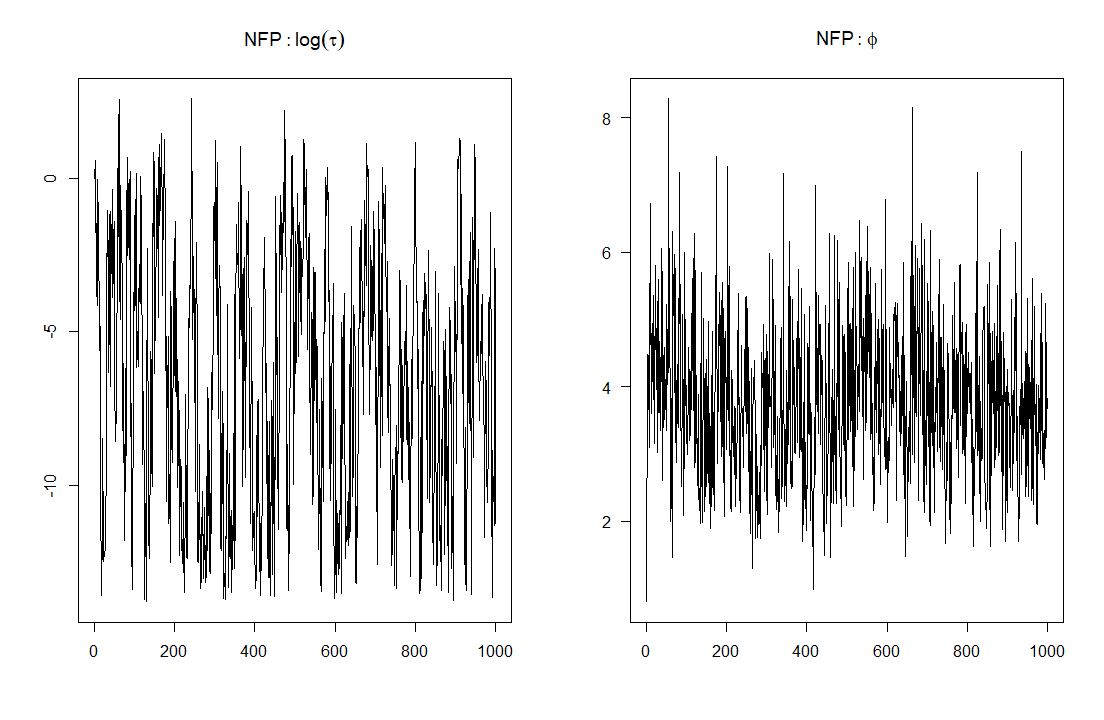}
    \caption{The trace plot of $\log(\tau)$ and $\phi$ for the NFP data set. 
    }
    \label{EP1_trace}
\end{figure}

The posterior distributions of ATT, $\log(\tau)$, $\phi$ and model size $|\gamma|$ are visualized in the histograms in Figure~\ref{EP1_hist}. The posterior mean of the ATT exhibits a clear unimodal and bell-shaped distribution, concentrated in the interval  $[0.26,0.34]$. The distribution of $\phi$ is symmetric with a mode between $[3.5,4.0]$, consistent with the posterior mean in Table~\ref{EP1table}, while the distribution of $\log(\tau)$ is right-skewed, indicating that a large portion of the posterior draws for $\tau$ are not concentrated near zero. This pattern is also reflected in the posterior mean of $\tau$ reported in Table~\ref{EP1table}. Taken together, the evidence suggests that the data in this empirical example do not favor the simplex constraint. The posterior distribution of $|\gamma|$ is left-skewed, with more than half of the posterior mass supporting a relatively small model that includes only two control units.

\begin{figure}[H]
    \centering
    \includegraphics[width=0.8\linewidth]{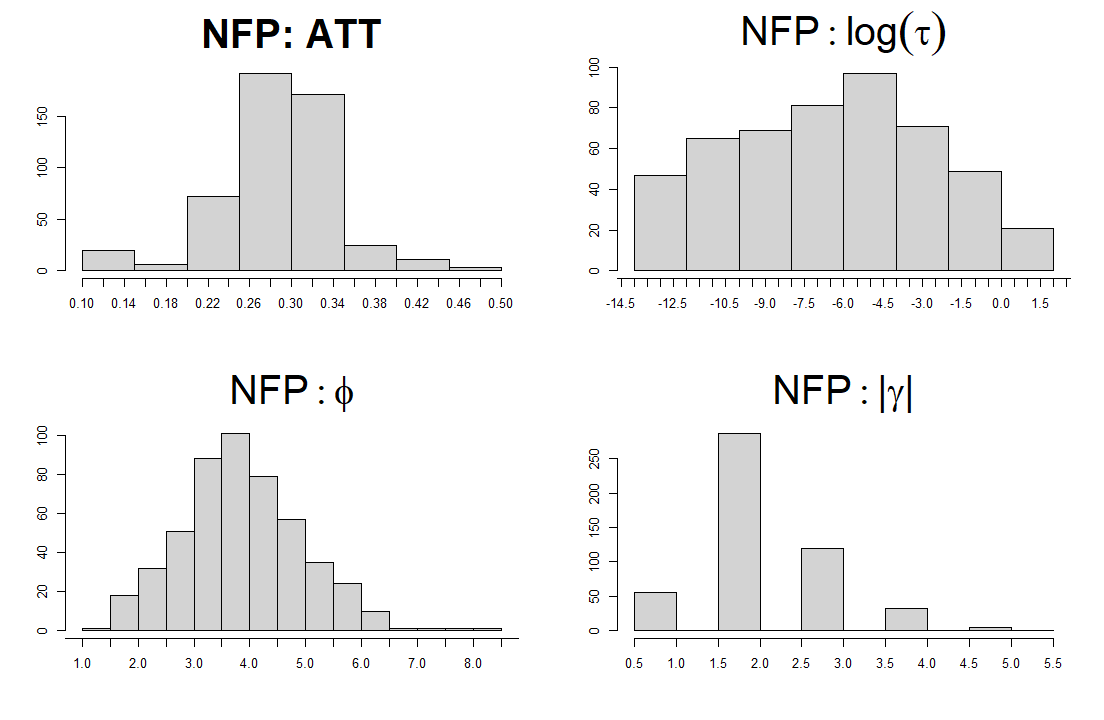}
    \caption{ Posterior distributions for the Nota Fiscal Paulista data set.}
    \label{EP1_hist}
\end{figure}

Finally, Figure~\ref{EP1_cfplot} plots the observed FHA inflation index together with the estimated counterfactual path.
Before the intervention (marked by the vertical green line), the estimated pre-treatment counterfactual (blue dashed line) tracks the observed FHA index well, indicating a good in-sample fit.
Following the intervention, the post-treatment counterfactual (red dashed line) diverges from the observed outcomes, capturing the estimated treatment effect over time.

\begin{figure}[H]
    \centering
    \includegraphics[width=0.8\linewidth]{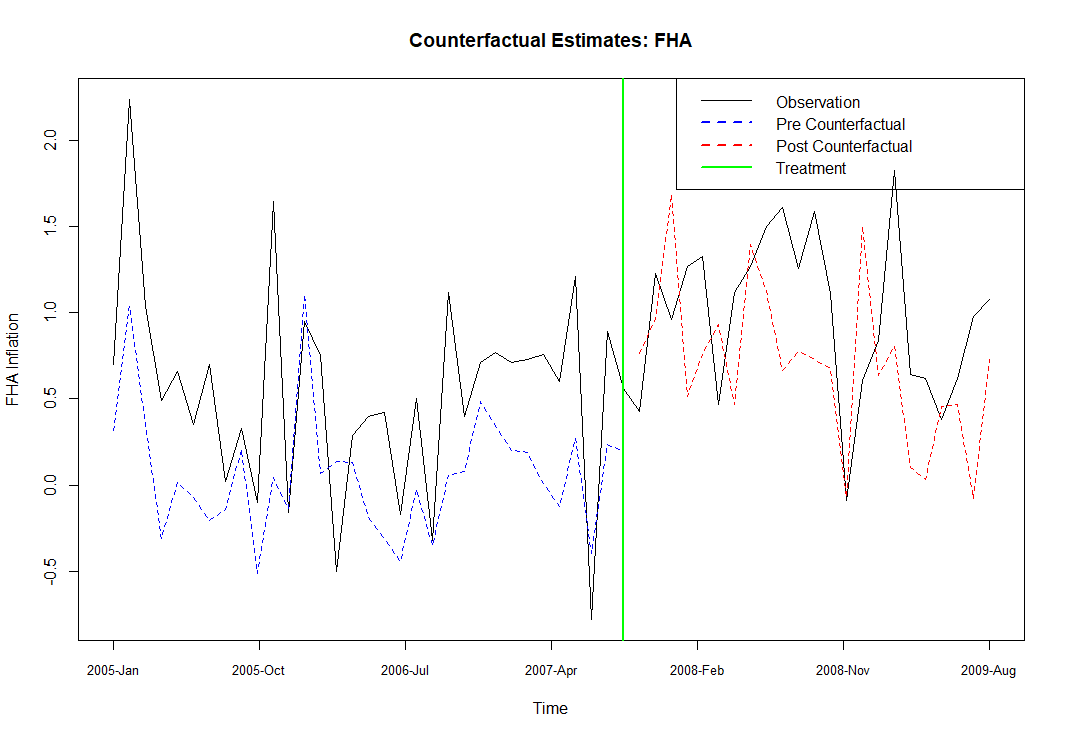}
    \caption{ Counterfactual estimation for the Nota Fiscal Paulista data set.}
    \label{EP1_cfplot}
\end{figure}

\subsection{Additional Results for China's Anti-corruption Campaign}\label{sec:ep2-more}

The section provides more details on the China's anti-corruption example. From Figure~\ref{EP2_trace}, one can observe that the posterior distribution of $\tau$ is still dispersed, but that of $\phi$ is more concentrated and much larger on average, compared to the previous example. The estimation of $\phi$ appears to be more accurate, suggesting that the error variance is around 0.05 (1/20.86). The
 ratio $\tau/\phi$ becomes much smaller, indicating that the simplex constraint is more likely to be
 satisfied (at least approximately) in this data set.

\begin{figure}[H]
    \centering
    \includegraphics[width=0.6\linewidth]{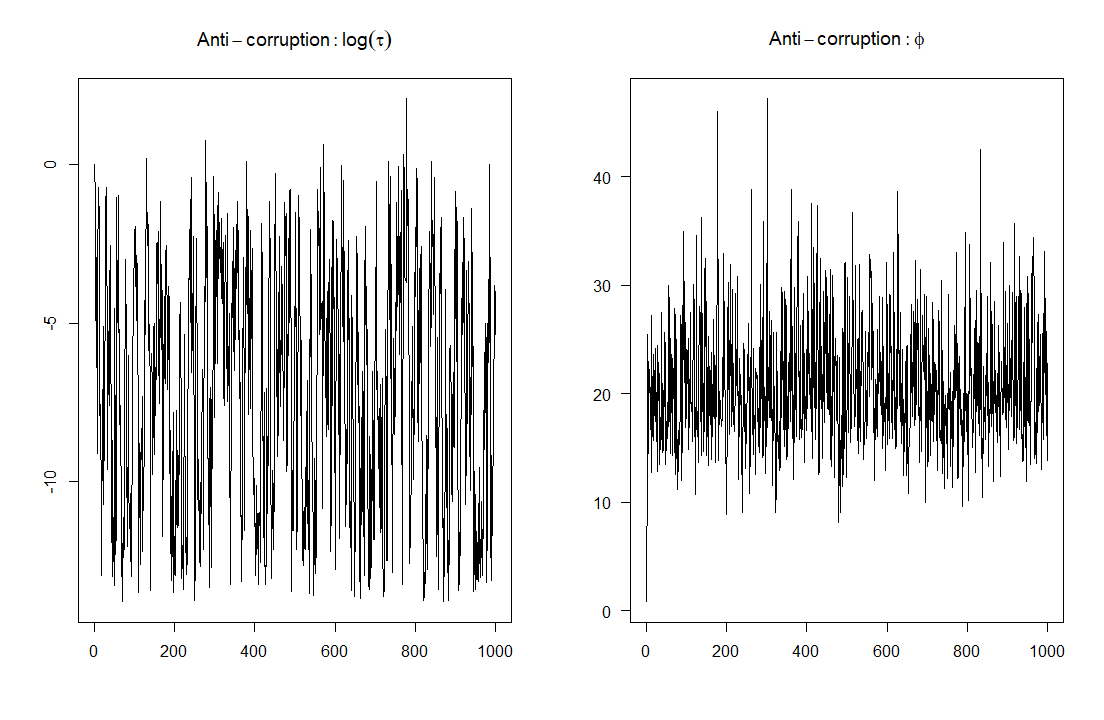}
    \caption{The trace plot of $\log(\tau)$ and $\phi$ for the anti-corruption data set.}
    \label{EP2_trace}
\end{figure}

Figure~\ref{EP2_hist} presents the posterior distributions. The ATT distribution is unimodal but a little right-skewed. The distribution of $\log(\tau)$ is in line with its trace plot, with a large probability mass for $\log(\tau) < -2$, but the distribution appears relatively diffuse since the bars in the histogram are fairly evenly distributed. More than 80\% samples of $\phi$ concentrate in the interval of $[15,30]$, providing evidence that our \BVS{} has detected abundant signals in the data.  \BVS{} still suggests a small model size for building the counterfactual, as around 70\% MCMC samples have model size smaller than six.

\begin{figure}[H]
    \centering
    \includegraphics[width=0.8\linewidth]{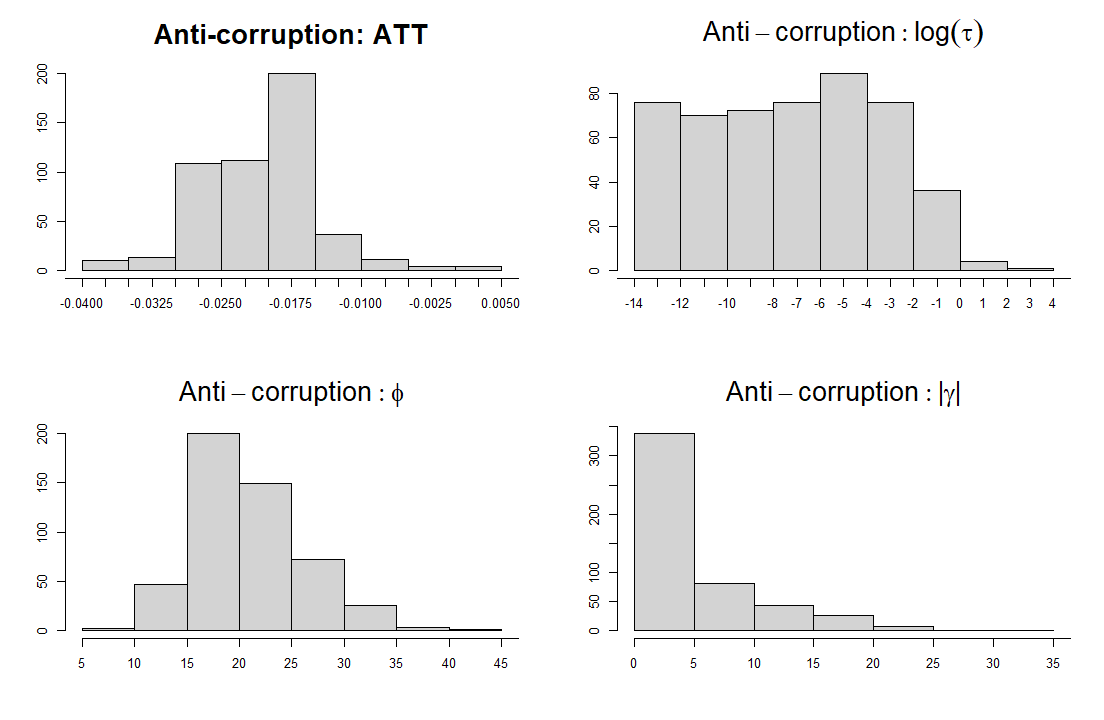}
    \caption{ Counterfactual Estimations for the China's Anti-corruption Campaign data set.}
    \label{EP2_hist}
\end{figure}

Finally, we provide a counterfactual plot for the growth rate of the luxury watches import in Figure~\ref{EP2_cfplot}. Compared with the previous example, the in-sample fit is noticeably better, likely due to the availability of a richer set of control units. A pronounced initial decline in the import growth rate is detected in the post-intervention periods, which implies a substantial immediate policy impact. However, the effect attenuates over time, with the gap between the observed and counterfactual paths gradually narrowing. Overall, the figure illustrates that our \BVS{} reproduces the pre-treatment dynamics effectively and provides a coherent counterfactual trajectory to assess the policy impact.

\begin{figure}[H]
    \centering
    \includegraphics[width=0.8\linewidth]{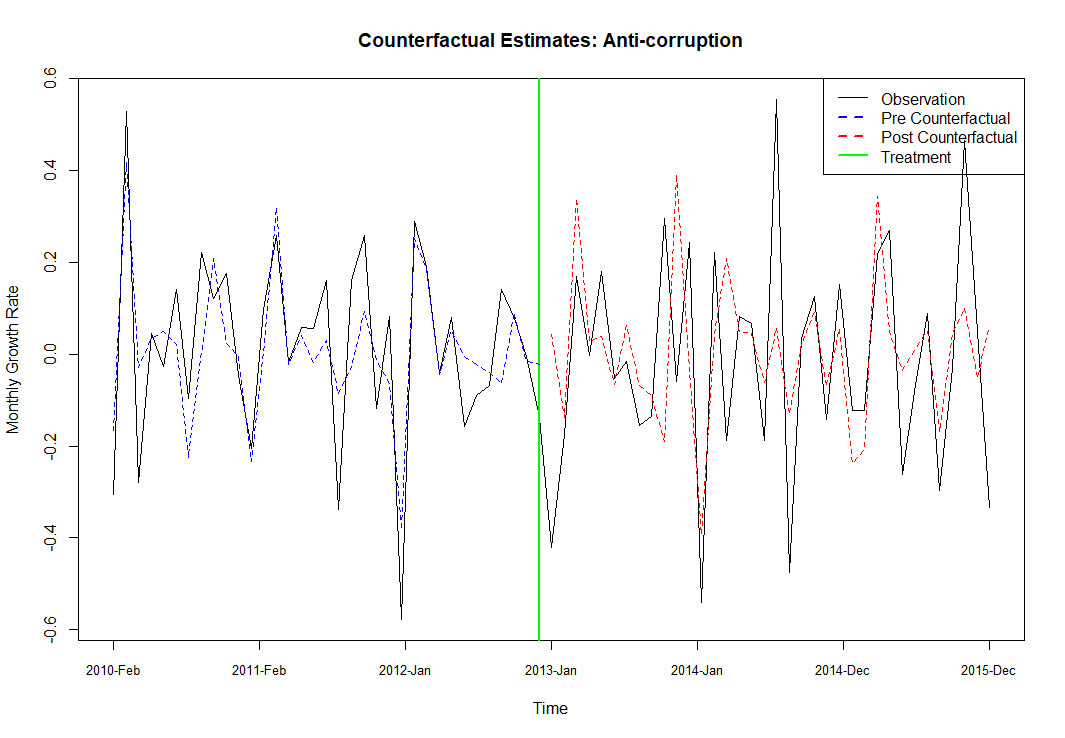}
    \caption{ Counterfactual estimation  for the China's Anti-corruption Campaign data set.}
    \label{EP2_cfplot}
\end{figure}

\bibliography{reff.bib}

\end{document}